\documentclass[11pt,a4paper]{article}
\usepackage{amssymb}
\usepackage{amsmath}
\usepackage{amsfonts}
\usepackage{bbm}
\usepackage{amsthm}
\usepackage{mathrsfs}
\usepackage{hyperref}
\usepackage{marvosym}
\usepackage{color}
\usepackage[margin=2.5cm]{geometry}
\usepackage[all,cmtip]{xy}
\usepackage[utf8]{inputenc}
\usepackage{graphicx}
\usepackage{varwidth}

\usepackage{upgreek}
\usepackage{rotating}

\usepackage{mathtools}

\makeatletter
\DeclareRobustCommand\widecheck[1]{{\mathpalette\@widecheck{#1}}}
\def\@widecheck#1#2{%
    \setbox\z@\hbox{\m@th$#1#2$}%
    \setbox\tw@\hbox{\m@th$#1%
       \widehat{%
          \vrule\@width\z@\@height\ht\z@
          \vrule\@height\z@\@width\wd\z@}$}%
    \dp\tw@-\ht\z@
    \@tempdima\ht\z@ \advance\@tempdima2\ht\tw@ \divide\@tempdima\thr@@
    \setbox\tw@\hbox{%
       \raise\@tempdima\hbox{\scalebox{1}[-1]{\lower\@tempdima\box
\tw@}}}%
    {\ooalign{\box\tw@ \cr \box\z@}}}
\makeatother

\definecolor{darkred}{rgb}{0.8,0.1,0.1}
\hypersetup{
     colorlinks=false,         
     linkcolor=darkred,
     citecolor=blue,
}

\theoremstyle{plain}
\newtheorem{theo}{Theorem}[section]
\newtheorem{lem}[theo]{Lemma}
\newtheorem{propo}[theo]{Proposition}
\newtheorem{cor}[theo]{Corollary}

\theoremstyle{definition}
\newtheorem{defi}[theo]{Definition}

\newtheorem{exx}[theo]{Example}
\newenvironment{ex}
  {\pushQED{\qed}\exx}
  {\popQED\endexx}

\newtheorem{remm}[theo]{Remark}
\newenvironment{rem}
  {\pushQED{\qed}\remm}
  {\popQED\endremm}

\numberwithin{equation}{section}

\def\nn{\nonumber}

\def\bbR{\mathbb{R}}

\def\bbZ{\mathbb{Z}}

\def\Hom{\mathrm{Hom}}

\def\Imm{\mathrm{Im}}

\def\PSh{\mathrm{PSh}}

\def\id{\mathrm{id}}

\def\dd{\mathrm{d}}

\def\dim{\mathrm{dim}}
\def\1{\mathbbm{1}}

\def\op{\mathrm{op}}

\def\pr{\mathrm{pr}}

\def\Loc{\mathsf{Loc}}
\def\Man{\mathsf{Man}}

\def\Set{\mathsf{Set}}

\def\CC{\mathsf{C}}

\def\HH{\mathsf{H}}

\def\Grpd{\mathsf{Grpd}}
\def\Grp{\mathsf{Group}}

\def\ad{\mathrm{ad}}

\def\BG{\mathrm{B}G}
\def\BGcon{\mathrm{B}G_\mathrm{con}}
\def\GCon{G\mathbf{Con}}
\def\bfOmega{\mathbf{\Omega}}
\def\GBun{G\mathbf{Bun}}
\def\GSol{G\mathbf{Sol}}
\def\GData{G\mathbf{Data}}
\def\g{\mathfrak{g}}
\def\ra{\triangleleft}
\def\ver{\mathrm{vert}}
\def\data{\mathrm{data}}

\def\AAA{\mathbf{A}}
\def\PPP{\mathbf{P}}
\def\bfomega{\boldsymbol{\omega}}
\def\EEE{\mathbf{E}}
\def\hhh{\mathbf{h}}

\def\YM{\mathbf{YM}}
\def\FF{\mathbf{F}}

\def\holim{\mathrm{holim}}

\def\hocolim{\mathrm{hocolim}}

\def\sk{\vspace{2mm}}

\title{%
The stack of Yang-Mills fields on Lorentzian manifolds
}

\author{%
Marco Benini$^{1,2,a}$, Alexander Schenkel$^{3,b}$ and Urs Schreiber$^{4,c}$\vspace{3mm}\\
{\small ${}^1$ Institut f\"ur Mathematik, Universit\"at Potsdam,}\\
{\small Karl-Liebknecht-Str.~24-25, 14476 Potsdam, Germany.}\vspace{2mm}\\
{\small ${}^2$ Fachbereich Mathematik, Universit\"at Hamburg,}\\
{\small Bundesstr.~55, 20146 Hamburg, Germany.}\vspace{2mm}\\
{\small ${}^3$ School of Mathematical Sciences, University of Nottingham,}\\
{\small University Park, Nottingham NG7 2RD, United Kingdom.}\vspace{2mm}\\
{\small ${}^4$ Mathematics Institute of the Academy,}\\ 
{\small \v{Z}itna 25, 115 67 Praha 1, Czech Republic.}\vspace{3mm}\\
{\footnotesize 
$^a$\texttt{mbenini87@gmail.com}, $^b$\texttt{alexander.schenkel@nottingham.ac.uk}, $^c$\texttt{urs.schreiber@gmail.com}}
 }

\date{March 2018}

\begin{document}

\maketitle

\begin{abstract}
\noindent We provide an abstract definition and an explicit construction of the stack of non-Abelian Yang-Mills fields on globally hyperbolic Lorentzian manifolds. We also formulate a stacky version of the Yang-Mills Cauchy problem and show that its well-posedness is equivalent to a whole family of parametrized PDE problems. Our work is based on the homotopy theoretical approach to stacks proposed in [S.~Hollander, Israel J.\ Math.\ {\bf 163}, 93-124 (2008)], which we shall extend by further constructions that are relevant for our purposes. In particular, we will clarify the concretification of mapping stacks to classifying stacks such as $\BGcon$. 
\end{abstract}


\tableofcontents




\section{\label{sec:intro}Introduction and summary}
Understanding quantum Yang-Mills theory is one of the most important 
and challenging open problems in mathematical physics. While approaching 
this problem in a fully non-perturbative fashion seems to be out of reach
within the near future, there are recent developments in classical and 
quantum field theory which make it plausible that quantum Yang-Mills theory
could relatively soon be constructed within a perturbative approach
that treats the coupling constant non-perturbatively, but Planck's constant $\hbar$ 
as a formal deformation parameter. See \cite{Cahiers,syntvar,Collini} and the next paragraph for further details. 
The advantage of such an approach 
compared to standard perturbative (algebraic) quantum field theory, 
see e.g.\ \cite{Kasia} for a recent monograph, is that by treating 
the coupling constant non-perturbatively the resulting quantum field theory
is sensitive to the global geometry of the field configuration spaces.
This is particularly interesting and rich in gauge theories, where the global geometry of
the space of gauge fields encodes various topological features such as
characteristic classes of the underlying principal bundles, holonomy groups 
and other topological charges.
\sk

Loosely speaking, the construction of a quantum field theory within this non-standard 
perturbative approach consists of the following three steps:
1.)~Understand the smooth structure and global geometry of
the space of solutions of the field equation of interest.
2.)~Use the approach of \cite{Zuckerman} to equip the solution space 
with a symplectic form and construct a corresponding Poisson algebra of smooth functions on it.
(The latter should be interpreted as the Poisson algebra of classical observables of our field theory.)
3.)~Employ suitable techniques from formal deformation quantization 
to quantize this Poisson algebra.
Even though stating these three steps in a loose language is very simple,
performing them rigorously is quite technical and challenging.
The reason behind this is that the configuration and solution spaces 
of field theories are typically infinite-dimensional, hence the standard techniques of 
differential geometry do not apply. In our opinion, the most
elegant and powerful method to study the smooth spaces appearing in field theories
is offered by sheaf topos techniques: In this approach a smooth space
$X$ is defined by coherently specifying {\em all} smooth maps $U \to X$ from
{\em all} finite-dimensional manifolds $U$ to $X$. More precisely,
this means that $X$ is a (pre)sheaf on a suitable site $\CC$,
which we may choose as the site of all finite-dimensional manifolds $U$
that are diffeomorphic to some $\bbR^n$, $n\in\bbZ_{\geq 0}$.
In particular, the maps from $U=\bbR^0$ determine the points of $X$ and the maps from
$U = \bbR^1$ the smooth curves in $X$. Employing such techniques, the first two steps of the program
sketched above have been successfully solved in \cite{Cahiers} and \cite{syntvar}
for the case of non-linear field theories without gauge symmetry. Concerning step 3.),
Collini \cite{Collini} obtained very interesting and promising results showing that Fedosov's construction
of a $\star$-product applies to the Poisson algebra of $\Phi^4$-theory in $4$ spacetime dimensions,
leading to a formal deformation quantization of this theory that is non-perturbative 
in the coupling constant. Collini's approach is different from ours as it
describes the infinite-dimensional spaces of fields and solutions by 
locally-convex manifolds, which are less flexible.
It would be an interesting problem to reformulate and generalize 
his results in our more elegant and powerful sheaf theoretic approach \cite{Cahiers,syntvar}.
\sk

The goal of this paper is to address and solve step 1.)~of the program sketched above for the case
of Yang-Mills theory with a possibly non-Abelian structure group $G$ on globally hyperbolic Lorentzian manifolds.
A crucial observation is that, due to the presence of gauge symmetries, the sheaf topos approach 
of \cite{Cahiers} and \cite{syntvar} is no longer sufficient and has to be generalized to
``higher sheaves'' ({\em stacks}) and the ``higher toposes'' they form.
We refer to \cite{Schreiber} for an overview of recent developments at the interface of higher topos theory
and mathematical physics, and also to \cite{Eggertsson} for a gentle introduction to the role of stacks in gauge theory.
See also \cite{StackyCS} for a formulation of Chern-Simons theory in this framework.
The basic idea behind this is as follows: The collection of all gauge fields on a manifold $M$
naturally forms a groupoid and not a set. The objects $(A,P)$ of this groupoid are principal $G$-bundles
$P$ over $M$ equipped with a connection $A$ and the morphisms 
$h : (A,P) \to (A^\prime,P^\prime)$
are gauge transformations, i.e.\ bundle isomorphisms preserving the connections.
There are two important points we would like to emphasize: 
i)~This groupoid picture is essential to capture all the topological
charge sectors of the gauge field (i.e.\ non-isomorphic principal bundles).
In particular, it is intrinsically non-perturbative as one does not have to fix 
a particular topological charge sector to perturb around.
ii)~Taking the quotient of the groupoid of gauge fields, i.e.\ forming ``gauge orbits'',
one loses crucial information that is encoded in the automorphism groups of objects of 
the groupoid, i.e.\ the stabilizer groups of bundles with connections. 
This eventually would destroy the important descent properties (i.e.\ gluing connections 
up to a gauge transformation) that are enjoyed by the groupoids of gauge fields, but 
not by their corresponding sets of ``gauge orbits''. 
It turns out that gauge fields on a manifold $M$ do not only form a groupoid but even a smooth groupoid. 
The latter may be described by groupoid-valued presheaves on our site $\CC$,
i.e.\ objects of the category $\HH := \PSh(\CC,\Grpd)$.
Following the seminal work by K.~Brown, J.~F.~Jardine and others, 
in a series of papers \cite{Hollander, Hollander2, Hollander3} Hollander 
developed the abstract theory of presheaves of groupoids
by using techniques from model category theory/homotopical algebra, 
see e.g.\ \cite{Dwyer} for a concise introduction.
One of her main insights was that the usual theory of {\em stacks} \cite{DM,Giraud} 
can be formalized very elegantly and efficiently in this framework by employing homotopical 
techniques. In short, stacks can be identified as those presheaves of groupoids satisfying a
notion of descent, which can be phrased in purely model categorical terms. 
With these techniques and developments in mind, we now can state more precisely the two main 
problems we address in this paper: 
{\bf (I)}~Understand and describe the stack of Yang-Mills fields 
in the framework of \cite{Hollander, Hollander2, Hollander3}.
{\bf (II)}~Understand what it means for the stacky version of the Yang-Mills Cauchy problem to be well-posed.
\sk

The present paper is part of a longer term research program of two of us (M.B.\ and A.S.)
on {\em homotopical algebraic quantum field theory}. The aim of this program is to develop
a novel and powerful framework for quantum field theory on Lorentzian manifolds
that combines ideas from locally covariant quantum field theory \cite{Brunetti}
with homotopical algebra \cite{Hovey, Dwyer}. This is essential to capture structural properties of quantum 
gauge theories that are lost at the level of gauge invariant observables. 
In previous works, we could confirm for toy-models 
that our homotopical framework is suitable to perform local-to-global 
constructions for gauge field observables \cite{hocolim} and
we constructed a class of toy-examples of homotopical algebraic quantum field
theories describing a combination of classical gauge fields and quantized matter fields \cite{horan}.
Based on the results of the present paper, we will be able to address steps 2.)\ and 3.)\
of the program outlined above for gauge theories and in particular for Yang-Mills theory.
We expect that this will allow us to obtain first examples of homotopical algebraic quantum field theories
which describe quantized gauge fields.
\sk

The outline of the remainder of this paper is as follows:
In Section \ref{sec:prelim} we provide a rather self-contained
introduction to presheaves of groupoids and their model category structures.
This should allow readers without much experience with this subject
to understand our statements and constructions.
In particular, we review the main results of \cite{Hollander}
which show that there are two model structures on the category $\HH = \PSh(\CC,\Grpd)$, 
called the global and the local model structure.
The local model structure, which is obtained by localizing the global one,
is crucial for detecting stacks in a purely model categorical fashion 
as those are precisely the {\em fibrant} objects for this model structure. 
We then provide many examples of stacks that are important in gauge theory,
including the stack represented by a manifold and some relevant classifying stacks, e.g.\
$\BGcon$ which classifies principal $G$-bundles with connections.
This section is concluded by explaining homotopy fiber products for stacks 
and (derived) mapping stacks, which are homotopically meaningful constructions on 
stacks that will be frequently used in our work.
\sk

In Section \ref{sec:gaugefields} we construct and explicitly describe the stack 
of gauge fields $\GCon(M)$ on a manifold $M$. Our main guiding principle 
is the following expectation of how the stack $\GCon(M)$ is supposed to look: 
Via the functor of points perspective, the groupoid $\GCon(M)(U)$ obtained by evaluating 
the stack $\GCon(M)$ on an object $U$ in $\CC$ should be interpreted as 
the groupoid of smooth maps $U \to \GCon(M)$. Because $\GCon(M)$ is supposed to
describe principal $G$-bundles with connections, any such smooth map will describe a smoothly 
$U$-parametrized family of principal $G$-bundles with connections on $M$, and 
the corresponding morphisms are smoothly $U$-parametrized gauge transformations. 
We shall obtain a precise and intrinsic definition of $\GCon(M)$ by concretifying 
the mapping stack from (a suitable cofibrant replacement of) the manifold $M$ 
to the classifying stack $\BGcon$. Our concretification prescription (cf.\ Definition \ref{defi:diffconcret}) improves 
the one originally proposed in \cite{Fiorenza,Schreiber}. In fact, as we explain in
more detail in Appendix \ref{app:concretification}, the original concretification does 
not produce the desired result (sketched above) for the stack of gauge fields, 
while our improved concretification fixes this issue.
We then show that assigning to connections their curvatures 
may be understood as a natural morphism $\FF_M : \GCon(M)\to \bfOmega^2(M,\ad(G))$
between concretified mapping stacks. 
\sk

Section \ref{sec:YM} is devoted to formalizing the Yang-Mills equation on globally 
hyperbolic Lorentzian manifolds $M$ in our model categorical framework. 
After providing a brief review of some basic aspects of
globally hyperbolic Lorentzian manifolds, we shall show that, similarly to the curvature, 
the Yang-Mills operator may be formalized as a natural morphism
$\YM_M : \GCon(M) \to \bfOmega^1(M,\ad(G))$ between concretified mapping stacks.
This then allows us to define abstractly the stack of solutions to the Yang-Mills equation $\GSol(M)$
by a suitable homotopy fiber product of stacks (cf.\ Definition \ref{defi:solstack}).
We shall explicitly work out the Yang-Mills stack $\GSol(M)$ and give 
a simple presentation up to weak equivalence in $\HH$.
This solves our problem {\bf (I)} listed above.
\sk

Problem {\bf (II)} is then addressed in Section \ref{sec:Cauchy}.
After introducing the stack of initial data $\GData(\Sigma)$ 
on a Cauchy surface $\Sigma\subseteq M$ and the morphism of stacks 
 $\data_\Sigma : \GSol(M) \to \GData(\Sigma)$ that assigns to
solutions of the Yang-Mills equation their initial data,
we formalize a notion of well-posedness for the stacky Yang-Mills Cauchy problem 
in the language of model categories (cf.\ Definition \ref{def:stackycauchypbl}).
Explicitly, we say that the stacky Cauchy problem is well-posed
if $\data_\Sigma $ is a weak equivalence in the local model structure on $\HH$.
We will then unravel this abstract condition and obtain that well-posedness of
the stacky Cauchy problem is equivalent to well-posedness 
of a whole family of smoothly $U$-parametrized Cauchy problems (cf.\ Proposition \ref{propo:cauchy}).
A particular member of this family 
(given by the trivial parametrization by a point $U=\bbR^0$) is the ordinary Yang-Mills Cauchy problem, 
which is known to be well-posed in dimension $m=2,3,4$ \cite{CSYM,CBYM}. 
To the best of our knowledge, smoothly $U$-parametrized Cauchy problems
of the kind we obtain in this work have not been studied in the PDE theory literature yet.
Because of their crucial role in understanding Yang-Mills theory, 
we believe that such problems deserve the attention of the PDE community. 
It is also interesting to notice that Yang-Mills theory, which is of primary interest to physics, 
provides a natural bridge connecting two seemingly distant branches of pure mathematics, 
namely homotopical algebra and PDE theory. 
\sk

In Section \ref{sec:Lorenz} we make the interesting observation that 
gauge fixings may be understood in our framework as weakly equivalent descriptions 
of the same stack. For the sake of simplifying our arguments,
we focus on the particular example of Lorenz gauge fixing, which is 
often used in applications to turn the Yang-Mills equation into a hyperbolic PDE.
We define a stack $\GSol_{\mathrm{g.f.}}(M)$ of gauge-fixed Yang-Mills fields,
which comes together with a morphism $\GSol_{\mathrm{g.f.}}(M) \to \GSol(M)$
to the stack of all Yang-Mills fields. Provided certain smoothly $U$-parametrized PDE 
problems admit a solution (cf.\ Proposition \ref{propo:fixing}), this morphism is a weak equivalence in $\HH$,
which means that the gauge-fixed Yang-Mills stack $\GSol_{\mathrm{g.f.}}(M)$ is
an equivalent description of $\GSol(M)$. 
\sk

The paper contains four appendices. In the first three appendices
we work out some relevant aspects of the model category $\HH$ of presheaves of groupoids
that were not discussed by Hollander in her series a papers \cite{Hollander, Hollander2, Hollander3}.
In Appendix \ref{app:monoidal} we show that $\HH$ is
a monoidal model category, which is essential for our construction of mapping stacks. In Appendix \ref{app:cofibrep} we
obtain functorial cofibrant replacements
of manifolds in $\HH$, which are needed for computing our mapping stacks explicitly.
In Appendix \ref{app:truncation} we develop very explicit techniques to compute fibrant
replacements in the $(-1)$-truncation (cf.\ \cite{Barwick,Rezk,ToenVezzosi} and also \cite{Lurie})
of the canonical model structure on over-categories $\HH/K$,
which are crucial for the concretification of our mapping stacks.
The last Appendix \ref{app:concretification} compares our concretification prescription
with the original one proposed in \cite{Fiorenza,Schreiber}. In particular, we show that the latter
does not lead to the desired stack of gauge fields, i.e.\ the stack describing smoothly parametrized
principal $G$-bundles with connections, which was our motivation to develop and propose
an improved concretification prescription in Definition \ref{defi:diffconcret}.


\section{\label{sec:prelim}Preliminaries}
We fix our notations and review some aspects of the theory of presheaves
of groupoids which are needed for our work. Our main reference for this section
is \cite{Hollander} and references therein. A good introduction
to model categories is \cite{Dwyer}, see also \cite{Hovey,Hirschhorn} for more details.
We shall use presheaves of groupoids as a model for stacks which are, loosely speaking,
generalized smooth spaces whose `points' may have non-trivial automorphisms.

\subsection{\label{subsec:Grpd}Groupoids}
Recall that a groupoid $G$ is a (small) category in which every morphism is an isomorphism.
For $G$ a groupoid, we denote its objects by symbols like $x$
and its morphisms by symbols like $g : x\to x^\prime$. A morphism $F : G\to H$ between 
two groupoids $G$ and $H$ is a functor between their underlying categories. 
We denote the category of groupoids by $\Grpd$.
\sk

The category $\Grpd$ is closed symmetric monoidal: The product $G\times H$ of two groupoids
is the groupoid whose object (morphism) set is the product of the object (morphism) sets
of $G$ and $H$. The monoidal unit is the groupoid $\{\ast\}$ with only one object and its identity morphism.
The functor $G\times (-) : \Grpd \to \Grpd$ has a right adjoint functor, which we denote by
$\Grpd(G,-) : \Grpd \to \Grpd$. We call $\Grpd(G,H)$ the  internal hom-groupoid
from $G$ to $H$. Explicitly, the objects of $\Grpd(G,H)$ are all
functors $F : G\to H$ and the morphisms from $F: G\to H$ to $F^\prime : G\to H$
are all natural transformations $\eta : F\to F^\prime$. It is easy to see that
a morphism in $\Grpd(G,H)$ may be equivalently described by a functor
\begin{flalign}\label{eqn:homgrpd}
U : G\times \Delta^1\longrightarrow H~,
\end{flalign}
where $\Delta^1$ is the groupoid with two objects, say $0$ and $1$, and a unique isomorphism $0\to 1$ 
between them. (The source and target of the morphism $U$ is obtained by restricting
$U$ to the objects $0$ and $1$ in $\Delta^1$, and their identity morphisms, and the natural transformation $\eta$ 
is obtained by evaluating $U$ on the morphism $0\to 1$ in $\Delta^1$.)
This latter perspective on morphisms in $\Grpd(G,H)$ will be useful later when we discuss presheaves of groupoids.
\sk

The category $\Grpd$ can be equipped with a model structure, i.e.\ one can do
homotopy theory with groupoids. For a proof of the theorem below we refer to \cite[Section~6]{Strickland}.
 Recall that a {\em model category}
is a complete and cocomplete category with three distinguished classes of morphisms
-- called {\em fibrations}, {\em cofibrations} and {\em weak equivalences} -- that have to satisfy a list of axioms,
see e.g.\ \cite{Dwyer} or \cite{Hovey,Hirschhorn}.
\begin{theo}\label{theo:groupoids}
Define a morphism $F : G\to H$ in $\Grpd$ to be
\begin{itemize}
\item[(i)] a weak equivalence if it is fully faithful and essentially surjective;

\item[(ii)] a fibration if for each object $x$ in $G$ and each morphism 
$h : F(x)\to y$ in $H$ there exists a morphism $g : x \to x^\prime$
in $G$ such that $F(g)=h$; and

\item[(iii)] a cofibration if it is injective on objects.
\end{itemize}
With these choices $\Grpd$ is a model category.
\end{theo}

In the following, we will need a simple and tractable 
model for the {\em homotopy limit} $\holim_{\Grpd}^{}$
of a cosimplicial diagram in $\Grpd$. See e.g.\ \cite{Dwyer}
for a brief introduction to homotopy limits and colimits.
Such a model was found by Hollander in \cite{Hollander} in terms of the descent category.
\begin{propo}\label{propo:descentcat}
Let $G^\bullet$ be any cosimplicial diagram in $\Grpd$, i.e.\
\begin{flalign}
\xymatrix@C=1.1pc{
G^\bullet := \bigg(
G^0  \ar@<0.5ex>[r]^-{\dd^0}\ar@<-0.5ex>[r]_-{\dd^1}~&~
G^1 \ar@<1ex>[r]\ar@<-1ex>[r]\ar[r]~&~ G^2
\ar@<0.5ex>[r]\ar@<-0.5ex>[r]\ar@<1.5ex>[r]\ar@<-1.5ex>[r]~&~    \cdots\bigg)
}~,
\end{flalign}
where $G^n$ are groupoids, for all $n \in\bbZ_{\geq 0}$, and
we suppressed as usual the codegeneracy maps $\mathrm{s}^i$ in this cosimplicial diagram.
Then the homotopy limit $\holim_{\Grpd}^{} G^\bullet$ is the groupoid whose
\begin{itemize}
\item objects are pairs $(x , h )$, where $x$ is an object in $G^0$ 
and $h :  \dd^1(x) \to \dd^0(x) $ is a morphism in $G^1$, such that
$\mathrm{s}^0(h) = \id_{x}$ and $\dd^0(h)\circ \dd^2(h) = \dd^1(h)$ in $G^2$; and

\item morphisms $ g : (x,h)\to (x^\prime,h^\prime)$ are morphisms $g : x\to x^\prime$
in $G^0$, such that the diagram
\begin{flalign}
\xymatrix{
\ar[d]_-{h}\dd^1(x) \ar[rr]^-{\dd^1(g)}&& \dd^1(x^\prime)\ar[d]^-{h^\prime}\\
\dd^0(x) \ar[rr]_-{\dd^0(g)}&&\dd^0(x^\prime)
}
\end{flalign}
in $G^1$ commutes. 
\end{itemize}
\end{propo}

\subsection{\label{subsec:H}Presheaves of groupoids and stacks}
Let $\CC$ be the category with objects given by all 
(finite-dimensional and paracompact) manifolds $U$ that are 
diffeomorphic to a Cartesian space $\bbR^n$, $n\in \bbZ_{\geq 0}$, 
and morphisms given by all smooth maps $\ell : U\to U^\prime$.
Notice that $U$ and $U^\prime$ are allowed to have different dimensions.
We equip $\CC$ with the structure of a site by declaring
a family of morphisms
\begin{flalign}
\big\{\ell_i : U_i \longrightarrow U\big\}
\end{flalign}
in $\CC$ to be a covering family whenever $\{U_i\subseteq U\}$ is a good open cover of $U$
and $\ell_i : U_i \to U$ are the canonical inclusions. As usual, we denote
intersections of the $U_i$ by $U_{i_1\dots i_n} := U_{i_1}\cap \cdots \cap U_{i_n}$.
By definition of good open cover, these intersections are either empty 
or open subsets diffeomorphic to some $\bbR^n$, i.e.\ objects in $\CC$.
We note that our site $\CC$ is equivalent to the site $\mathsf{Cart}$ of Cartesian spaces 
used in \cite{Fiorenza,Schreiber}. We however prefer to work with $\CC$ instead of
$\mathsf{Cart}$, because covering families $\{\ell_i : U_i = \bbR^m \to U=\bbR^m\}$
in $\mathsf{Cart}$ are less intuitive as the $\ell_i$ are in general {\em not} subset inclusions.
In practice, see e.g.\ the examples in Section \ref{subsec:examples}, 
this would complicate the notations for cocycle conditions,
where, instead of the familiar restrictions $\vert_{U_i}$ of functions or forms 
to $U_i\subseteq U$ (as it is the case for the site $\CC$), 
pullbacks along more general smooth maps $\ell_i$ would appear.
\sk

Let us denote by
\begin{flalign}
\HH := \PSh(\CC,\Grpd)
\end{flalign}
the category of presheaves on $\CC$ with values in $\Grpd$.
(At this point our site structure on $\CC$ does not yet play a role, but it will enter later
when we introduce a model structure on $\HH$.)
An object in $\HH$ is a functor $X : \CC^\op \to \Grpd$
from the opposite category of $\CC$ to the category of groupoids.
A morphism in $\HH$ is a natural transformation $f : X \to Y$
between functors $X,Y : \CC^\op \to \Grpd$.
\sk

Let us recall the fully faithful Yoneda embedding 
\begin{flalign}\label{eqn:Yoneda}
\underline{(-)} : \CC\longrightarrow \HH~.
\end{flalign}
It assigns to an object $U$ in $\CC$ the functor 
$\underline{U} : \CC^\op\to \Grpd$ that acts on objects as
\begin{subequations}
\begin{flalign}
\underline{U}(U^\prime) := \Hom_{\CC}^{}(U^\prime,U)~,
\end{flalign} 
where $\Hom^{}_{\CC}$ are the morphism sets in $\CC$ which we regard 
as groupoids with just identity morphisms, and on 
morphisms $\ell^\prime : U^\prime\to U^{\prime\prime}$ as
\begin{flalign}
\underline{U}(\ell^\prime) : \underline{U}(U^{\prime\prime}) \longrightarrow \underline{U}(U^\prime)~,~~
\big(\ell : U^{\prime\prime}\to U\big)\longmapsto \big(\ell\circ \ell^\prime : U^{\prime} \to U\big)~.
\end{flalign}
\end{subequations}
To a morphism $\ell : U\to U^{\prime}$ in $\CC$ the Yoneda
embedding assigns the morphism $\underline{\ell} : \underline{U}\to \underline{U^\prime}$
in $\HH$ whose stages are
\begin{flalign}
\underline{\ell} : \underline{U}(U^{\prime\prime})\longrightarrow \underline{U^\prime}(U^{\prime\prime})
~,~~\big(\ell^\prime : U^{\prime\prime}\to U\big)\longmapsto \big(\ell\circ \ell^\prime : U^{\prime\prime} \to U^{\prime}\big)~,
\end{flalign}
for all objects $U^{\prime\prime}$ in $\CC$.
\sk

Given two objects $X$ and $Y$ in $\HH$, their product
$X\times Y$ in $\HH$ is given by the functor $X\times Y : \CC^\op\to \Grpd$ that acts on objects as
\begin{subequations}
\begin{flalign}
(X\times Y) (U) := X(U)\times Y(U)~,
\end{flalign}
where on the right-hand side the product is the one in $\Grpd$,
and on morphisms $\ell : U\to U^\prime$ as
\begin{flalign}
(X\times Y) (\ell) :=X(\ell)\times Y(\ell) :  (X\times Y)(U^\prime) \longrightarrow (X\times Y)(U)~.
\end{flalign}
\end{subequations}
The product $\times$ equips $\HH$ with the structure of a symmetric monoidal category.
The unit object is the constant presheaf of groupoids defined by the groupoid $\{\ast\}$,
i.e.\ the functor $\CC^\op \to\Grpd$ that assigns $U\mapsto \{\ast\}$ and $(\ell: U\to U^\prime)\mapsto
\id_{\{\ast\}}$. In the following we denote groupoids and their corresponding constant
presheaves of groupoids by the same symbols. Explicitly, when $G$ is a groupoid
we also denote by $G$ the presheaf of groupoids specified by $U\mapsto G$ and $
(\ell : U\to U^\prime)\mapsto \id_{G}$. For example, we denote the unit object in $\HH$
simply by $\{\ast\}$.
\sk

The symmetric monoidal category $\HH$ has internal homs, i.e.\ it is closed symmetric monoidal.
To describe those explicitly, we first have to introduce mapping groupoids
$\Grpd(X,Y)$ between two objects $X$ and $Y$ in $\HH$.
Recalling \eqref{eqn:homgrpd}, we define $\Grpd(X,Y)$ to be the groupoid
with objects given by all  $\HH$-morphisms
\begin{subequations}
\begin{flalign}
f : X\longrightarrow Y
\end{flalign}
and morphisms given by all $\HH$-morphisms
\begin{flalign}
u : X\times \Delta^1\longrightarrow Y~,
\end{flalign}
\end{subequations}
where $\Delta^1$ is the groupoid from \eqref{eqn:homgrpd} regarded as a constant presheaf of groupoids.
The internal hom-object $Y^X$ in $\HH$ is then given by the functor
$Y^X : \CC^\op \to \Grpd$ that acts on objects as
\begin{flalign}
Y^X(U) := \Grpd(\underline{U} \times X,Y)~.
\end{flalign}
On morphisms $\ell : U\to U^\prime$ in $\CC$, we define
\begin{flalign}
Y^X(\ell) : Y^X(U^\prime)\longrightarrow Y^X(U)~
\end{flalign}
as follows: It is the functor that assigns to an object
$f : \underline{U}^\prime\times X\to Y$ of the groupoid $ Y^X(U^\prime)$
the object $Y^X(\ell)(f) : \underline{U}\times X\to Y$ of the groupoid
$Y^X(U)$ that is defined by the composition
$Y^X(\ell)(f) := f \circ (\underline{\ell}\times \id_X)$.
To a morphism $u : \underline{U^\prime}\times X\times \Delta^1\to Y$
of the groupoid $ Y^X(U^\prime)$ it assigns the morphism $Y^X(\ell)(u) : 
\underline{U}\times X\times\Delta^1\to Y$ of the groupoid
$Y^X(U)$ that is defined by the composition
$Y^X(\ell)(u):=u \circ (\underline{\ell}\times\id_X\times \id_{\Delta^1})$.
For the terminal object $\bbR^0$ in $\CC$, we have that $Y^X(\bbR^0) = \Grpd(X,Y)$ is the mapping groupoid.
\sk

Our category $\HH$ can be equipped with (at least) two model structures.
To define them, let us recall that a 
morphism $f : X\to Y$ in $\HH$ is said to have the {\em left lifting property}
with respect to a morphism $f^\prime : X^\prime\to Y^\prime$ in $\HH$
if all commutative squares of the form
\begin{flalign}
\xymatrix{
\ar[d]_-{f}\ar[rr] X && \ar[d]^-{f^\prime}X^\prime\\
\ar@{-->}[rru]_-{f^{\prime\prime}}Y \ar[rr] && Y^\prime
}
\end{flalign}
admit a lift $f^{\prime\prime}$, i.e.\ the two triangles commute.
Vice versa, one says that a morphism $f^\prime : X^\prime\to Y^\prime$ in $\HH$
has the {\em right lifting property} with respect to a morphism $f : X\to Y$ in $\HH$
if all commutative squares of the form above admit a lift $f^{\prime\prime}$.
\begin{lem}[Global model structure on $\HH$ \cite{Hollander}]\label{lem:globmod}
Define a morphism $f : X\to Y$ in $\HH$ to be
\begin{itemize}
\item[(i)] a weak equivalence if each stage $f : X(U)\to Y(U)$ is a weak equivalence in $\Grpd$;

\item[(ii)] a fibration if each stage $f : X(U)\to Y(U)$ is a fibration in $\Grpd$;

\item[(iii)] a cofibration if it has the left lifting property with respect to all acyclic fibrations (i.e.\ all morphisms
in $\HH$ that are both fibrations and weak equivalences).
\end{itemize}
With these choices $\HH$ becomes a model category. This is called the global
model structure on $\HH$.
\end{lem}

The global model structure on $\HH = \PSh(\CC,\Grpd)$ 
is not yet the correct one as it does not encode any information about our site structure on $\CC$.
It was shown in \cite{Hollander} that one can localize the global model structure
to obtain what is called the local model structure on $\HH$. See also \cite{Hirschhorn}
or \cite[Section 5]{Dugger} for details on localizations of model categories.
The set of $\HH$-morphisms with respect to which one localizes is given by
\begin{flalign}\label{eqn:locmorph}
S := \big\{ \hocolim_{\HH}^{}\, \underline{U_{\bullet}} \longrightarrow \underline{U} \, : \, \{U_i\subseteq U\} \text{ good open cover}\, \big\}~.
\end{flalign}
As before, $\underline{U}$ is the object in $\HH$ that is represented by the object 
$U$ in $\CC$ via the Yoneda embedding.
By $\underline{U_{\bullet}} $ we denote the simplicial diagram in $\HH$ given by 
\begin{flalign}\label{eqn:simplicaldiag}
\underline{U_\bullet} := \bigg( 
\xymatrix@C=1.1pc{
\coprod\limits_{i} \underline{U_i} ~&~ \ar@<0.5ex>[l]\ar@<-0.5ex>[l] \coprod\limits_{ij} \underline{U_{ij}} ~&~\ar@<1ex>[l]\ar@<-1ex>[l]\ar[l] \coprod\limits_{ijk}\underline{U_{ijk}}
~&~ \ar@<0.5ex>[l]\ar@<-0.5ex>[l]\ar@<1.5ex>[l]\ar@<-1.5ex>[l]   \cdots
}  \bigg)~,
\end{flalign}
where $\coprod$ denotes the coproduct in $\HH$ and we suppressed as usual the degeneracy maps. 
Moreover, $\hocolim_{\HH}^{}\, \underline{U_{\bullet}} $ denotes the homotopy
colimit of the simplicial diagram \eqref{eqn:simplicaldiag} in $\HH$.
\sk

To describe the weak equivalences in the local model structure
on $\HH$ it is useful to assign sheaves of homotopy groups to 
presheaves of groupoids, cf.\ \cite[Section 5]{Hollander}.
\begin{defi}\label{def:sheavesofhomgroups}
Let $X : \CC^\op\to\Grpd$ be an object in $\HH$.
\begin{itemize}
\item[a)] Then $\pi_0 (X) : \CC^\op\to\Set$ is the presheaf of sets defined by
$\pi_0 (X)(U) := \pi_0(X(U))$, for all objects $U$ in $\CC$. (Here $\pi_0 (H)$ denotes the set 
of isomorphism classes of objects of a groupoid $H$.) 

\item[b)] Given any object $x$ in $X(U)$, then $\pi_1(X,x) : (\CC/U)^\op \to \Grp$ is the presheaf of groups
on the over-category $\CC/U$ defined by
\begin{flalign}
\pi_1(X,x)\big(\ell  \big) := \pi_1\big(X(U^\prime), X(\ell)(x)\big)~,
\end{flalign}
for all objects $\ell : U^\prime\to U$ in  $\CC/U$. (Here $\pi_1(H,y)$ denotes the automorphism group
of an object $y$ of a groupoid $H$.)
\end{itemize}
The sheafifications of $\pi_0 (X)$ and $\pi_1(X,x)$, for all objects $x$ in $X(U)$
and all objects $U$ in $\CC$, are called the {\em sheaves of homotopy groups} 
associated to an object $X$ in $\HH$.
\end{defi}

The following theorem summarizes the relevant aspects of the local model structure on $\HH$.
\begin{theo}[Local model structure on $\HH$ \cite{Hollander}]\label{theo:localH}
There exists a model structure on $\HH$ which is the localization of the global model structure
on $\HH$ (cf.\ Lemma \ref{lem:globmod}) with respect to the set of morphisms \eqref{eqn:locmorph}.
It is called the local model structure on $\HH$.  The following holds true:
\begin{itemize}
\item[(i)] A morphism $f : X\to Y$ in $\HH$ is a weak equivalence in 
the local model structure if and only if it induces an isomorphism on the associated sheaves of homotopy groups.

\item[(ii)] A morphism $f : X\to Y$ in $\HH$ is a cofibration in 
the local model structure if and only if it is a cofibration in the global model structure.

\item[(iii)] A morphism $f : X\to Y$ in $\HH$ is a fibration in 
the local model structure if and only if it has the right lifting property with respect
to all acyclic cofibrations in the local model structure. See also \cite[Proposition 4.2]{Hollander3}
or Proposition \ref{propo:localfibration} for a more explicit characterization.

\item[(iv)] An object $X$ in $\HH$ is fibrant in the local model structure, 
i.e.\ the canonical morphism $X\to\{\ast\}$ to the terminal object in $\HH$
is a fibration, if and only if for all good open covers $\{U_i\subseteq U\}$ the canonical morphism
\begin{flalign}\label{eqn:stack}
\xymatrix@C=1.1pc{
X(U) ~\ar[rr] &&~ \holim_{\Grpd}^{}\bigg(
\prod\limits_{i} X(U_i)  \ar@<0.5ex>[r]\ar@<-0.5ex>[r]~&~
\prod\limits_{ij} X(U_{ij}) \ar@<1ex>[r]\ar@<-1ex>[r]\ar[r]~&~ \prod\limits_{ijk} X(U_{ijk})
\ar@<0.5ex>[r]\ar@<-0.5ex>[r]\ar@<1.5ex>[r]\ar@<-1.5ex>[r]~&~    \cdots\bigg)
}
\end{flalign}
is a weak equivalence in $\Grpd$, where $\holim_{\Grpd}^{}$ denotes the homotopy 
limit of cosimplicial diagrams in $\Grpd$.

\item[(v)] A morphism $f : X\to Y$ in $\HH$ between two fibrant objects $X$ and $Y$ 
in the local model structure is a weak equivalence in the local model structure if and only if
each stage $f : X(U)\to Y(U)$ is a weak equivalence in $\Grpd$.
\end{itemize}
\end{theo}
\begin{rem}\label{rem:simplenotation}
To simplify notations, we denote the cosimplicial diagram of groupoids 
in \eqref{eqn:stack} also by
\begin{flalign}
X(U_{\bullet}) := \bigg( 
\xymatrix@C=1.1pc{\prod\limits_{i} X(U_i)  \ar@<0.5ex>[r]\ar@<-0.5ex>[r]~&~
\prod\limits_{ij} X(U_{ij}) \ar@<1ex>[r]\ar@<-1ex>[r]\ar[r]~&~ \prod\limits_{ijk} X(U_{ijk})
\ar@<0.5ex>[r]\ar@<-0.5ex>[r]\ar@<1.5ex>[r]\ar@<-1.5ex>[r]~&~    \cdots }\bigg)~.
\end{flalign}
Then the homotopy limit in \eqref{eqn:stack} simply reads as $\holim_{\Grpd}^{}\,X(U_{\bullet})$.
\end{rem}

Unless otherwise stated, we will always work with the local model structure
on $\HH$. The reason for this will be explained below.
Hence, by fibration, cofibration and weak equivalence in $\HH$ 
we always mean the ones in the local model structure on $\HH$.
\sk

Of particular relevance for us will be the fibrant objects in $\HH$ (in the local model structure).
It was shown by Hollander \cite{Hollander} that the fibrant condition \eqref{eqn:stack} captures 
the notion of descent for stacks. In particular, fibrant presheaves of groupoids
provide us with an equivalent, but simpler, model for stacks than the traditional models
based on lax presheaves of groupoids or categories fibered in groupoids, 
see e.g.\  \cite{DM,Giraud}. In our work we hence shall use the following definition of a stack
\cite{Hollander}.
\begin{defi}
A {\em stack} is a fibrant object $X$ in the local model structure on $\HH$. More concretely,
a stack is a presheaf of groupoids $X : \CC^\op\to \Grpd$ such that the canonical morphism
\eqref{eqn:stack} is a weak equivalence in $\Grpd$, for all good open covers $\{U_i\subseteq U\}$.
\end{defi}
We will later need the following standard 
\begin{lem}\label{lem:fibrationlocfibrant}
Let $X$ and $Y$ be fibrant objects in the local model structure on $\HH$, i.e.\ stacks.
Then an $\HH$-morphism $f: X \to Y$ is a fibration in the local model structure
if and only if it is a fibration in the global model structure, i.e.\  a stage-wise
fibration in $\Grpd$.
\end{lem}
\begin{proof}
This is \cite[Proposition 3.3.16]{Hirschhorn} applied to  
$\HH$ with the global model structure (cf.\ Lemma \ref{lem:globmod})
and to its localization with respect to \eqref{eqn:locmorph},
i.e.\ the local model structure of Theorem \ref{theo:localH}.
\end{proof}

\subsection{\label{subsec:examples}Examples of stacks}
We collect some well-known examples of stacks that will play a major role in our work.
We also refer to \cite{Fiorenza2,Fiorenza,Schreiber} for a description of some 
of these stacks in the language of $\infty$-stacks. 
All of our stacks are motivated by the structures arising in gauge theories.

\begin{ex}[Manifolds]\label{ex:manifolds}
To any (finite-dimensional and paracompact) manifold $M$ we may assign
an object $\underline{M}$ in $\HH$. It is given by the functor $\underline{M} : \CC^\op\to\Grpd$
that acts on objects as
\begin{flalign}
\underline{M}(U) := C^\infty(U,M)~,
\end{flalign}
where $C^\infty(U,M)$ is the set of smooth maps regarded as a groupoid with just identity morphisms,
and on morphisms $\ell : U\to U^\prime$ as
\begin{flalign}
\underline{M}(\ell) : \underline{M}(U^\prime)\longrightarrow \underline{M}(U)~,~~\big( \rho : U^\prime\to M \big)\longmapsto
\big(\rho\circ \ell : U \to M\big)~.
\end{flalign}
Note the similarity to the Yoneda embedding \eqref{eqn:Yoneda}, which is our motivation to use
the same notation by an underline. Notice further that, when $M = U$ 
is diffeomorphic to some $\bbR^n$, 
then $\underline{M}$ coincides with the Yoneda embedding $\underline{U}$ of the object $U$ in $\CC$.
\sk

Let us confirm that $\underline{M}$ is a stack, i.e.\ that \eqref{eqn:stack} is 
a weak equivalence in $\Grpd$ for all good open covers $\{U_i\subseteq U\}$.
In the notation of Remark \ref{rem:simplenotation}, we have to compute the
homotopy limit $\holim_{\Grpd}^{} \, \underline{M}(U_\bullet)$ of the cosimplicial diagram
$\underline{M}(U_\bullet)$ in $\Grpd$.
Using Proposition \ref{propo:descentcat}, we find that 
$\holim_{\Grpd}^{} \, \underline{M}(U_\bullet)$ is the groupoid whose
objects are families $\{ \rho_i\in C^\infty(U_i,M)\}$ satisfying
$\rho_i\vert_{U_{ij}}^{} = \rho_{j}\vert_{U_{ij}}^{}$, for all $i,j$, and whose morphisms are just the identities.
The canonical morphism $\underline{M}(U) \to \holim_{\Grpd}^{} \, \underline{M}(U_\bullet)$
assigns to $\rho\in C^\infty(U,M)$ the family $\{\rho\vert_{U_i}^{}\}$, hence it is an isomorphism 
(and thus also a weak equivalence) because $C^\infty(-,M)$ is a sheaf on $\CC$.\sk

Denoting by $\Man$ the category of (finite-dimensional and paracompact) manifolds,
it is easy to see that our constructions above define a fully faithful functor $\underline{(-)} : \Man\to \HH$
that takes values in stacks. Hence, manifolds can be equivalently described in terms of stacks.
\end{ex}

\begin{ex}[Classifying stack of principal $G$-bundles]\label{ex:BG}
Let $G$ be a Lie group. The object $\BG$ in $\HH$ is defined by the following functor
$\BG : \CC^\op\to \Grpd$: To an object $U$ in $\CC$ it assigns the groupoid $\BG(U)$ that has just one object, say $\ast$,
with automorphisms $C^\infty(U,G)$ given by the set of smooth maps from $U$ to the Lie group $G$.
The composition of morphisms in $\BG(U)$ is given by the opposite point-wise product $g^\prime\circ g := g\, g^\prime$
of $g,g^\prime \in C^\infty(U,G)$. (We take the opposite product because we will be 
interested later in group representations of $C^\infty(U,G)$ from the right.)
Moreover, the identity morphism in $\BG(U)$ is the constant map  $e\in C^\infty(U,G)$ from $U$
to the unit element in $G$.
To any morphism $\ell : U\to U^\prime$ in $\CC$ the functor $\BG$ assigns
the groupoid morphism  $\BG(\ell) : \BG(U^\prime)\to\BG(U)$ that acts on objects as $\ast\mapsto \ast$
and on morphisms as $C^\infty(U^\prime,G)\ni g \mapsto \ell^\ast g := g \circ\ell \in C^\infty(U,G)$. 
\sk

Let us confirm that $\BG$ is a stack. Let $\{U_i\subseteq U\}$ be any good open cover.
Using Proposition \ref{propo:descentcat}, we find that $\holim_{\Grpd}^{}\,\BG(U_\bullet)$
is the groupoid whose objects are families $\{g_{ij}\in C^\infty(U_{ij},G)\}$
satisfying $g_{ii}=e\in C^\infty(U_{i},G)$, for all $i$,
and the cocycle condition $g_{ij}\vert_{U_{ijk}}^{}~g_{jk}\vert_{U_{ijk}}^{} = g_{ik}\vert_{U_{ijk}}^{}$, for all $i,j,k$.
The morphisms of $\holim_{\Grpd}^{}\,\BG(U_\bullet)$ from $\{g_{ij}\}$ to $\{g^\prime_{ij}\}$
are families $\{h_i\in C^\infty(U_i,G)\}$  satisfying
$g^{\prime}_{ij} = h_i^{-1}\vert_{U_{ij}}^{} ~ g_{ij}~h_j\vert_{U_{ij}}^{}$, for all $i,j$.
The canonical morphism $\BG(U) \to \holim_{\Grpd}^{}\,\BG(U_\bullet)$
assigns to the object $\ast$ in $\BG(U)$ the object $\{ g_{ij} =e \in C^\infty(U_{ij},G)  \}$, i.e.\ the trivial cocycle,
and to a morphism $g\in C^\infty(U,G)$ in $\BG(U)$ the isomorphism
$\{ g\vert_{U_i}^{} \in C^\infty(U_i,G)\} $ of the trivial cocycle.
The canonical morphism is a weak equivalence in $\Grpd$ 
(cf.\ Theorem \ref{theo:groupoids}) because of the following reasons:
1.)~It is fully faithful because $C^\infty(-,G)$ is a sheaf on $\CC$.
2.)~It is essentially surjective because all cocycles are trivializable on manifolds diffeomorphic to $\bbR^n$ 
($U$ in this case), i.e.\ one can find $\{h_i\in C^\infty(U_i,G)\}$ 
such that $g_{ij} = h_i^{-1}\vert_{U_{ij}}^{} ~h_j\vert_{U_{ij}}^{}$, for all $i,j$. 
Hence, we have shown that $\BG$ is a stack.
It is called the {\em classifying stack of principal $G$-bundles}, see also \cite[Section 3.2]{Fiorenza2}.
\sk

Our classifying stack $\BG$ is a smooth and stacky analog of the usual classifying space
of a topological group $G$ \cite[Section 3]{Segal}, 
which is the topological space obtained as the 
geometric realization of the simplicial topological space
\begin{flalign}\label{eqn:BGtopological}
\xymatrix@C=1.1pc{
\{\ast\} ~&~ \ar@<0.5ex>[l]\ar@<-0.5ex>[l] G ~&~\ar@<1ex>[l]\ar@<-1ex>[l]\ar[l] G^2
~&~ \ar@<0.5ex>[l]\ar@<-0.5ex>[l]\ar@<1.5ex>[l]\ar@<-1.5ex>[l]   \cdots
}  ~.
\end{flalign}
Notice that \eqref{eqn:BGtopological} may be obtained by equipping the
nerve of the groupoid $\BG(\{\ast\})$ with the topologies induced by
$G$. 
\end{ex}

\begin{ex}[Classifying stack of principal $G$-bundles with connections]\label{ex:BGcon}
Let again $G$ be a Lie group, which we assume to be a matrix Lie group in order to simplify some formulas below.
Let $\g$ denote the Lie algebra of $G$.
The object $\BGcon$ in $\HH$ is defined by the following functor $\BGcon : \CC^\op \to \Grpd$:
To an object $U$ in $\CC$ it assigns the groupoid $\BGcon(U)$ whose set of objects is $\Omega^1(U,\g)$, i.e.\  the
set of Lie algebra valued $1$-forms on $U$, and whose set of morphisms is $\Omega^1(U,\g)\times C^\infty(U,G)$.
The source and target of a morphism $(A,g)\in \Omega^1(U,\g)\times C^\infty(U,G)$ is as follows
\begin{flalign}
\xymatrix@C=1.1pc{
A~\ar[rr]^-{(A,g)} &&~ A\ra g := g^{-1}\,A\,g + g^{-1}\,\dd g
}~, 
\end{flalign}
where $\dd$ denotes the de Rham differential and $A\ra g$ is the usual right action 
of gauge transformations $g\in C^\infty(U,G)$ on gauge fields $A\in \Omega^1(U,\g)$. 
The identity morphism is $(A,e) : A\to A$ and composition of two 
composable morphisms in $\BGcon(U)$ is given by
\begin{flalign}
\xymatrix@C=1.1pc{
\ar@/_1pc/[rrrr]_-{(A,g\,g^\prime)}  A~\ar[rr]^-{(A,g)} &&~ A\ra g ~ \ar[rr]^-{(A\ra g,g^\prime)} && ~ A\ra(g\,g^\prime)
}~. 
\end{flalign}
To any morphism $\ell : U\to U^\prime$ in $\CC$ the functor $\BGcon$ assigns
the groupoid morphism $\BGcon(\ell) : \BGcon(U^\prime)\to \BGcon(U)$ that acts on objects
as $A\mapsto \ell^\ast A$ (i.e.\ via pullback of differential forms) and on morphisms
as $(A,g) \mapsto (\ell^\ast A , \ell^\ast g)$.
\sk

Let us confirm that $\BGcon$ is a stack. Let $\{U_i\subseteq U\}$ be any good open cover.
Using Proposition \ref{propo:descentcat}, we find that $\holim_{\Grpd}^{}\,\BGcon(U_\bullet)$
is the groupoid whose objects are pairs of families $(\{A_i\in\Omega^1(U_i,\g) \} , \{g_{ij}\in C^\infty(U_{ij},G)\})$,
such that $A_j\vert_{U_{ij}}^{} = A_i\vert_{U_{ij}}^{} \ra g_{ij}$, for all $i,j$,
$g_{ii}=e$, for all $i$, and $g_{ij}\vert_{U_{ijk}}^{}~g_{jk}\vert_{U_{ijk}}^{} = g_{ik}\vert_{U_{ijk}}^{}$, for all $i,j,k$.
The morphisms of  $\holim_{\Grpd}^{}\,\BGcon(U_\bullet)$ from 
$(\{A_i \} , \{g_{ij}\})$ to $(\{A^\prime_i \} , \{g^\prime_{ij}\})$ are families $\{h_i\in C^\infty(U_i,G)\}$
satisfying $A_i^\prime = A_i\ra h_i$, for all $i$, and $g_{ij}^{\prime} = h_i^{-1}\vert_{U_{ij}}^{}~g_{ij}~h_j\vert_{U_{ij}}^{}$,
for all $i,j$. The canonical morphism $\BGcon(U) \to \holim_{\Grpd}^{}\,\BGcon(U_\bullet)$
assigns to an object $A\in\Omega^1(U,\g)$ of $\BGcon(U)$ the pair of families
$( \{ A\vert_{U_i}^{} \}, \{ g_{ij} =e \} )$, with the trivial cocycle, and to a morphism
$(A,g)\in\Omega^1(U,\g)\times C^\infty(U,G)$ in $\BGcon(U)$ the family $\{ g\vert_{U_i}^{}\}$.
With a similar argument as in Example \ref{ex:BG} we find that the canonical morphism
is a weak equivalence in $\Grpd$ and hence that $\BGcon$ is a stack. It is called
the {\em classifying stack of principal $G$-bundles with connections}, see also \cite[Section 3.2]{Fiorenza2}.
\end{ex}

\begin{ex}[Classifying stack of principal $G$-bundles with adjoint bundle valued $p$-forms]\label{ex:Liealgforms}
Let $G$ be a matrix Lie group and $\g$ its Lie algebra.
For $p\in\bbZ_{\geq 0}$, we define the object $\BG_{\Omega^p_{\ad}}$ in $\HH$
by the following functor $\BG_{\Omega^p_{\ad}} : \CC^\op\to \Grpd$:
To an object $U$ in $\CC$ it assigns the groupoid $\BG_{\Omega^p_{\ad}} (U)$ whose set
of objects is the set of $\g$-valued $p$-forms $\Omega^p(U,\g)$ and whose set of 
morphisms is $\Omega^p(U,\g)\times C^\infty(U,G)$. The source
and target of a morphism $(\omega,g)\in \Omega^p(U,\g)\times C^\infty(U,G)$ is, similarly to Example \ref{ex:BGcon},
given by
\begin{flalign}
\xymatrix@C=1.1pc{
\omega~\ar[rr]^-{(\omega,g)} &&~ \ad_g(\omega) := g^{-1}\,\omega\,g 
}~, 
\end{flalign}
where $\ad_g(\omega)$ is the usual right adjoint 
action of gauge transformations $g\in C^\infty(U,G)$ on $\g$-valued $p$-forms $\omega\in\Omega^p(U,\g)$.
The identity morphism is $(\omega,e) : \omega\to\omega$ and the composition of two composable
morphisms in $\BG_{\Omega^p_{\ad}}(U)$ is given by
\begin{flalign}
\xymatrix@C=2pc{
\ar@/_1pc/[rrrr]_-{(\omega , g\,g^\prime)} \omega~\ar[rr]^-{(\omega , g)} &&~ \ad_{g}(\omega) ~ \ar[rr]^-{(\ad_g(\omega) , g^\prime)} && ~ \ad_{g\,g^\prime}(\omega)
}~. 
\end{flalign}
To any morphism $\ell : U\to U^\prime$ in $\CC$ the functor $\BG_{\Omega^p_{\ad}}$ assigns
the groupoid morphism $\BG_{\Omega^p_{\ad}}(\ell) : \BG_{\Omega^p_{\ad}}(U^\prime)\to \BG_{\Omega^p_{\ad}}(U)$ 
that acts on objects and 
morphisms  via pullback, i.e.\  $\omega \mapsto \ell^\ast \omega$ 
and $(\omega ,g) \mapsto (\ell^\ast \omega , \ell^\ast g)$. The fact that $\BG_{\Omega^p_{\ad}}$ is a stack
can be proven analogously to Examples \ref{ex:BG} and \ref{ex:BGcon}. We call it the {\em classifying stack
of principal $G$-bundles with adjoint bundle valued $p$-forms}. We will later use 
$\BG_{\Omega^p_{\ad}}$ to describe, for example, the curvatures of connections.
\end{ex}

\subsection{Fiber product of stacks and mapping stacks}
Our constructions in this paper will use a variety of techniques 
to produce new stacks out of old ones, e.g.\ out of the stacks described in Section \ref{subsec:examples}.
Of particular relevance for us will be fiber products of stacks and mapping stacks. 
\sk

The following observation is standard in homotopy theory, however it 
is crucial to understand the definitions below:
Given a pullback diagram
\begin{flalign}\label{eqn:pullback}
\xymatrix{
X\ar[r]^-{f_1} & Z &\ar[l]_-{f_2} Y
}
\end{flalign}
in $\HH$, we may of course compute the fiber product $X\times_Z^{} Y$
in the usual way by taking the limit of this diagram. (Recall that $\HH$ is complete, hence all
limits exist in $\HH$.) The problem with this construction is that it does not preserve weak equivalences in $\HH$,
i.e.\ replacing the pullback diagram by a weakly equivalent one in general leads to a fiber product that
is not weakly equivalent to $X\times_Z^{} Y$. These problems are avoided by replacing 
the limit with the {\em homotopy limit}, which is a derived functor of the limit functor \cite{Dwyer}.
We denote the homotopy limit of the diagram \eqref{eqn:pullback} by
$X\times_Z^{h} Y$ and call it the {\em homotopy fiber product} in $\HH$. 
A similar problem arises when we naively use the internal hom-object
$Y^X$ in $\HH$ in order to describe an ``object of mappings'' from $X$ to $Y$;
replacing $X$ and $Y$ by weakly equivalent objects in $\HH$ in general
leads to an internal hom-object that is not weakly equivalent to $Y^X$.
These problems are avoided by replacing the internal-hom functor with its derived functor.
\sk

We now give rather explicit models for the homotopy fiber product and the derived internal hom.
The following model for the homotopy fiber product in $\HH$ was obtained in \cite{Hollander2}.
\begin{propo}\label{propo:homfibprod}
The homotopy fiber product $X\times_Z^{h} Y$ in $\HH$, i.e.\ the 
homotopy limit of the pullback diagram \eqref{eqn:pullback}, 
is the presheaf of groupoids defined as follows: For all objects $U$ in $\CC$,
\begin{itemize}
\item the objects of $(X\times_Z^{h} Y)(U)$ are triples $(x,y,k)$, where $x$ is an object in $X(U)$,
$y$ is an object in $Y(U)$ and $k : f_1(x)\to f_2 (y)$ is a morphism in $Z(U)$; and

\item the morphisms of $(X\times_Z^{h} Y)(U)$ from $(x,y,k)$ to $(x^\prime,y^\prime,k^\prime)$
are pairs $(g,h)$, where $g: x\to x^\prime$ is a morphism in $X(U)$ and
$h : y\to y^\prime$ is a morphism in $Y(U)$, such that the diagram
\begin{flalign}
\xymatrix{
\ar[d]_-{k} f_1(x) \ar[rr]^-{f_1(g)} && f_1(x^\prime)\ar[d]^-{k^\prime}\\
f_2(y) \ar[rr]_-{f_2(h)} && f_2(y^\prime)
}
\end{flalign}
in $Z(U)$ commutes.
\end{itemize}
If $X$, $ Y$ and $Z$ in \eqref{eqn:pullback} are stacks, then $X\times_Z^{h} Y$ is a stack too.
\end{propo}

Our model for the derived internal hom-functor is the standard one resulting from 
the theory of derived functors, see \cite{Dwyer} and also \cite{Hovey} for details. 
For the theory of derived functors to apply to our present situation, 
it is essential that $\HH$ (with the local model structure)
is a monoidal model category, see Appendix \ref{app:monoidal} for a proof.
\begin{propo}\label{propo:derivedhoms}
The derived internal hom-functor in $\HH$ is 
\begin{flalign}
R(-)^{Q(-)} : \HH^\op\times \HH\longrightarrow \HH~,
\end{flalign} 
where $Q : \HH\to \HH$ is any cofibrant replacement
functor and $R :\HH\to\HH$ is any fibrant replacement functor.
The following holds true:
\begin{itemize}
\item[(i)]
If $X$ is any object in $\HH$ and $Y$ is a stack, then $R(Y)^{Q(X)}$ 
is weakly equivalent to the object  $Y^{Q(X)}$ in $\HH$.

\item[(ii)] If $X$ is a cofibrant object in $\HH$ and $Y$ is a stack, then $R(Y)^{Q(X)}$ 
is weakly equivalent to the ordinary internal hom-object  $Y^X$ in $\HH$.

\item[(iii)] If $X$ is any object in $\HH$ and $Y$ is a stack,
then $Y^{Q(X)}$  is a stack too. (This follows from \cite[Remark 4.2.3]{Hovey}.)

\end{itemize}
\end{propo}

\begin{rem}
In our work we shall only need derived internal hom-objects $R(Y)^{Q(X)}$ 
for the case where $Y$ is a stack and $X=\underline{M}$ is the stack represented by
a manifold $M$, cf.\ Example \ref{ex:manifolds}. We develop in Appendix \ref{app:cofibrep}
a suitable cofibrant replacement functor $\widecheck{(-)} : \Man_{\hookrightarrow}\to \HH$
on the category $\Man_{\hookrightarrow}$ of (finite-dimensional and paracompact) manifolds
with morphisms given by open embeddings. Proposition \ref{propo:derivedhoms}
then implies that we can compute, up to a weak equivalence, 
such derived internal hom-objects by $Y^{\widecheck{M}}$.
\end{rem}


\section{\label{sec:gaugefields}Gauge fields on a manifold}
Let $G$ be a matrix Lie group and $M$ a (finite-dimensional and paracompact) manifold.
The goal of this section is to construct and study the stack $\GCon(M)$ 
of principal $G$-bundles with connections on $M$, which we shall also call the 
{\em stack of gauge fields} on $M$. Let us recall our main guiding principle:
The stack $\GCon(M)$ is supposed to describe smoothly parametrized
families of principal $G$-bundles with connections on $M$, which is motivated
by the functor of points perspective. More precisely, this means that evaluating 
the stack $\GCon(M)$ on an object $U$ in $\CC$, we would like to 
discover (up to weak equivalence) the groupoid describing smoothly 
$U$-parametrized principal $G$-bundles with connections on $M$
and their smoothly $U$-parametrized gauge transformations.
We shall obtain the stack $\GCon(M)$ by an intrinsic construction in $\HH$
that is given by a concretification of the derived mapping stack 
from $M$ to the classifying stack $\BGcon$ of principal $G$-bundles with connections 
(cf.\ Example \ref{ex:BGcon}). Using Proposition \ref{propo:derivedhoms},
we may compute (up to a weak equivalence) the derived mapping stack in
terms of the ordinary internal-hom ${\BGcon}_{~}^{\widecheck{M}}$, where
$\widecheck{M}$ is the canonical cofibrant replacement of $M$ (cf.\ Appendix \ref{app:cofibrep}).
It is important to emphasize that our concretification prescription improves the one 
originally proposed in \cite{Fiorenza,Schreiber}, which fails to give the desired result, 
i.e.\ a stack describing smoothly parametrized families of principal $G$-bundles with connections, 
together with smoothly parametrized gauge transformations, see 
Appendix \ref{app:concretification} for details. 
\sk

By construction, our model ${\BGcon}_{~}^{\widecheck{(-)}} : \Man_{\hookrightarrow}^\op \to \HH$
for the derived mapping stacks is functorial on the category of (finite-dimensional and paracompact) 
manifolds with morphisms given by open embeddings. The same holds true
for the stacks of gauge fields, i.e.\ we will obtain a functor 
$\GCon : \Man_{\hookrightarrow}^\op \to \HH$.
This is an advantage compared to the usual approach to construct the stack of gauge fields by 
using a (necessarily non-functorial) {\em good} open cover to obtain a cofibrant 
replacement for $M$, see e.g.\ \cite{Fiorenza2}.

\subsection{\label{subsec:mapstack}Mapping stack}
We now compute the object ${\BGcon}_{~}^{\widecheck{M}}$ in $\HH$, 
which by Proposition \ref{propo:derivedhoms} is a weakly equivalent model for
the derived internal-hom object from $M$ to $\BGcon$. 
Its underlying functor ${\BGcon}_{~}^{\widecheck{M}} : \CC^\op \to \Grpd$
assigns to an object $U$ in $\CC$ the mapping groupoid
\begin{flalign}\label{eqn:mappingstackU}
{\BGcon}_{~}^{\widecheck{M}}(U) = \Grpd(\underline{U}\times \widecheck{M}, \BGcon)~.
\end{flalign}
By definition, an object in this groupoid is a morphism $f : \underline{U}\times \widecheck{M} \to \BGcon$
in $\HH$ and a morphism in this groupoid 
is a morphism $u: \underline{U}\times \widecheck{M} \times\Delta^1 \to \BGcon$ in $\HH$.
We would like to describe the objects and morphisms of the groupoid \eqref{eqn:mappingstackU}
by more familiar data related to gauge fields and gauge transformations.
To achieve this goal, we have to explicate the stages of the $\HH$-morphisms 
$f$ and $u$. 
\sk

Let us start with $f$ and consider the associated groupoid morphism
$f : \underline{U}(U^\prime)\times \widecheck{M}(U^\prime) \to \BGcon(U^\prime) $
at stage $U^\prime$ in $\CC$. Using the explicit description of $\widecheck{M}(U^\prime) $ 
given in Appendix \ref{app:cofibrep}, we may visualize the objects in 
$ \underline{U}(U^\prime)\times \widecheck{M}(U^\prime)$ by diagrams
\begin{flalign}
(\ell,(V,\nu)) := \Big(\xymatrix{
M & \ar[l]_-{\rho_V}V & \ar[l]_-{\nu}U^\prime  \ar[r]^-{\ell}& U
}\Big)~
\end{flalign}
of smooth maps, where $V$ runs over all open subsets of $M$ that are diffeomorphic to $\bbR^{\dim(M)}$ 
and the last arrow points to the right.
Morphisms in $ \underline{U}(U^\prime)\times \widecheck{M}(U^\prime)$ may be visualized 
by commutative diagrams
\begin{flalign}\label{eqn:mappingstackUmor}
\xymatrix@R=0.1pc{\\ (\ell,(V,\nu),(V^\prime,\nu^\prime)) ~:=~ \Bigg( \\}
\xymatrix@R=0.6pc{
&\ar[ld]_-{\rho_V}~V~&&\\
M ~&& \ar[lu]_-{\nu}\ar[ld]^-{\nu^\prime} ~U^\prime ~ \ar[r]^-{\ell} & ~U\\
&\ar[lu]^-{\rho_{V^\prime}} ~V^\prime~&&
}
\xymatrix@R=0.1pc{\\ \Bigg) \\}
\end{flalign}
of smooth maps.
The groupoid morphism 
$f : \underline{U}(U^\prime)\times \widecheck{M}(U^\prime) \to \BGcon(U^\prime) $
is then given by an assignment
\begin{flalign}
\nn f ~:~\qquad\qquad  (\ell, (V,\nu)) ~~&\longmapsto~~ A_{(\ell, (V,\nu))}\in\Omega^1(U^\prime,\g)~,\\
(\ell,(V,\nu),(V^\prime,\nu^\prime))~~ &\longmapsto~~ g_{(\ell,(V,\nu),(V^\prime,\nu^\prime)) } \in C^\infty(U^\prime,G)~,\label{eqn:mappingstackUfunc}
\end{flalign}
satisfying the following compatibility conditions: 1.)\ For all morphisms $(\ell,(V,\nu),(V^\prime,\nu^\prime))$ 
in $ \underline{U}(U^\prime)\times \widecheck{M}(U^\prime)$,
\begin{subequations}
\begin{flalign}
A_{(\ell, (V^\prime,\nu^\prime))} = A_{(\ell, (V,\nu))}\ra g_{(\ell,(V,\nu),(V^\prime,\nu^\prime)) }~.
\end{flalign}
2.)\ For all  objects $(\ell,(V,\nu))$ in $ \underline{U}(U^\prime)\times \widecheck{M}(U^\prime)$,
$g_{(\ell,(V,\nu),(V,\nu))} = e$. 3.)\ For all composable morphisms $(\ell,(V,\nu),(V^\prime,\nu^\prime))$
and $(\ell,(V^\prime,\nu^\prime),(V^{\prime\prime},\nu^{\prime\prime}))$
in $ \underline{U}(U^\prime)\times \widecheck{M}(U^\prime)$,
\begin{flalign}
g_{(\ell,(V,\nu),(V^\prime,\nu^\prime))}~g_{(\ell,(V^\prime,\nu^\prime),(V^{\prime\prime},\nu^{\prime\prime}))} 
= g_{(\ell,(V,\nu),(V^{\prime\prime},\nu^{\prime\prime}))}~.
\end{flalign}
\end{subequations}
The property of $f : \underline{U}(U^\prime)\times \widecheck{M}(U^\prime) \to \BGcon(U^\prime) $
being the stages of an $\HH$-morphism (i.e.\ natural transformation of functors $\CC^\op\to \Grpd$)
leads to the following coherence conditions:
For all morphisms $\ell^\prime : U^\prime\to U^{\prime\prime}$ in $\CC$
and all objects $(\ell,(V,\nu))$ in $ \underline{U}(U^{\prime\prime})\times \widecheck{M}(U^{\prime\prime})$,
\begin{subequations}
\begin{flalign}
A_{(\ell\circ\ell^\prime,(V,\nu\circ \ell^\prime))} = {\ell^\prime}^\ast A_{(\ell,(V,\nu))}~,
\end{flalign}
and, for all morphisms $\ell^\prime : U^\prime\to U^{\prime\prime}$ in $\CC$
and all morphisms $(\ell,(V,\nu),(V^\prime,\nu^\prime))$ 
in $ \underline{U}(U^{\prime\prime})\times \widecheck{M}(U^{\prime\prime})$,
\begin{flalign}
g_{(\ell\circ\ell^\prime,(V,\nu\circ \ell^\prime),(V^\prime,\nu^\prime\circ \ell^\prime) )}
={\ell^\prime}^\ast g_{(\ell,(V,\nu),(V^\prime,\nu^\prime) )}~.
\end{flalign}
\end{subequations}
These coherences constrain the amount of independent data described by 
\eqref{eqn:mappingstackUfunc}: Given any object $(\ell, (V,\nu))$ 
in $ \underline{U}(U^\prime)\times \widecheck{M}(U^\prime)$,
we notice that there is a factorization
\begin{flalign}
\xymatrix{
M ~ &~\ar[l]_-{\rho_V} V ~&~\ar[l]_-{\pr_V} V \times U~\ar[r]^-{\pr_{U}} & ~U\\
&&U^\prime\ar[u]_-{(\nu,\ell)}&
}
\end{flalign}
which implies that
\begin{flalign}
A_{(\ell,(V,\nu))} = (\nu,\ell)^\ast A_{(\pr_U,(V,\pr_V))}~.
\end{flalign}
In words, this means that the $A$ of the form $A_{(\pr_U,(V,\pr_V))} \in \Omega^1(V\times U,\g)$, 
for all open subsets $V\subseteq M$ diffeomorphic to $\bbR^{\dim(M)}$, determine all the others.
As there are no further coherences between $A$'s of the form $A_{(\pr_U,(V,\pr_V))}$,
it follows that the action of the functor \eqref{eqn:mappingstackUfunc} on objects is uniquely specified by choosing
\begin{flalign}\label{eqn:AV}
A_{V}:= A_{(\pr_U,(V,\pr_V))} \in \Omega^1(V\times U,\g)~,
\end{flalign} 
for all open subsets $V\subseteq M$ diffeomorphic to $\bbR^{\dim(M)}$.
A similar argument applies to the morphisms assigned by \eqref{eqn:mappingstackUfunc},
which are determined by
\begin{flalign}
g_{VV^\prime} := g_{(\pr_U,(V,\pr_V),(V^\prime,\pr_{V^\prime}))} \in C^\infty((V\cap V^\prime) \times U,G)~,
\end{flalign}
for all open subsets $V\subseteq M$ and $V^\prime\subseteq M$ diffeomorphic to $\bbR^{\dim(M)}$,
and the coherences
\begin{flalign}
g_{(\ell,(V,\nu),(V^\prime,\nu^\prime))} = (\nu,\nu^\prime,\ell)^\ast g_{(\pr_U,(V,\pr_V),(V^\prime,\pr_{V^\prime}))} ~,
\end{flalign}
for all morphisms $(\ell,(V,\nu),(V^\prime,\nu^\prime))$ in $ \underline{U}(U^\prime)\times \widecheck{M}(U^\prime)$.
The morphisms of the groupoid \eqref{eqn:mappingstackU} can be analyzed analogously.
In summary, we obtain
\begin{lem}\label{lem:BGconcheckM}
The objects of the groupoid \eqref{eqn:mappingstackU} 
are pairs of families
\begin{flalign}
\Big(\big\{ A_V  \in\Omega^1(V\times U,\g)  \big\},
\big\{ g_{VV^\prime} \in C^\infty((V\cap V^\prime)\times U , G)\big\}\Big)~,
\end{flalign}
where $V\subseteq M$ and $V^\prime \subseteq M$ run over all 
open subsets diffeomorphic to $\bbR^{\dim(M)}$,
which satisfy the following conditions:
\begin{itemize}
\item For all open subsets $V\subseteq M$ and $V^\prime \subseteq M$ diffeomorphic to $\bbR^{\dim(M)}$,
\begin{flalign}
A_{V^\prime}\vert_{(V\cap V^\prime)\times U}^{} ~=~ A_V\vert_{(V\cap V^\prime)\times U}^{} \ra g_{VV^\prime}~.
\end{flalign}
\item For all open subsets $V\subseteq M$ diffeomorphic to $\bbR^{\dim(M)}$, $g_{VV} =e$.
\item For all open subsets $V\subseteq M$, 
$V^\prime \subseteq M$ and $V^{\prime\prime} \subseteq M$ diffeomorphic to $\bbR^{\dim(M)}$,
\begin{flalign}
g_{VV^\prime}\vert_{(V\cap V^\prime\cap V^{\prime\prime})\times U}^{}~
g_{V^\prime V^{\prime\prime}}\vert_{(V\cap V^\prime\cap V^{\prime\prime})\times U}^{}
~=~
g_{VV^{\prime\prime}}\vert_{(V\cap V^\prime\cap V^{\prime\prime})\times U}^{}~.
\end{flalign}
\end{itemize}
The morphisms of the groupoid \eqref{eqn:mappingstackU}
from $(\{A_V\},\{g_{VV^\prime}\})$ to $(\{A_{V}^\prime\}, \{g^\prime_{VV^\prime}\})$
are families
\begin{flalign}
\big\{ h_{V} \in C^\infty(V\times U,G)  \big\}~,
\end{flalign}
where $V\subseteq M$ runs over all open subsets diffeomorphic to $\bbR^{\dim(M)}$,
which satisfy the following conditions:
\begin{itemize}
\item For all open subsets $V\subseteq M$ diffeomorphic to $\bbR^{\dim(M)}$,
\begin{flalign}
A_{V}^\prime ~=~ A_V \ra h_V~.
\end{flalign}
\item For all open subsets $V\subseteq M$ and $V^\prime \subseteq M$ diffeomorphic to $\bbR^{\dim(M)}$,
\begin{flalign}
g_{VV^\prime}^\prime ~=~ h_V^{-1}\vert^{}_{(V\cap V^\prime)\times U}~~g_{VV^\prime}~~h_{V^\prime}\vert^{}_{(V\cap V^\prime)\times U}~.
\end{flalign}
\end{itemize}
\end{lem}

Moreover, the functor
${\BGcon}_{~}^{\widecheck{M}} : \CC^\op\to\Grpd$
assigns to a morphism $\ell : U\to U^\prime$ in $\CC$  the groupoid morphism 
${\BGcon}_{~}^{\widecheck{M}}(\ell) : {\BGcon}_{~}^{\widecheck{M}}(U^\prime) \to {\BGcon}_{~}^{\widecheck{M}}(U)$
specified by
\begin{flalign}
\nn {\BGcon}_{~}^{\widecheck{M}}(\ell) ~:~ \big(\{A_V\},\{g_{VV^\prime}\}\big) ~&\longmapsto~ \big(\{(\id_V\times \ell)^\ast \,A_V\} , \{(\id_{V\cap V^\prime}\times\ell)^\ast\,g_{VV^\prime}\}\big)~,\\
\{h_{V} \} ~&\longmapsto ~ \{(\id_{V}\times\ell)^\ast\, h_{V}\}~.
\end{flalign}
This completes our description of the mapping stack 
${\BGcon}_{}^{\widecheck{M}}$.
\begin{rem}\label{rem:mappinstackfunctor}
As a consequence of our functorial cofibrant replacement of manifolds (cf.\ Corollary \ref{cor:functorialcofibrepMan}),
the mapping stacks ${\BGcon}_{~}^{\widecheck{(-)}} : \Man_{\hookrightarrow}^{\op}\to \HH$ are functorial
on the category of manifolds with open embeddings. Explicitly, given any open embedding 
$f : M\to M^\prime$, the stages of the associated $\HH$-morphism
${\BGcon}_{~}^{\widecheck{f}} : {\BGcon}_{~}^{\widecheck{M^\prime}} \to {\BGcon}_{~}^{\widecheck{M}}$
are given by
\begin{flalign}
\nn {\BGcon}_{~}^{\widecheck{f}} ~:~\big(\{A_W\},\{g_{WW^\prime}\}\big)~&\longmapsto~\big(\{ (f\vert_V\times \id_U)^\ast \,A_{f(V)}\},\{(f\vert_{V\cap V^\prime}\times \id_U)^\ast \,g_{f(V)f(V^\prime)}\}\big)~,\\
\{h_{W}\} ~&\longmapsto ~\{(f\vert_V\times \id_U)^\ast\, h_{f(V)}\}~,
\end{flalign}
where $W \subseteq M^\prime$, $W^\prime \subseteq M^\prime$, $V \subseteq M$ 
and $V^\prime \subseteq M$ are open subsets diffeomorphic to $\bbR^{\dim(M)}=\bbR^{\dim(M^\prime)}$ 
and $f\vert_V : V\to f(V)$ denotes the restriction of $f: M\to M^\prime$ to $V$ and its image $f(V)$. 
(The smooth map $f\vert_{V \cap V^\prime} : V \cap V^\prime\to f(V) \cap f(V^\prime)$ is defined similarly.)
\end{rem}

\subsection{\label{subsec:concrete}Concretification}
The mapping stack ${\BGcon}_{~}^{\widecheck{M}}$ we described in the previous
subsection (see in particular Lemma \ref{lem:BGconcheckM})
is not yet the correct stack of gauge fields on the manifold $M$.
Even though the groupoid of {\em global} points ${\BGcon}_{~}^{\widecheck{M}}(\bbR^0)$
correctly describes the gauge fields and gauge transformations on $M$,
the smooth structure on ${\BGcon}_{~}^{\widecheck{M}}$, which is encoded
in the groupoids ${\BGcon}_{~}^{\widecheck{M}}(U)$ for all other objects $U$ in $\CC$,
is not the desired one yet. In fact,
the groupoid ${\BGcon}_{~}^{\widecheck{M}}(U)$  describes by construction gauge fields
and gauge transformations on the product manifold $M\times U$, while 
the correct stack of gauge fields on $M$, when evaluated on $U$, should be the groupoid
of smoothly $U$-parametrized gauge fields and gauge transformations on $M$. Hence, 
the problem is that ${\BGcon}_{~}^{\widecheck{M}}(U)$ includes also gauge fields 
along the parameter space $U$ and not only along $M$.
\sk

This problem has already been observed in \cite{Fiorenza,Schreiber}, 
where a solution in terms of {\em concretification} was proposed. The goal of this subsection
is to work out explicitly the concretification of our mapping stack ${\BGcon}_{~}^{\widecheck{M}}$. 
This is achieved by using the results on fibrant replacements in the $(-1)$-truncation 
of the model structure on over-categories $\HH/K$
that we develop in Appendix \ref{app:truncation}. As we explain in more detail in Appendix \ref{app:concretification}, 
the original concretification prescription of \cite{Fiorenza,Schreiber} fails to produce the desired result, i.e.\ a stack
which describes smoothly parametrized families of principal $G$-bundles with connections, 
together with smoothly parametrized
gauge transformations. Hence, we propose an improved concretification prescription 
that is valid for the case of interest in this paper, 
namely principal bundles and connections on manifolds $M$.
A general concretification prescription for the $\infty$-stacks of higher bundles and connections
is beyond the scope of this paper and will be developed elsewhere.
\sk

Crucial for concretification is existence of the following Quillen adjunction
\begin{flalign}\label{eqn:flatsharpadj}
\xymatrix{
\flat : \HH ~\ar@<0.5ex>[r]&\ar@<0.5ex>[l]  ~\HH : \sharp
}~.
\end{flalign}
The left adjoint functor $\flat : \HH \to \HH$
assigns to an object $X$ in $\HH$
the object given by the following presheaf of groupoids $\flat X : \CC^\op\to\Grpd$:
To any object $U$ in $\CC$, it assigns $\flat X(U) = X(\bbR^0)$, i.e.\ the groupoid of global points of $X$,
and to any $\CC$-morphism $\ell : U\to U^\prime$ it assigns the identity morphism
$\flat X(\ell) = \id_{X(\bbR^0)}$. The action of $\flat$ on morphisms in $\HH$ is the obvious one.
Loosely speaking, $\flat X$ is something like a `discrete space' 
as it forgot all the smooth structure on $X$.
The right adjoint functor $\sharp : \HH\to \HH$
assigns to an object $X$ in $\HH$ the object $\sharp X$ defined as follows:
To any object $U$ in $\CC$, it assigns the groupoid
\begin{flalign}
\sharp X(U) = \prod_{p\in U} X(\bbR^0)~,
\end{flalign}
where the product goes over all points $p\in U$,
and to any $\CC$-morphism $\ell : U\to U^\prime$ it assigns the groupoid morphism
defined by universality of products and the commutative diagrams
\begin{flalign}
\xymatrix{
\ar[dr]_-{\pr_{\ell(p)}}\sharp X(U^\prime) \ar[rr]^-{\sharp X(\ell)} && \sharp X(U) \ar[dl]^-{\pr_{p}}\\
&X(\bbR^0)&
}~
\end{flalign}
for all points $p\in U$, where $\pr$ denote the projection $\Grpd$-morphisms 
associated to the products.
The action of $\sharp$ on morphisms in $\HH$ is the obvious one.
It is easy to prove that $\flat \dashv \sharp$ is a Quillen adjunction by using the explicit characterization
of fibrations in $\HH$ (in the local model structure) given by \cite[Proposition 4.2]{Hollander3},
see also Proposition \ref{propo:localfibration}.
The conceptual interpretation of $\sharp X$ is as follows:
For any object $U$ in $\CC$, there exist isomorphisms of groupoids
\begin{flalign}
\sharp X(U) \simeq \Grpd(\underline{U}, \sharp X) \simeq \Grpd(\flat \underline{U},X)~,
\end{flalign}
where we make use of the Yoneda lemma and the adjunction property of $\flat \dashv \sharp$. This shows that,
loosely speaking, the groupoid $\sharp X(U)$ is given by `evaluating' $X$ on 
the discrete space $\flat \underline{U}$. The key idea behind concretification 
is, again loosely speaking, to make use of the passage 
to discrete spaces $\flat \underline{U}$ to avoid gauge fields along the parameter spaces $U$.
\sk

Let us now focus on our explicit example. Consider the object $\sharp ({\BGcon}_{~}^{\widecheck{M}})$ 
in $\HH$, which is a stack because ${\BGcon}_{~}^{\widecheck{M}}$ is a stack and $\sharp$ 
is a right Quillen functor. The objects (respectively morphisms) of the groupoid  $\sharp ({\BGcon}_{~}^{\widecheck{M}})(U)$  
describe by construction families of gauge fields (respectively gauge transformations)
on $M$ that are labeled by the points  $p\in U$. In particular, there appear no gauge fields 
along the parameter space $U$. 
Unfortunately, there is no smoothness requirement on such families.
To solve this problem, consider the canonical $\HH$-morphism
\begin{subequations}\label{eqn:zetaall}
\begin{flalign}\label{eqn:zeta}
\zeta_{\mathrm{con}} \,:\, {\BGcon}_{~}^{\widecheck{M}}~\longrightarrow~ \sharp \big({\BGcon}_{~}^{\widecheck{M}}\big)~.
\end{flalign}
Explicitly, the stages 
$\zeta_{\mathrm{con}} : {\BGcon}_{~}^{\widecheck{M}}(U) \to \sharp ({\BGcon}_{~}^{\widecheck{M}})(U)$
are the following groupoid morphisms:
To any object $(\{A_V\},\{g_{VV^\prime}\})$ of the source groupoid (cf.\ Lemma \ref{lem:BGconcheckM}), 
it assigns
\begin{flalign}\label{eqn:zetaobjects}
\zeta_{\mathrm{con}}\big(\{A_V\},\{g_{VV^\prime}\}\big) = 
\Big( \big\{ (\id_V\times p)^\ast \, A_{V} \big\}  ,\big\{ (\id_{V\cap V^\prime} \times p)^\ast \,g_{VV^\prime} \big\}\Big)~,
\end{flalign}
where we regarded points $p\in U$ as $\CC$-morphisms $p : \bbR^0\to U$. 
To any morphism $\{h_V\}$ of the source groupoid (cf.\ Lemma \ref{lem:BGconcheckM}), 
it assigns
\begin{flalign}\label{eqn:zetamorphisms}
\zeta_{\mathrm{con}}\big(\{h_V\}\big) = \big\{ (\id_V\times p)^\ast\, h_{V} \big\}~.
\end{flalign}
\end{subequations}
Notice that the naive image of this groupoid morphism would solve our problem: It describes 
families of gauge fields and gauge transformations on $M$ 
which are smoothly parametrized by $U$
because they arise as pullbacks along the point embeddings 
$\id_M\times p : M\times \bbR^0 \to M\times U$ of smooth
gauge fields and gauge transformations on $M\times U$.
Unfortunately, the stage-wise naive image of $\zeta_{\mathrm{con}}$ is in general not a homotopically 
meaningful construction in the sense that it does not preserve weak equivalences.
We thus have to solve our problem in a more educated manner to ensure that its solution
is homotopically meaningful.
\sk

For this let us consider the following pullback diagram 
\begin{flalign}\label{eqn:pullbackconcret}
\xymatrix{
\sharp\big({\BGcon}_{~}^{\widecheck{M}}\big) \ar[rr]^-{\sharp (\mathrm{forget}^{\widecheck{M}})} ~&&~ \sharp\big({\BG}_{~}^{\widecheck{M}}\big) ~&&~\ar[ll]_-{\zeta} {\BG}_{~}^{\widecheck{M}}
}
\end{flalign}
in $\HH$. Here ${\BG}_{~}^{\widecheck{M}}$ is the mapping stack from $\widecheck{M}$ to the classifying
stack $\BG$ of principal $G$-bundles (cf.\ Example \ref{ex:BG}), $\zeta$ is a canonical $\HH$-morphism 
similar to \eqref{eqn:zetaall} and $\mathrm{forget} : {\BGcon} \to {\BG}$ 
is the $\HH$-morphism that forgets the connections. (The mapping stack ${\BG}_{~}^{\widecheck{M}}$ is
similar to ${\BGcon}_{~}^{\widecheck{M}}$, cf.\ Lemma \ref{lem:BGconcheckM},
however without the connection data.) Note that there exists a canonical $\HH$-morphism
\begin{flalign}
\xymatrix{
 {\BGcon}_{~}^{\widecheck{M}}\ar[r]~&~ \sharp\big({\BGcon}_{~}^{\widecheck{M}}\big)\times^h_{ \sharp({\BG}_{~}^{\widecheck{M}})} {\BG}_{~}^{\widecheck{M}}
 }
\end{flalign}
to the homotopy fiber product associated to  \eqref{eqn:pullbackconcret}.
We shall use the following abstract
\begin{defi}\label{defi:diffconcret}
The {\em differential concretification} of the mapping stack $ {\BGcon}_{~}^{\widecheck{M}}$,
which we will also call the {\em stack of gauge fields} on $M$, is defined by
\begin{flalign}
\GCon(M) ~:=~\Imm_1\Big(\xymatrix{
 {\BGcon}_{~}^{\widecheck{M}}\ar[r]~&~ \sharp\big({\BGcon}_{~}^{\widecheck{M}}\big)\times^h_{ \sharp({\BG}_{~}^{\widecheck{M}})} {\BG}_{~}^{\widecheck{M}}
 }\Big) ~,
\end{flalign}
where $\Imm_1$ denotes the $1$-image, 
i.e.\ the fibrant replacement in the $(-1)$-truncation of the canonical model structure on the over-category
$\HH\big/\sharp({\BGcon}_{~}^{\widecheck{M}})\times^h_{ \sharp({\BG}_{~}^{\widecheck{M}})} {\BG}_{~}^{\widecheck{M}}$,
see Appendix \ref{app:truncation}. 
\end{defi}

We shall now explicitly compute $\GCon(M)$ in order to confirm that our abstract definition
leads to the desired stack of gauge fields on $M$, i.e.\ the stack describing smoothly parametrized
gauge fields and gauge transformations on $M$. It is practically very convenient 
to compute instead of $\GCon(M)$ defined in Definition \ref{defi:diffconcret}
a weakly equivalent object of $\HH$ that has a simpler and more familiar explicit description.
\sk

As a first step towards a simplified description of $\GCon(M)$, 
we notice that we actually do not have to compute the homotopy fiber
product in Definition \ref{defi:diffconcret} by using the explicit construction 
of Proposition \ref{propo:homfibprod}. The reason is as follows:
Using Proposition \ref{propo:localfibration} it is easy to prove that the morphism 
$\mathrm{forget} : {\BGcon} \to {\BG}$ is a fibration (in the local model structure on $\HH$).
Because $(-)^{\widecheck{M}} : \HH\to \HH$ and $\sharp : \HH\to \HH$ are right Quillen functors,
it follows that the right-pointing arrow in the pullback diagram \eqref{eqn:pullbackconcret} is 
a fibration too. Using that $\HH$ is a right proper model category (cf.\ \cite[Corollary 5.8]{Hollander}),
the statement in \cite[Corollary 13.3.8]{Hirschhorn} implies that the canonical $\HH$-morphism
\begin{flalign}
\xymatrix{
P ~:=~  \sharp\big({\BGcon}_{~}^{\widecheck{M}}\big)\times^{~}_{ \sharp({\BG}_{~}^{\widecheck{M}})} {\BG}_{~}^{\widecheck{M}} ~~\ar[rr]  &&~~\sharp\big({\BGcon}_{~}^{\widecheck{M}}\big)\times^h_{ \sharp({\BG}_{~}^{\widecheck{M}})} {\BG}_{~}^{\widecheck{M}}
}
\end{flalign}
from the {\em ordinary} fiber product to the homotopy fiber product is a weak equivalence.
Hence, we may replace the homotopy fiber product in  Definition \ref{defi:diffconcret} 
by the ordinary fiber product $P$ in order to find a weakly equivalent description of $\GCon(M)$. 
The ordinary fiber product $P$ is much easier to compute: For any object $U$ in $\CC$,
the groupoid $P(U)$ has as objects all pairs of families
\begin{flalign}
\Big(\big \{A_{p;\,V} \in \Omega^1(V,\g)\big\} , \big\{g_{VV^\prime} \in C^\infty((V\cap V^\prime)\times U,G)\big\}\Big)~,
\end{flalign}
where $V \subseteq M$, $V^\prime \subseteq M$ run over all open subsets diffeomorphic to $\bbR^{\dim(M)}$ and $p\in U$ runs over all points of $U$,
such that $(\{A_{p;\,V} \} , \{ (\id_{V\cap V^\prime} \times p)^\ast\, g_{VV^\prime}\})$
is an object in $ \sharp({\BGcon}_{~}^{\widecheck{M}})(U)$.
The morphisms in $P(U)$ from $( \{A_{p;\,V}\} , \{g_{VV^\prime} \})$
to $( \{A^\prime_{p;\,V}\} , \{g^\prime_{VV^\prime} \})$
are all families
\begin{flalign}
\big\{h_{V} \in C^\infty(V \times U,G)\big\}~,
\end{flalign}
where $V \subseteq M$ runs over all open subsets diffeomorphic to $\bbR^{\dim(M)}$,
such that $\{(\id_V\times p)^\ast\, h_{V}\}$ is a morphism
in $ \sharp({\BGcon}_{~}^{\widecheck{M}})(U)$
from $(\{A_{p;\,V} \} , \{ (\id_{V\cap V^\prime} \times p)^\ast\, g_{VV^\prime}\})$
to $(\{A^\prime_{p;\,V} \} , \{ (\id_{V\cap V^\prime} \times p)^\ast\, g^\prime_{VV^\prime}\})$.
For a morphism $\ell : U\to U^\prime$ in $\CC$, the associated
$\Grpd$-morphism $P(\ell ) : P(U^\prime) \to P(U)$
is given by
\begin{flalign}
\nn P(\ell) ~:~ \big( \big\{A_{p^\prime;\,V}\big\} , \big\{g_{VV^\prime} \big\}\big) ~ &\longmapsto ~
 \big( \big\{  \,A_{\ell(p);\,V} \big\} , \big\{(\id_{V\cap V^\prime}\times \ell)^\ast\,g_{VV^\prime} \big\}\big)~,\\
 \big\{h_{V} \big\} ~&\longmapsto ~ \big\{(\id_V\times \ell)^\ast\,h_{V}\big\}~.
\end{flalign}
As a side remark, notice that $P(U)$ has the desired morphisms, however 
the gauge fields are not smoothly $U$-parametrized yet.
\sk

As a further simplification in our explicit description of $\GCon(M)$, 
we may combine Proposition \ref{propo:easyfibrepH/K} and Proposition \ref{propo:basechange}
to compute (again up to weak equivalence) the $1$-image in 
Definition \ref{defi:diffconcret} by the full image sub-presheaf of groupoids corresponding 
to the canonical morphism $ {\BGcon}_{~}^{\widecheck{M}}\to P$ to the {\em ordinary} fiber product $P$.
This defines the following object in $\HH$, which is weakly equivalent to $\GCon(M)$ given
in Definition \ref{defi:diffconcret}.
\begin{propo}\label{propo:GCon}
Up to weak equivalence in $\HH$, the stack of gauge fields $\GCon(M)$
defined in Definition \ref{defi:diffconcret} has the following explicit description:
For all objects $U$ in $\CC$, the objects of the groupoid $\GCon(M)(U)$ are pairs of families
\begin{flalign}
\big(\AAA,\PPP\big) := \Big(\big\{A_{V} \in \Omega^{1,0}(V\times U,\g)\big\},\big\{ g_{VV^\prime}\in C^\infty((V\cap V^\prime)\times U,G) \big\}\Big)~,
\end{flalign}
where  $V\subseteq M$ and $V^\prime \subseteq M$ run over all 
open subsets diffeomorphic to $\bbR^{\dim(M)}$ 
and $\Omega^{1,0}$ denotes the vertical $1$-forms on $V\times U\to U$,
which satisfy the following  conditions: 
\begin{itemize}
\item For all open subsets $V\subseteq M$ and $V^\prime \subseteq M$ diffeomorphic to $\bbR^{\dim(M)}$,
\begin{flalign}
A_{V^\prime}\vert_{(V\cap V^\prime)\times U}^{}~=~
A_{V}\vert_{(V\cap V^\prime)\times U}^{} \ra^{\ver} g_{VV^\prime}~,
\end{flalign}
where $A \ra^{\ver} g := g^{-1}\, A\, g + g^{-1}\, \dd^{\ver} g$ 
is the vertical action of gauge transformations that is
defined by the vertical de Rham differential $\dd^{\ver}$  on $V\times U\to U$.

\item  For all open subsets $V\subseteq M$ diffeomorphic to $\bbR^{\dim(M)}$, $g_{VV} = e$.

\item  For all open subsets $V\subseteq M$, $V^\prime \subseteq M$ 
and $V^{\prime\prime} \subseteq M$ diffeomorphic to $\bbR^{\dim(M)}$,
\begin{flalign}
g_{VV^\prime}\vert_{(V\cap V^\prime\cap V^{\prime\prime})\times U}^{}
~~g_{V^\prime V^{\prime\prime}}\vert_{(V\cap V^\prime\cap V^{\prime\prime})\times U}^{}
~=~
g_{VV^{\prime\prime}}\vert_{(V\cap V^\prime\cap V^{\prime\prime})\times U}^{}~.
\end{flalign}
\end{itemize}
The morphisms of the groupoid $\GCon(M)(U)$ from $(\AAA,\PPP) =(\{A_V\},\{g_{VV^\prime}\})$ 
to  $(\AAA^\prime,\PPP^\prime) = (\{A^\prime_V\},\{g^\prime_{VV^\prime}\})$ are families
\begin{flalign}
\hhh:= \big\{h_{V} \in C^\infty(V\times U,G)\big\} \,:\, \big(\AAA,\PPP\big)~\longrightarrow ~\big(\AAA^\prime,\PPP^\prime\big)~,
\end{flalign}
where $V\subseteq M$ runs over all open subsets diffeomorphic to $\bbR^{\dim(M)}$,
which satisfy the following conditions:
\begin{itemize}
\item For all open subsets $V\subseteq M$ diffeomorphic to $\bbR^{\dim(M)}$,
\begin{flalign}
A_{V}^\prime ~=~ A_V \ra^{\ver} h_V~.
\end{flalign}

\item For all open subsets $V\subseteq M$ and $V^\prime \subseteq M$ diffeomorphic to $\bbR^{\dim(M)}$,
\begin{flalign}
g_{VV^\prime}^\prime ~=~ h_V^{-1}\vert^{}_{(V\cap V^\prime)\times U}~~g_{VV^\prime}~~h_{V^\prime}\vert^{}_{(V\cap V^\prime)\times U}~.
\end{flalign}
\end{itemize}
For all morphisms $\ell : U\to U^\prime$ in $\CC$,  the groupoid morphism 
$\GCon(M)(\ell) : \GCon(M)(U^\prime) \to \GCon(M)(U)$ is given by
\begin{flalign}
\nn\GCon(M)(\ell)~:~ \big(\AAA,\PPP\big) ~&\longmapsto~ \big(\ell^\ast\AAA,\ell^\ast\PPP\big) := \big(\{(\id_V\times \ell)^\ast \,A_V\} , \{(\id_{V\cap V^\prime}\times\ell)^\ast\,g_{VV^\prime}\}\big)~,\\
\hhh ~&\longmapsto ~ \ell^\ast\hhh := \{(\id_{V}\times\ell)^\ast\, h_{V}\}~.
\end{flalign}
\end{propo}
\begin{defi}\label{defi:Ugaugefield}
Let $U$ be an object in $\CC$. We call an object $(\AAA,\PPP)$ 
of the groupoid $\GCon(M)(U)$ a {\em smoothly $U$-parametrized gauge field} on $M$.
More precisely, we call $\PPP$ a {\em smoothly $U$-parametrized principal $G$-bundle} on $M$ and
$\AAA$ a {\em smoothly $U$-parametrized connection} on $\PPP$.
A morphism $\hhh :(\AAA,\PPP)\to(\AAA^\prime,\PPP^\prime)$ 
of the groupoid  $\GCon(M)(U)$ is called a {\em smoothly $U$-parametrized gauge transformation}.
\end{defi}
\begin{rem}\label{rem:concmappinstackfunctor}
Similarly to Remark \ref{rem:mappinstackfunctor},
our explicit model for the stack of gauge fields 
presented in Proposition \ref{propo:GCon}
provides a functor $\GCon : \Man_{\hookrightarrow}^{\op}\to \HH$
on the category of manifolds with open embeddings. Given any open embedding 
$f : M\to M^\prime$, we also use the simplified notation 
\begin{flalign}
\nn \GCon(f)\big( \AAA,\PPP \big)&:= \big(f^\ast\AAA,f^\ast\PPP\big) := \big(\{ (f\vert_V\times \id_U)^\ast \,A_{f(V)}\},\{(f\vert_{V\cap V^\prime}\times \id_U)^\ast \,g_{f(V)f(V^\prime)}\}\big)~,\\
 \GCon(f)\big(\hhh\big) &:= f^\ast\hhh := \{(f\vert_V\times \id_U)^\ast\, h_{f(V)}\}~,
\end{flalign}
for any stage $U$ in $\CC$ of the $\HH$-morphism $\GCon(f) : \GCon(M^\prime) \to \GCon(M)$.
\end{rem}
\begin{rem}\label{rem:GConisstack}
The abstract description of $\GCon(M)$ provided in Definition \ref{defi:diffconcret}
is automatically a stack, i.e.\ a fibrant object in $\HH$. This is a consequence of the following facts:
1.)\ All three objects entering the homotopy fiber product are stacks, hence the homotopy fiber product
is a stack too, see Proposition \ref{propo:homfibprod}. 
2.)\ $\GCon(M)$ is constructed by a fibrant replacement in the 
$(-1)$-truncation of the corresponding over-category 
(see Appendix \ref{app:truncation} for the relevant terminology), hence 
the induced $\HH$-morphism from $\GCon(M)$ to the homotopy fiber product is 
a fibration in $\HH$. (Recall that fibrant objects in the $(-1)$-truncation 
of the over-category, i.e.\ $S_{-1}$-local objects, are by Definition \ref{defi:S-1local} 
in particular fibrations in $\HH$.) Combining
1.)\ and 2.), it follows that $\GCon(M)$ is a stack because the morphism $\GCon(M) \to \{\ast\}$ to
the terminal object factorizes over the homotopy fiber product in terms of two 
fibrations in $\HH$, and thus is a fibration too.
\sk

This unfortunately {\em does not} automatically imply that our weakly equivalent simplified
description of $\GCon(M)$ given in Proposition \ref{propo:GCon} is a stack too.
One can, however, verify explicitly that the stack condition \eqref{eqn:stack} 
holds true for the presheaf of groupoids $\GCon(M)$ presented in Proposition \ref{propo:GCon}.
This is a straightforward, but rather tedious, calculation using Proposition \ref{propo:descentcat} 
to compute the relevant homotopy limits. As this calculation is not very 
instructive, we shall not write it out in full detail and just mention that it uses arguments similar to those in
Example \ref{ex:BGcon}. (In particular, it uses that all cocycles on manifolds diffeomorphic to 
some $\bbR^n$ may be trivialized, and that $G$-valued functions and $\g$-valued forms are sheaves.)
\end{rem}

We introduce the following
notation for the mapping stack $\BG^{\widecheck{M}}$ in 
order to match the notations established 
in Proposition \ref{propo:GCon} and Definition \ref{defi:Ugaugefield}.
\begin{defi}\label{defi:GBun}
We call $\GBun(M):= \BG^{\widecheck{M}}$ the {\em stack of principal $G$-bundles} on $M$.
For any object $U$ in $\CC$, we denote the objects of $\GBun(M)(U)$ by symbols
like $\PPP := \{g_{VV^\prime}\}$ and the morphisms of $\GBun(M)(U)$ by symbols
like $\hhh :=\{h_V\} : \PPP\to \PPP^\prime$. 
\end{defi}

\subsection{\label{subsec:curvature}Curvature morphism}
Recalling the definition of the classifying stack $\BG_{\Omega^2_{\ad}}$ of 
principal $G$-bundles with adjoint bundle valued $2$-forms (cf.\ Example \ref{ex:Liealgforms}), 
we consider the canonical curvature $\HH$-morphism 
\begin{flalign}
F:\BGcon \longrightarrow \BG_{\Omega^2_{\ad}}~,
\end{flalign}
whose stages $F: \BGcon(U) \to \BG_{\Omega^2_{\ad}}(U)$ are
the groupoid morphisms given by
\begin{flalign}
\nn F\,:\qquad A~&\longmapsto~ F(A) := \dd A + A \wedge A~,\\
 (A,g)~&\longmapsto ~(F(A),g)~.
\end{flalign}
Given any manifold $M$, we have an induced $\HH$-morphism 
$F^{\widecheck{M}}:{\BGcon}^{\widecheck{M}} \to {\BG_{\Omega^2_{\ad}}}^{\widecheck{M}}$ 
between mapping stacks. We shall now construct the
counterpart of this morphism on concretified mapping stacks.
This will later be used to formalize the Yang-Mills equation on $\GCon(M)$. 
\sk

As a first step, we explicitly describe the mapping stack
from $\widecheck{M}$ to $\BG_{\Omega^p_{\ad}}$, $p\in\bbZ_{\geq0}$, following the lines
of Section \ref{subsec:mapstack}. Up to weak equivalence in $\HH$, the associated presheaf of groupoids
\begin{flalign}
{\BG_{\Omega^p_{\ad}}}_{~}^{\widecheck{M}} : \CC^\op \longrightarrow \Grpd~
\end{flalign} 
has a description similar to Lemma \ref{lem:BGconcheckM}: For any object $U$ in $\CC$,
the objects of the groupoid ${\BG_{\Omega^p_{\ad}}}^{\widecheck{M}}(U)$ are pairs of families 
\begin{flalign}
\Big(\big\{ \omega_V  \in\Omega^p(V\times U,\g)  \big\},
\big\{ g_{VV^\prime} \in C^\infty((V\cap V^\prime)\times U , G)\big\}\Big)~,
\end{flalign}
where $V\subseteq M$ and $V^\prime \subseteq M$ run over all 
open subsets diffeomorphic to $\bbR^{\dim(M)}$, which satisfy
\begin{subequations}\label{eqn:formcomp}
\begin{flalign}
\omega_{V^\prime}\vert_{(V\cap V^\prime)\times U}^{} ~&=~ 
\ad_{g_{VV^\prime}} \big(\omega_V\vert_{(V\cap V^\prime)\times U}^{}\big)~, \\
 g_{VV} ~&=~e~, \\
g_{VV^\prime}\vert_{(V\cap V^\prime\cap V^{\prime\prime})\times U}^{}~~
g_{V^\prime V^{\prime\prime}}\vert_{(V\cap V^\prime\cap V^{\prime\prime})\times U}^{}
~&=~ g_{VV^{\prime\prime}}\vert_{(V\cap V^\prime\cap V^{\prime\prime})\times U}^{}~,
\end{flalign}
\end{subequations}
for all open subsets $V\subseteq M$, $V^\prime \subseteq M$ and $V^{\prime\prime} \subseteq M$ 
diffeomorphic to $\bbR^{\dim(M)}$. 
The morphisms of the groupoid ${\BG_{\Omega^p_{\ad}}}^{\widecheck{M}}(U)$ from 
$(\{\omega_V\},\{g_{VV^\prime}\})$ to $(\{\omega_{V}^\prime\}, \{g^\prime_{VV^\prime}\})$ 
are families
\begin{flalign}
\big\{ h_{V} \in C^\infty(V\times U,G)  \big\}~,
\end{flalign}
where $V\subseteq M$ runs over all open subsets  diffeomorphic to $\bbR^{\dim(M)}$,
which satisfy
\begin{subequations}\label{eqn:formcomp2}
\begin{flalign}
\omega_{V}^\prime ~&=~ \ad_{h_V} (\omega_V) ~,\\
g_{VV^\prime}^\prime ~&=~ h_V^{-1}\vert^{}_{(V\cap V^\prime)\times U}~~g_{VV^\prime}~~h_{V^\prime}\vert^{}_{(V\cap V^\prime)\times U}~,
\end{flalign} 
\end{subequations}
for all open subsets $V\subseteq M$ and $V^\prime \subseteq M$ diffeomorphic to $\bbR^{\dim(M)}$.
Moreover, the functor ${\BG_{\Omega^p_{\ad}}}^{\widecheck{M}}: \CC^\op \to \Grpd$ 
assigns to a morphism $\ell: U \to U^\prime$ in $\CC$ the groupoid morphism 
${\BG_{\Omega^p_{\ad}}}^{\widecheck{M}}(\ell): {\BG_{\Omega^p_{\ad}}}^{\widecheck{M}}(U^\prime) 
\to {\BG_{\Omega^p_{\ad}}}^{\widecheck{M}}(U)$ specified by
\begin{flalign}
\nn {\BG_{\Omega^p_{\ad}}}_{~}^{\widecheck{M}}(\ell) ~:~ \big(\{\omega_V\},\{g_{VV^\prime}\}\big) ~&\longmapsto~ \big(\{(\id_V\times \ell)^\ast \,\omega_V\} , \{(\id_{V\cap V^\prime}\times\ell)^\ast\,g_{VV^\prime}\}\big)~,\\
\{h_{V} \} ~&\longmapsto ~ \{(\id_{V}\times\ell)^\ast\, h_{V}\}~.
\end{flalign}

The stages of the $\HH$-morphism 
\begin{flalign}
F^{\widecheck{M}} \, :\, {\BGcon}^{\widecheck{M}} ~\longrightarrow~ {\BG_{\Omega^2_{\ad}}}_{~}^{\widecheck{M}}~
\end{flalign}
are by construction the groupoid morphisms from 
${\BGcon}^{\widecheck{M}}(U) = \Grpd(\underline{U}\times \widecheck{M}, \BGcon)$ 
to ${\BG_{\Omega^2_{\ad}}}^{\widecheck{M}}(U) = \Grpd(\underline{U}\times \widecheck{M}, \BG_{\Omega^2_{\ad}})$ 
resulting from post-composition by $F: \BGcon \to \BG_{\Omega^2_{\ad}}$. 
More explicitly, using the description of ${\BGcon}_{~}^{\widecheck{M}}(U)$ 
given in Lemma \ref{lem:BGconcheckM} and the one of ${\BG_{\Omega^2_{\ad}}}^{\widecheck{M}}$ provided above, 
we find that $F^{\widecheck{M}}:{\BGcon}^{\widecheck{M}}(U) \to 
{\BG_{\Omega^2_{\ad}}}^{\widecheck{M}}(U)$ is the groupoid morphism 
\begin{flalign}
\nn F^{\widecheck{M}} ~:~ \big(\{A_V\},\{g_{VV^\prime}\}\big) ~&\longmapsto~ \big(\{F(A_V)\} , \{g_{VV^\prime}\}\big)~,\\
\{h_{V} \} ~&\longmapsto ~ \{h_{V}\}~,
\end{flalign}
for all objects $U$ in $\CC$.
It is easy to confirm directly  that 
$F^{\widecheck{M}}: {\BGcon}^{\widecheck{M}} \to {\BG_{\Omega^2_{\ad}}}^{\widecheck{M}}$ 
is a natural transformation between functors $\CC^\op \to \Grpd$. 
\sk

Our concretification prescription for ${\BGcon}_{~}^{\widecheck{M}}$ (cf.\ Definition \ref{defi:diffconcret})
can be also applied to the mapping stacks 
${\BG_{\Omega^p_{\ad}}}^{\widecheck{M}}$, $p\in\bbZ_{\geq 0}$. The relevant pullback diagram (replacing
\eqref{eqn:pullbackconcret}) is
\begin{flalign}
\xymatrix{
\sharp\big({\BG_{\Omega^p_{\ad}}}_{~}^{\widecheck{M}}\big) \ar[rr]^-{\sharp (\mathrm{forget}^{\widecheck{M}})} ~&&~ \sharp\big({\BG}_{~}^{\widecheck{M}}\big) ~&&~\ar[ll]_-{\zeta} {\BG}_{~}^{\widecheck{M}}
}~,
\end{flalign}
where now the $\HH$-morphism $\mathrm{forget} : {\BG_{\Omega^p_{\ad}}} \to \BG$ is the 
one which forgets the adjoint bundle valued $p$-forms.
We define in analogy to Definition \ref{defi:diffconcret} the concretification of ${\BG_{\Omega^p_{\ad}}}^{\widecheck{M}}$
abstractly by
\begin{flalign}
\bfOmega^p(M,\ad(G)) ~:=~\Imm_1\Big(\xymatrix{
 {\BG_{\Omega^p_{\ad}}}_{~}^{\widecheck{M}}\ar[r]~&~ \sharp\big({\BG_{\Omega^p_{\ad}}}_{~}^{\widecheck{M}}\big)\times^h_{ \sharp({\BG}_{~}^{\widecheck{M}})} {\BG}_{~}^{\widecheck{M}}
 }\Big) ~.
\end{flalign}
Similarly to Proposition \ref{propo:GCon}, 
the concretified mapping stack $\bfOmega^p(M,\ad(G))$ has the following
explicit description, up to weak equivalence in $\HH$.  For any object $U$ in $\CC$,
the objects of the groupoid $\bfOmega^p(M,\ad(G))(U)$ are pairs of families 
\begin{flalign}\label{eqn:bfOmegaobj}
\big(\bfomega,\PPP\big) := \Big(\big\{\omega_{V} \in \Omega^{p,0}(V\times U,\g)\big\},\big\{ g_{VV^\prime}\in C^\infty((V\cap V^\prime)\times U,G) \big\}\Big)~,
\end{flalign}
where $V\subseteq M$ and $V^\prime \subseteq M$ run over all open subsets
diffeomorphic to $\bbR^{\dim(M)}$ and $\Omega^{p,0}$ denotes the vertical $p$-forms on $V\times U\to U$,
which satisfy the conditions in \eqref{eqn:formcomp}. The morphisms of the groupoid 
$\bfOmega^p(M,\ad(G))(U)$ from $(\bfomega,\PPP)= (\{\omega_V\},\{g_{VV^\prime}\})$ to 
$(\bfomega^\prime,\PPP^\prime) = (\{\omega_{V}^\prime\}, \{g^\prime_{VV^\prime}\})$ are families 
\begin{flalign}\label{eqn:bfOmegamor}
\hhh:= \big\{ h_{V} \in C^\infty(V\times U,G)  \big\}~:~\big(\bfomega,\PPP\big)~\longrightarrow~\big(\bfomega^\prime,\PPP^\prime\big)~,
\end{flalign}
where $V\subseteq M$ runs over all open subsets  diffeomorphic to $\bbR^{\dim(M)}$,
which satisfy the conditions in \eqref{eqn:formcomp2}. The groupoid morphisms
$\bfOmega^p(M,\ad(G))(\ell) $  associated to $\CC$-morphisms $\ell :U\to U^\prime$
are analogous to the ones in Proposition \ref{propo:GCon}. Similarly to Remark \ref{rem:GConisstack},
one can verify that this simplified, weakly equivalent description
of $\bfOmega^p(M,\ad(G))$  is a stack too, for all $p\in\bbZ_{\geq 0}$.
\sk

From our abstract concretification prescription and the $\HH$-morphism
$F^{\widecheck{M}} :{\BGcon}^{\widecheck{M}} 
\to {\BG_{\Omega^2_{\ad}}}^{\widecheck{M}}$ we obtain an $\HH$-morphism
\begin{flalign}
\FF_M\,:\, \GCon(M)~ \longrightarrow~ \bfOmega^2(M,\ad(G))~
\end{flalign}
between the concretified mapping stacks.
Explicitly, using our simplified, but weakly equivalent, models for both concretified mapping stacks, 
the stages  $\FF_M: \GCon(M)(U) \to \bfOmega^2(M,\ad(G))(U)$ 
of this $\HH$-morphism read as
\begin{flalign}\label{eqn:curvaturemorph}
\nn \FF_M ~: ~~\big(\AAA,\PPP\big) = \big(\{A_V\} , \{g_{VV^\prime}\}\big) ~&\longmapsto~ \big(\{F^{\ver}(A_V)\} , \{g_{VV^\prime}\}\big)~,\\
\big(\hhh : (\AAA,\PPP)\to(\AAA^\prime,\PPP^\prime)\big) ~&\longmapsto ~ \big(\hhh:  \FF_M(\AAA,\PPP)\to  \FF_M(\AAA^\prime,\PPP^\prime)\big)~,
\end{flalign}
where $F^{\ver}(A_V) := \dd^{\ver} A_V + A_V \wedge A_V$ is the vertical curvature 
defined by the vertical de Rham differential $\dd^{\ver}$ on $V\times U\to U$. 
\begin{rem}\label{rem:conccurvaturenat}
Similar to Remark \ref{rem:concmappinstackfunctor} it is easy to see that
$\bfOmega^p(-,\ad(G)) : \Man_{\hookrightarrow}^{\op}\to \HH$ is a functor, for all $p\in\bbZ_{\geq0}$.
It is also easy to check that the $\HH$-morphisms $\FF_M: \GCon(M) \to \bfOmega^2(M,\ad(G))$ constructed
above are the components of a natural transformation
\begin{flalign}
\FF \,:\, \GCon ~\longrightarrow~ \bfOmega^2(-,\ad(G))
\end{flalign}
between functors from $\Man_{\hookrightarrow}^{\op}$ to $\HH$.
\end{rem}


\section{\label{sec:YM}Yang-Mills equation}
We formalize the Yang-Mills equation on
globally hyperbolic Lorentzian manifolds in terms of a morphism between 
concretified mapping stacks. The corresponding stack of Yang-Mills solutions 
is defined by an appropriate homotopy fiber product and it will be worked out explicitly, 
up to weak equivalence in $\HH$.
Our constructions are functorial on the usual category $\Loc_m$
of oriented and time-oriented globally hyperbolic Lorentzian manifolds
(of a fixed but arbitrary dimension $m\geq 2$), which we shall review below.

\subsection{\label{subsec:spacetimes}Globally hyperbolic Lorentzian manifolds}
Spacetimes in physics are described by globally hyperbolic Lorentzian manifolds,
see \cite[Chapter 3]{BEE}, \cite[Chapter 14]{ONeill} 
and also \cite[Section 1.3]{Baer} for a more concise introduction.
Recall that a Lorentzian manifold is a manifold $M$ that is equipped with a metric
of Lorentzian signature $(-++\cdots +)$.
We further assume our Lorentzian manifolds to be equipped with an orientation
and a time-orientation, and that they are
of a fixed but arbitrary dimension $m\geq 2$. 
For notational simplicity, we denote oriented and time-oriented 
Lorentzian manifolds by symbols like $M$, i.e.\ we
suppress the orientation, time-orientation and metric from our notation.
\sk

A Cauchy surface in a Lorentzian manifold $M$ 
is a subset $\Sigma\subseteq M$ such that 
every inextensible timelike curve in $M$ meets $\Sigma$ exactly once.
A Lorentzian manifold that admits a Cauchy surface 
is called globally hyperbolic.
Globally hyperbolic Lorentzian manifolds $M$ provide us with a suitable 
geometric framework to study hyperbolic partial differential equations,
whose initial data are assigned on a spacelike Cauchy surface $\Sigma\subseteq M$.
Later we will also need an equivalent characterization 
of global hyperbolicity given in terms of strong causality 
and compactness of double-cones. Concretely, this means that 
every point admits a basis of causally convex open neighborhoods and 
that $J^+_M(p) \cap J^-_M(p^\prime)$ is compact, for all $p,p^\prime \in M$. 
Recall that a subset $S \subseteq M$ is called causally convex 
if $J^+_M(S) \cap J^-_M(S) \subseteq S$ 
(in other words, causal curves with endpoints in $S$ lie entirely in it), 
where $J^\pm_M(S)\subseteq M$ denotes the causal future/past of $S\subseteq M$.
(This is the subset consisting of all points of $M$ that can be reached by a future/past-directed 
piecewise smooth causal curve emanating from $S$.)
\sk

Let us denote by $\Loc_m$ the category of all
$m$-dimensional oriented and time-oriented globally hyperbolic Lorentzian manifolds $M$
with morphisms given by all causal embeddings $f : M\to M^\prime$ (see \cite[Section 2.2]{SPAS}).
Explicitly, the latter are orientation and time-orientation preserving isometric embeddings, 
whose image is a causally convex open subset of $M^\prime$. 
\sk

Let $\mathcal{V} = \{V_\alpha\subseteq M\}$ be an open cover of (the manifold underlying) an object 
$M$ in $\Loc_m$. Then each $V_\alpha$ may be equipped with 
the induced Lorentzian metric, orientation and time-orientation,
however it is not necessarily globally hyperbolic and hence not necessarily an object in $\Loc_m$.
This motivates us to introduce the notion of {\em causally convex open covers}
of objects $M$ in $\Loc_m$, which are open covers $\mathcal{V} = \{V_\alpha\subseteq M\}$
such that each $V_\alpha$ is causally convex. 
By the characterization of global hyperbolicity in terms 
of strong causality and compactness of double-cones, it follows 
that each $V_\alpha$ is an object of $\Loc_m$ when equipped 
with the induced Lorentzian metric, orientation and time-orientation,
and that the canonical inclusion $\rho_\alpha : V_\alpha\to M$ is a $\Loc_m$-morphism.
Moreover, if non-empty, $V_\alpha\cap V_\beta$ is a causally convex open subset
of $M$ (as well as of $V_\alpha$ and of $V_\beta$) and thus may be regarded as an object in $\Loc_m$.
The canonical inclusion of $V_\alpha\cap V_\beta\neq \emptyset$ into 
$M$ (as well as the ones into $V_\alpha$ and $V_\beta$) is a $\Loc_m$-morphism.
A similar statement holds true for all higher intersections $V_{\alpha_1}\cap \cdots\cap V_{\alpha_n}\neq \emptyset$.
\sk

Similarly to Corollary \ref{cor:functorialcofibrepMan},
there exists a canonical cofibrant replacement in $\HH$ 
for objects in $\Loc_m$ which makes use of causally convex open covers. 
For $M$ an object in $\Loc_m$, we define 
\begin{flalign}\label{eqn:VMcc}
\mathcal{V}_M^{\mathrm{cc}}:= \Big\{V\subseteq M\,:\, V\text{ open, causally convex and diffeomorphic to }\bbR^m\Big\}
\end{flalign}
to be the set of {\em all} causally convex open subsets of $M$ that are diffeomorphic 
to $\bbR^m$.\footnote{It is easy to see that \eqref{eqn:VMcc} indeed defines a cover of $M$:
Given any point $p\in M$, choose a Cauchy surface $\Sigma \subseteq M$ such that
$p\in \Sigma$. Taking any open neighborhood $W\subseteq \Sigma$ of $p$ 
that is diffeomorphic to $\bbR^{m-1}$, the Cauchy development $D(W)\subseteq M$ 
contains $p$ and belongs to $\mathcal{V}_M^{\mathrm{cc}}$ 
(in particular, $D(W)$ is diffeomorphic to $\bbR \times W$ because $W$ is one of its Cauchy surfaces). 
See \cite[Definition 14.35 and Lemma 14.43]{ONeill} for the relevant
statements from Lorentzian geometry that we used in this argument.}
We denote the corresponding presheaf of \v{C}ech groupoids by 
$\widecheck{M}{}^{\mathrm{cc}}:= \widecheck{\mathcal{V}_M^{\mathrm{cc}}}$
and note that by the same arguments as in Appendix \ref{app:cofibrep} (see in particular 
Proposition \ref{propo:cofibcover} and Corollary \ref{cor:functorialcofibrepMan})
this defines a functorial cofibrant replacement in $\HH$
of objects in $\Loc_m$. The corresponding functor and
natural transformation will be denoted by
\begin{flalign}
\widecheck{(-)}{}^{\mathrm{cc}} \,:\, \Loc_m\longrightarrow \HH\quad,\qquad q^{\mathrm{cc}} \,:\,\widecheck{(-)}{}^{\mathrm{cc}} \longrightarrow \underline{(-)}~.
\end{flalign}
When working with objects $M$ in $\Loc_m$,
we may use these causally convex cofibrant replacements to 
compute the mapping stacks and their concretifications 
presented in Section \ref{sec:gaugefields} up to a weak equivalence in $\HH$.
This is a consequence of
\begin{propo}\label{propo:cccovervscover}
For all objects $M$ in $\Loc_m$, 
there exists a natural weak equivalence
$\widecheck{M}{}^{\mathrm{cc}} \to \widecheck{M}$ in $\HH$
from the causally convex cofibrant replacement to the one of 
Corollary \ref{cor:functorialcofibrepMan}.
\end{propo}
\begin{proof}
Notice that $\mathcal{V}_M^{\mathrm{cc}}$ given in \eqref{eqn:VMcc}
is a sub-cover of the cover 
\begin{flalign}
\mathcal{V}_M= \big\{V\subseteq M \,:\, V\text{ open and diffeomorphic to }\bbR^m  \big\}
\end{flalign}
used in Appendix \ref{app:cofibrep}. Hence, there exists a natural
$\HH$-morphism $\widecheck{M}{}^{\mathrm{cc}} = \widecheck{\mathcal{V}_M^{\mathrm{cc}}} \to 
\widecheck{\mathcal{V}_M} = \widecheck{M}$ between the associated presheaves
of \v{C}ech groupoids. Explicitly, the stage at $U$ in $\CC$ is specified by
\begin{flalign}
 \Big(\xymatrix{
M & \ar[l]_-{\rho_V}V & \ar[l]_-{\nu}U
}\Big) ~~\longmapsto~~\Big(\xymatrix{
M & \ar[l]_-{\rho_V}V & \ar[l]_-{\nu}U
}\Big)~,
\end{flalign}
where $V\subseteq M$ runs over all causally convex open subsets 
diffeomorphic to $\bbR^m$.
The commutative diagram
\begin{flalign}
\xymatrix{
\ar[rd]_-{q_M^{\mathrm{cc}}}\widecheck{M}{}^{\mathrm{cc}} \ar[rr]^-{}&& \widecheck{M}\ar[dl]^-{q_M}\\
&\underline{M}&
}
\end{flalign}
in $\HH$ and the fact that both $q_M$ and $q_M^\mathrm{cc}$ are weak equivalences
implies via the 2-out-of-3 property of weak equivalences 
that $\widecheck{M}{}^{\mathrm{cc}} \to \widecheck{M}$ is a weak equivalence too.
\end{proof}
\begin{cor}\label{cor:cccofibrep}
Let $Y$ be a stack in $\HH$ and $M$ an object in $\Loc_m$.
Then the natural weak equivalence 
$\widecheck{M}{}^{\mathrm{cc}} \to \widecheck{M}$
of Proposition \ref{propo:cccovervscover} induces a
natural weak equivalence $Y^{\widecheck{M}} \to Y^{\widecheck{M}{}^{\mathrm{cc}}}$
of mapping stacks. Hence, up to natural 
weak equivalence, all mapping stacks and their concretifications in Section \ref{sec:gaugefields} 
can be computed using the cofibrant replacement $\widecheck{M}{}^{\mathrm{cc}}$ instead of $\widecheck{M}$.
\end{cor}
\begin{rem}\label{rem:cccofibrep}
In practice, the statement of Corollary \ref{cor:cccofibrep} means that
when working on an object $M$ in $\Loc_m$, all open subsets 
$V,V^\prime,V^{\prime\prime}\subseteq M$ diffeomorphic to $\bbR^m$ 
entering our explicit models for mapping stacks and their concretifications
(see e.g.\ Proposition \ref{propo:GCon}) may be assumed to be also causally convex.
We shall use this weakly equivalent description from now on.
\end{rem}

\subsection{\label{subsec:YMmap}Yang-Mills morphism}
Let us first fix some notations.
Let $V$ be an object in $\Loc_m$ that is diffeomorphic (as a manifold) to $\bbR^m$.
Consider the vector space $\Omega^p(V)$
of $p$-forms on $V$. Because $V$ is in particular an oriented Lorentzian manifold,
we may introduce the Hodge operator $\ast : \Omega^p(V) \to \Omega^{m-p}(V)$.
Using also the de Rham differential $\dd : \Omega^p(V)\to \Omega^{p+1}(V)$,
we may define the codifferential $\delta := (-1)^{m\,(p+1)} \, \ast\dd\ast : \Omega^p(V) \to \Omega^{p-1}(V)$.
These operations extend to the vector space $\Omega^p(V,\g)$ of $\g$-valued $p$-forms.
Given $A\in \Omega^1(V,\g)$, we also define
\begin{flalign}
\dd_A^{} \,:\, \Omega^p(V,\g)\longrightarrow \Omega^{p+1}(V,\g)~,~~
\omega \longmapsto \dd_A^{}\omega := \dd \omega + A\wedge\omega - (-1)^p\,\omega\wedge A~,
\end{flalign}
and
\begin{flalign}
\delta_{A}^{} :=  (-1)^{m\,(p+1)} \, \ast\dd_A^{} \ast  \,:\, \Omega^p(V,\g)\longrightarrow \Omega^{p-1}(V,\g)~.
\end{flalign}
Interpreting $A$ as a gauge field, the Yang-Mills equation on $V$ reads as $\delta_A^{} F(A) =0$, where
$F(A) = \dd A + A\wedge A$ is the curvature of $A$.
\sk

Given also an object $U$ in $\CC$, the Hodge operator, the de Rham differential and the codifferential on $V$
may be extended vertically along $V\times U\to U$ to the vector space $\Omega^{p,0}(V\times U)$ of vertical 
$p$-forms on $V\times U\to U$.
We denote the vertical Hodge operator by $\ast^\ver : \Omega^{p,0}(V\times U) \to \Omega^{m-p,0}(V\times U)$,
the vertical de Rham differential by $\dd^\ver : \Omega^{p,0}(V\times U) \to \Omega^{p+1,0}(V\times U)$
and the vertical codifferential by $\delta^\ver : \Omega^{p,0}(V\times U) \to \Omega^{p-1,0}(V\times U)$.
Of course, these operations extend to the vector space $\Omega^{p,0}(V\times U,\g)$ of vertical 
$\g$-valued $p$-forms on $V\times U\to U$.
Given $A\in \Omega^{1,0}(V\times U,\g)$, we also define
\begin{flalign}
\nn \dd_A^{\ver} \,:\, \Omega^{p,0}(V\times U,\g) &\longrightarrow \Omega^{p+1,0}(V\times U,\g)~,~\\
\omega & \longmapsto \dd_A^{\ver}\omega := \dd^\ver \omega + A\wedge\omega - (-1)^p\,\omega\wedge A~,
\end{flalign}
and
\begin{flalign}
\delta_{A}^{\ver} :=  (-1)^{m\,(p+1)} \, \ast^\ver \dd_A^{\ver} \ast^\ver  \,:\, \Omega^{p,0}(V\times U,\g)
\longrightarrow \Omega^{p-1,0}(V\times U,\g)~.
\end{flalign}
Interpreting $A$ as a smoothly $U$-parametrized
gauge field, the vertical
Yang-Mills equation on $V\times U\to U$ reads as $\delta_A^{\ver} F^\ver(A) =0$,
where $F^\ver(A) = \dd^\ver A + A\wedge A$ is the vertical curvature of $A$, see also Section \ref{subsec:curvature}.
\sk

Let now $M$ be any object in $\Loc_m$. The gauge fields on $M$ are described by
the stack $\GCon(M)$. In the following we shall always use
the weakly equivalent explicit description given by 
Proposition \ref{propo:GCon} and Corollary \ref{cor:cccofibrep}. (Recall also Remark \ref{rem:cccofibrep}
for what this means in practice.)
We define the Yang-Mills $\HH$-morphism 
\begin{flalign}
\YM_M \, :\,  \GCon(M) ~  \longrightarrow~\bfOmega^1(M,\ad(G))
\end{flalign}
in terms of the vertical Yang-Mills operators $\delta_A^{\ver} F^\ver(A)$
mimicking the definition of the curvature $\HH$-morphism $\FF_M:
\GCon(M) \to \bfOmega^2(M,\ad(G))$ in terms of the vertical curvatures $F^\ver(A)$, 
cf.\ \eqref{eqn:curvaturemorph}.
Explicitly, at stage $U$ in $\CC$, the associated
groupoid morphism $\YM_M :  \GCon(M)(U) \to \bfOmega^1(M,\ad(G))(U)$ is defined as
\begin{flalign}
\nn \YM_M ~:~~\big(\AAA,\PPP\big) = \big(\{A_V\} , \{g_{VV^\prime}\}\big)~&\longmapsto~ \big(\{\delta_{A_V}^\ver F^{\ver}(A_V)\} , \{g_{VV^\prime}\}\big)~,\\
\big(\hhh : (\AAA,\PPP)\to(\AAA^\prime,\PPP^\prime)\big) ~&\longmapsto ~\big( \hhh : \YM_M(\AAA,\PPP)\to \YM_M(\AAA^\prime,\PPP^\prime)\big)~.\label{eqn:YMeqnM}
\end{flalign}
To confirm that this indeed defines a groupoid morphism observe that
$\delta_{A\ra^\ver g}^\ver F^{\ver}(A\ra^\ver g) = \ad_g(\delta^\ver_A F^\ver (A))$
under gauge transformations. Naturality with respect to a change of stage 
under any $\CC$-morphism $\ell : U\to U^\prime$ is clear because, loosely speaking, 
all vertical operations do not involve the objects $U$ in $\CC$ by definition.
Finally, it is easy to confirm that the Yang-Mills $\HH$-morphisms $\YM_M$ are the components
of a natural transformation
\begin{flalign}
\YM \,:\, \GCon ~\longrightarrow~ \bfOmega^1(-,\ad(G))
\end{flalign}
between the functors $\GCon : \Loc_m^{\op} \to \HH$
and $\bfOmega^1(-,\ad(G)) : \Loc_m^\op \to \HH$. 
(This uses the fact that morphisms $f: M\to M^\prime$ 
in $\Loc_m$ preserve causally convex open subsets 
and intertwine the Hodge operators. The same holds true for 
their restrictions $f_V : V\to f(V)$ to all causally convex open subsets $V\subseteq M$.)

\subsection{\label{subsec:YMsol}Yang-Mills stack}
A naive way to define the stack of solutions $\GSol(M)$ to the Yang-Mills
equation on $M$ is as follows: For any object $U$ in $\CC$,
define $\GSol(M)(U)$ to be the full sub-groupoid
of $\GCon(M)(U)$ consisting of all objects 
$(\AAA,\PPP) = (\{A_V\},\{g_{VV^\prime}\})$ satisfying
\begin{flalign}\label{eqn:naiveYM}
\YM_M\big(\AAA,\PPP\big) = \big(\{\delta_{A_V}^\ver F^{\ver}(A_V)\} , \{g_{VV^\prime}\}\big)=\big(\{0 \}, \{g_{VV^\prime}\}\big) ~,
\end{flalign}
where $\YM_M$ is the Yang-Mills morphism  \eqref{eqn:YMeqnM} and $0\in \Omega^{1,0}(V\times U,\g)$
are the vanishing vertical $1$-forms on $V\times U\to U$. The groupoid morphism
$\GCon(M)(\ell) : \GCon(M)(U^\prime)\to \GCon(M)(U) $
associated to a $\CC$-morphism $\ell : U \to U^\prime$
clearly induces a groupoid morphism $\GSol(M)(\ell) : \GSol(M)(U^\prime)\to \GSol(M)(U)$,
hence $\GSol(M) : \CC^\op\to \Grpd$ is a functor, i.e.\ an object in $\HH$.
The problem with this naive construction is that it is not clear whether it is homotopically meaningful,
i.e.\ whether it  preserves weak equivalences in $\HH$.
\sk

To ensure that the stack of solutions to the Yang-Mills equation 
is a homotopically meaningful concept,
we will propose an abstract definition.
Recall the stack $\GBun(M)$ of principal $G$-bundles on $M$ from Definition \ref{defi:GBun}.
For any $p\in\bbZ_{\geq0}$, we introduce the $\HH$-morphism
\begin{flalign}
\mathbf{0}_M \,:\, \GBun(M) ~\longrightarrow~ \bfOmega^p(M,\ad(G))~
\end{flalign}
whose stages $\mathbf{0}_M: \GBun(M)(U) \to \bfOmega^p(M,\ad(G))(U)$ 
are the groupoid morphisms
\begin{flalign}
\nn \mathbf{0}_M ~:~~ \PPP=\{g_{VV^\prime}\} ~&\longmapsto~\big(\{0\}, \{g_{VV^\prime}\}\big)~,\\
\big(\hhh : \PPP\to \PPP^\prime\big) ~&\longmapsto~ \big(\hhh : \mathbf{0}_M(\PPP)\to \mathbf{0}_M(\PPP^\prime)\big)~.\label{eqn:0mstagewise}
\end{flalign}
(Similarly to the curvature $\HH$-morphism $\FF_M : \GCon(M) \to \bfOmega^2(M,\ad(G))$
in Section \ref{subsec:curvature}, the $\HH$-morphism $\mathbf{0}_M$
may be obtained more abstractly by inducing the canonical $\HH$-morphism
$0 : \BG\to \BG_{\Omega^p_{\ad}}$ between the classifying stacks
to their corresponding mapping stacks and concretifications.)
In words, the $\HH$-morphism $\mathbf{0}_M$ 
assigns to a smoothly $U$-parametrized principal $G$-bundle $\PPP$ 
the vanishing smoothly $U$-parametrized $p$-form $\mathbf{0}_M(\PPP)$ with values in the associated adjoint bundle.
It is clear that $\mathbf{0}_M$ are the components of a natural transformation 
$\mathbf{0} : \GBun \to \bfOmega^p(-,\ad(G))$
between functors from $\Loc_m^\op$ to $\HH$. 
\begin{defi}\label{defi:solstack}
Let $M$ be an object in $\Loc_m$.  We define the 
{\em Yang-Mills stack} on $M$ as 
the homotopy fiber product 
\begin{flalign}
\GSol(M) ~:= ~\GCon(M) ~\times^h_{\bfOmega^1(M,\ad(G))} ~ \GBun(M)
\end{flalign}
of the pullback diagram
\begin{flalign}\label{eqn:solutionpullback}
\xymatrix{
\GCon(M) \ar[rr]^-{\YM_M}~~&&~~\bfOmega^1(M,\ad(G))~~&&~~\ar[ll]_-{\mathbf{0}_M}\GBun(M)
}
\end{flalign}
in $\HH$. Notice that the pullback diagram is natural in $M$, 
hence this construction defines a functor $\GSol : \Loc_m^\op \to H$. 
Furthermore, on account of Proposition \ref{propo:homfibprod} 
and recalling that $\GCon(M)$, $\bfOmega^1(M,\ad(G))$ and $\GBun(M)$ are stacks, 
$\GSol(M)$ is a stack too.
\end{defi}
\begin{rem}
Using Proposition \ref{propo:homfibprod}, we may explicitly compute $\GSol(M)$.
Due to reasons explained below, this computation can be simplified considerably.
Nonetheless, we find it instructive to see explicitly how the homotopy fiber product enforces the
Yang-Mills equation, hence we briefly sketch the interesting part of the computation of the objects in 
$\GSol(M)$. Let $U$ be any object in $\CC$. By Proposition \ref{propo:homfibprod},
the objects of $\GSol(M)(U)$ are given by triples 
\begin{flalign}
\big((\AAA,\PPP), \widetilde{\PPP}, \mathbf{k}\big)~,
\end{flalign}
where $(\AAA,\PPP)$ is an object in $\GCon(M)(U)$,
$\widetilde{\PPP}$ is an object in $\GBun(M)(U)$
and
\begin{flalign}\label{eqn:tmpYMhom0}
\mathbf{k} \,: \, \YM_M\big(\AAA,\PPP\big)~ \longrightarrow~ \mathbf{0}_M\big(\widetilde{\PPP}\big)
\end{flalign}
is a morphism in $\bfOmega^1(M,\ad(G))(U)$. Observe that this does {\em not} enforce the 
Yang-Mills equation in the strict sense as above in our naive construction (cf.\ \eqref{eqn:naiveYM}),
but it demands that $\YM_M(\AAA,\PPP)$ is isomorphic to one of the zeros $\mathbf{0}_M(\widetilde{\PPP})$
in $\bfOmega^1(M,\ad(G))(U)$. It is important to notice that every morphism in
$\bfOmega^1(M,\ad(G))(U)$ with target given by $\mathbf{0}_M(\widetilde{\PPP})$
necessarily has to originate from an object of the form $\mathbf{0}_M(\PPP^\prime)$,
because the vanishing vertical $1$-forms $0\in\Omega^{1,0}(V\times U,\g)$ 
are invariant under the adjoint action $\ad$ of gauge transformations.
As a consequence, the fact that the morphism \eqref{eqn:tmpYMhom0} exists implies already the strict condition
that $\delta^\ver_{A_V} F^\ver(A_V)=0$ for all $V$. Summing up, we observe that
even though a priori the homotopy fiber product enforces a weaker version of the Yang-Mills equation
(i.e.\ up to isomorphism), a specific feature of the $\HH$-morphism $\mathbf{0}_M$
turns this weaker version into a strict equality similar to \eqref{eqn:naiveYM}.
Below we shall make this statement precise by proving that $\mathbf{0}_M$ is a fibration
in $\HH$, hence the homotopy fiber product may be computed (up to weak equivalence) 
by the ordinary fiber product. This will eventually show that  $\GSol(M)$
is  weakly equivalent to our naive solution stack discussed at 
the beginning of this subsection.
\end{rem}

As already mentioned in the remark above, there exists a weakly equivalent simplified description
of the Yang-Mills stack $\GSol(M)$ given in Definition \ref{defi:solstack}. It relies on the following
observation.
\begin{lem}
For every $p\in\bbZ_{\geq0}$, the $\HH$-morphism 
$\mathbf{0}_M: \GBun(M) \to \bfOmega^p(M,\ad(G))$ is a fibration
in the local model structure on $\HH$.
\end{lem}
\begin{proof}
Because $\GBun(M)$ and (the simplified description of) 
$\bfOmega^p(M,\ad(G))$ are stacks, by Lemma \ref{lem:fibrationlocfibrant} 
the thesis follows if we show that $\mathbf{0}_M: \GBun(M) \to \bfOmega^p(M,\ad(G))$
is a stage-wise fibration in $\Grpd$ (recall Theorem \ref{theo:groupoids}). 
Using the explicit expressions for the stages \eqref{eqn:0mstagewise},
one immediately realizes that this is indeed the case 
because the functor $\mathbf{0}_M: \GBun(M)(U) \to \bfOmega^p(M,\ad(G))(U)$ is fully faithful  
and any object of the groupoid $\bfOmega^p(M,\ad(G))(U)$ isomorphic 
to one in the image of $\mathbf{0}_M: \GBun(M)(U) \to \bfOmega^p(M,\ad(G))(U)$ 
has to lie in the image too. 
\end{proof}

As a consequence of this lemma,
the homotopy fiber product in Definition \ref{defi:solstack}
is weakly equivalent to the ordinary fiber product.
Therefore, similarly to Proposition \ref{propo:GCon}, we obtain 
a weakly equivalent simplified description of $\GSol(M)$,
which agrees with our naive solution stack from the beginning of this subsection. 
Summing up, we obtained
\begin{propo}\label{propo:GSol}
Up to weak equivalence in $\HH$, the Yang-Mills stack $\GSol(M)$
defined in Definition \ref{defi:solstack} has the following explicit description:
For all objects $U$ in $\CC$, the groupoid $\GSol(M)(U)$ is the full sub-groupoid
of $\GCon(M)(U)$  (cf.\ Proposition \ref{propo:GCon}) consisting of all objects $(\AAA,\PPP)=(\{A_V\},\{g_{VV^\prime}\})$ 
satisfying the vertical Yang-Mills equation
\begin{flalign}
\YM_M \big(\AAA,\PPP\big) = \mathbf{0}_M(\PPP) \qquad \Longleftrightarrow\qquad \delta_{A_V}^\ver F^{\ver}(A_V)=0~,
\end{flalign}
for all causally convex open subsets $V\subseteq M$ diffeomorphic to $\bbR^m$.
For all morphisms $\ell : U\to U^\prime$ in $\CC$,  the groupoid morphism 
$\GSol(M)(\ell) : \GSol(M)(U^\prime) \to \GSol(M)(U)$ is 
the one induced by $\GCon(M)(\ell) : \GCon(M)(U^\prime) \to \GCon(M)(U)$ (cf.\ Proposition \ref{propo:GCon}). 
\end{propo}
\begin{rem}
Similarly to Remark \ref{rem:concmappinstackfunctor}, 
it is easy to see that our weakly equivalent simplified 
model for the Yang-Mills stack given in Proposition \ref{propo:GSol}
is functorial, i.e.\ we have a functor $\GSol : \Loc_m^\op \to \HH$. 
Moreover, by an explicit computation similar to the one in Remark \ref{rem:GConisstack},
one can show that the above mentioned model for $\GSol(M)$ is a stack too.
\end{rem}


\section{\label{sec:Cauchy}Stacky Cauchy problem}
In this section we introduce and discuss a  {\em stacky} version of the Yang-Mills Cauchy problem.
It turns out that well-posedness of the stacky Cauchy problem is a 
stronger statement than well-posedness of the ordinary Cauchy problem for
gauge equivalence classes of Yang-Mills fields. In particular,
the solution for each given initial datum in the stacky Cauchy problem must be unique 
up to a {\em unique} isomorphism, which is stronger than uniqueness 
of their associated gauge equivalence classes. 
To set up the stacky Cauchy problem, we first introduce
a stack $\GData(\Sigma)$ that describes initial data on a Cauchy surface $\Sigma$
for the Yang-Mills equation and an $\HH$-morphism $\data_{\Sigma} : \GSol(M)\to \GData(\Sigma)$
that assigns to Yang-Mills fields their initial data on $\Sigma$.
This will allow us to define well-posedness of the stacky Yang-Mills Cauchy problem 
using the language of model categories.
We conclude explaining that this condition is equivalent to a family of parametrized PDE problems,
which may be addressed by ordinary PDE-theoretical techniques. 

\subsection{\label{subsec:data}Initial data stack}
Let $M$ be any object in $\Loc_m$ and $\Sigma\subseteq M$ any spacelike Cauchy surface.
Recall  that $\dim(\Sigma) = \dim(M)-1 = m-1$.
In the usual approach, see e.g.\ \cite{CSYM,CBYM}, an initial datum on $\Sigma$ 
for the Yang-Mills equation on $M$ is a triple $(A^\Sigma, E,P^\Sigma)$ consisting of a 
principal $G$-bundle $P^\Sigma$ on $\Sigma$ with connection
$A^\Sigma$ and a $1$-form $E$ on $\Sigma$  with values 
in the corresponding adjoint bundle,
which satisfies the Yang-Mills constraint $\delta_{A^\Sigma}^{} E=0$.
Here $\delta_{A^\Sigma}^{}$ is the covariant
codifferential on $\Sigma$ that is obtained from the induced Hodge operator on $\Sigma$. 
As a refinement of the set of initial data used in \cite{CSYM,CBYM},
our approach allows us to introduce a stack $\GData(\Sigma)$ of initial data on $\Sigma$.
Abstractly, this stack may be obtained by the following construction:
Consider the stack of gauge fields $\GCon(\Sigma)$ on $\Sigma$
and form its tangent stack $T^{\mathrm{st}}\GCon(\Sigma)$ using similar techniques as in \cite{Hepworth}.
Then implement the Yang-Mills constraint by a homotopy fiber product similarly to Section \ref{subsec:YMsol}.
Since for our practical purposes the construction of tangent stacks in \cite{Hepworth} is too involved,
we shall not employ this abstract perspective and instead define directly the stack of initial data 
(up to weak equivalence) in an explicit form.
\begin{defi}\label{defi:GData}
Let $M$ be any object in $\Loc_m$ and $\Sigma\subseteq M$ any spacelike Cauchy surface.
The {\em stack of initial data} on $\Sigma$ is the following presheaf of groupoids
$\GData(\Sigma) : \CC^\op \to \Grpd$: 
\begin{itemize}
\item For all objects $U$ in $\CC$,
the objects of the groupoid $\GData(\Sigma)(U)$ are triples
\begin{flalign}
\big(\AAA^\Sigma , \EEE,\PPP^\Sigma\big)~,
\end{flalign} 
such that $(\AAA^\Sigma,\PPP^\Sigma)$ is an object in $\GCon(\Sigma)(U)$ and
$(\EEE,\PPP^\Sigma)$ an object in $\bfOmega^1(\Sigma,\ad(G))(U)$,
satisfying the vertical Yang-Mills constraint
\begin{flalign}
 \delta^{\ver}_{A_W^\Sigma} E_W =0~,
\end{flalign} 
for all open subsets $W\subseteq \Sigma$ diffeomorphic to $\bbR^{m-1}$.
The morphisms of the groupoid  $\GData(\Sigma)(U)$ are 
\begin{flalign}
\hhh^\Sigma  ~:~(\AAA^\Sigma,\EEE,\PPP^\Sigma)~\longrightarrow~(\AAA^{\prime\Sigma},\EEE^\prime,\PPP^{\prime\Sigma})~,
\end{flalign}
such that $\hhh^\Sigma : (\AAA^\Sigma,\PPP^\Sigma) \to (\AAA^{\prime\Sigma},\PPP^{\prime\Sigma})$
is a morphism in $\GCon(\Sigma)(U)$ and 
$\hhh^\Sigma : (\EEE,\PPP^\Sigma)\to(\EEE^\prime,\PPP^{\prime\Sigma})$
a morphism in $\bfOmega^1(\Sigma,\ad(G))(U)$.

\item For all $\CC$-morphisms $\ell : U\to U^\prime$, the groupoid morphism 
$\GData(\Sigma)(\ell) : \GData(\Sigma)(U^\prime) \to \GData(\Sigma)(U)$
is given by
\begin{flalign}
\nn\GData(\Sigma)(\ell) ~: ~ (\AAA^\Sigma,\EEE,\PPP^\Sigma) ~&\longmapsto~ (\ell^\ast\AAA^\Sigma,\ell^\ast\EEE,\ell^\ast\PPP^\Sigma) ~,\\
\hhh^\Sigma ~&\longmapsto ~ \ell^\ast\hhh^\Sigma~,
\end{flalign}
where our pullback notation $\ell^\ast$ is analogous to the one introduced in Proposition \ref{propo:GCon}.
\end{itemize}
\end{defi}

\subsection{\label{subsec:datamap}Initial data morphism}
Let $M$ be any object in $\Loc_m$ and $\Sigma\subseteq M$ any spacelike Cauchy surface.
We denote by $\iota_\Sigma : \Sigma\to M$ the embedding of the Cauchy surface
$\Sigma$ into $M$.  We further choose and fix any normalized future-directed timelike 
vector field $n$ on $M$ whose restriction to $\Sigma\subseteq M$ is 
normal to the Cauchy surface. Given any object $U$ in $\CC$,
we denote by the same symbol $n$ also the vector field on $M\times U$ 
that is obtained by extending $n$ constantly along $U$. 
The restrictions of $n$ to open subsets $V\times U$ of $M\times U$  are also denoted by 
$n$ in order to simplify our notations.
\sk

The aim of this subsection is to define the {\em initial data} $\HH$-morphism
\begin{flalign}\label{eqn:datamorphism}
\data_{\Sigma}^{} \,:\, \GSol(M) ~\longrightarrow~\GData(\Sigma)
\end{flalign}
from the Yang-Mills stack (cf.\ Proposition \ref{propo:GSol}) to 
the initial data stack (cf.\ Definition \ref{defi:GData}) whose role is
to assign to solutions of the Yang-Mills equation their initial data.
\sk

This requires some preparations. Let $W\subseteq \Sigma$ be any open subset of the Cauchy
surface that is diffeomorphic to $\bbR^{m-1}$. We denote by 
$D(W)\subseteq M$ the {\em Cauchy development} of $W$ in $M$.
By definition, see e.g.\ \cite[Definition 14.35]{ONeill},
$D(W)\subseteq M$ is the subset of points $p \in M$ such that every inextensible
causal curve in $M$ emanating from $p$ meets $W$. By \cite[Lemma 14.43]{ONeill},
we may regard the Cauchy development $D(W)$ equipped with the induced metric, 
orientation and time-orientation as an object in $\Loc_m$.
It is diffeomorphic (as a manifold) to $\bbR^m$ because $W$ defines a Cauchy surface for $D(W)$
and $W$ is by hypothesis diffeomorphic to $\bbR^{m-1}$.
We denote by $\iota_W : W \to D(W)$ the restriction of the embedding 
$\iota_\Sigma : \Sigma\to M$ to the domain $W$ and the codomain $D(W)$.
Let us now consider two open subsets $W\subseteq \Sigma$ and $W^\prime\subseteq \Sigma$
of the Cauchy surface that are both diffeomorphic to $\bbR^{m-1}$. Because $W$ and $W^\prime$
are both subsets of the same Cauchy surface, it is easy to show that $D(W\cap W^\prime) = D(W)\cap D(W^\prime)$.
The latter is again a causally convex open subset of $M$ (not necessarily diffeomorphic to $\bbR^m$).
We denote by $\iota_{W\cap W^\prime} : W\cap W^\prime  \to D(W)\cap D(W^\prime)$ 
the restriction of the embedding  $\iota_\Sigma : \Sigma\to M$ to the domain $W\cap W^\prime$ 
and the codomain $D(W)\cap D(W^\prime)$. Similar statements hold for higher intersections.
\sk

For any object $U$ in $\CC$, we define the stage $\data_{\Sigma}^{} : \GSol(M)(U) \to \GData(\Sigma)(U)$
of the $\HH$-morphism \eqref{eqn:datamorphism} as the groupoid morphism specified by the assignment 
\begin{flalign}
\nn \data_{\Sigma}^{} ~:~\quad\qquad\qquad\qquad (\AAA,\PPP) ~&\longmapsto~\Big(\iota_\Sigma^\ast\AAA , \iota_\Sigma^\ast\big( n\,\lrcorner\,\FF_M(\AAA)\big),\iota_{\Sigma}^\ast\PPP\Big)~,\\
\big(\hhh: (\AAA,\PPP) \to (\AAA^\prime,\PPP^\prime)\big) ~&\longmapsto~ \big(\iota_\Sigma^\ast\hhh: \data_{\Sigma}^{}(\AAA,\PPP) \to \data_{\Sigma}^{}(\AAA^\prime,\PPP^\prime)\big)~,\label{eqn:datamorphismstage}
\end{flalign}
where we used an intuitive compact notation. 
Explicitly, for $(\AAA,\PPP) = ( \{A_V\},\{g_{VV^\prime}\})$ and $\hhh = \{h_V\}$, 
we set
\begin{subequations}
\begin{flalign}
\iota_\Sigma^\ast\AAA &:=  \big\{(\iota_W\times \id_U)^\ast \, A_{D(W)}\big\}~,\\
 \iota_\Sigma^\ast\big( n\,\lrcorner\,\FF_M(\AAA)\big) &:=  \big\{(\iota_W\times \id_U)^\ast \, \big(n \,\lrcorner\, F^\ver(A_{D(W)})\big)\big\}~,\\
\iota_\Sigma^\ast\PPP &:= \big\{(\iota_{W\cap W^\prime}\times\id_U)^\ast\,g_{D(W) D(W^\prime)} \big\}~,\\
\iota_\Sigma^\ast\hhh &:= \big\{(\iota_W \times\id_U)^\ast\, h_{D(W)}\big\}~,
\end{flalign}
\end{subequations}
where $n \,\lrcorner\, F^\ver(A_{D(W)}) \in \Omega^{1,0}(D(W)\times U,\g)$ denotes
the contraction of the vertical curvature $F^\ver(A_{D(W)}) \in \Omega^{2,0}(D(W)\times U,\g)$ 
with our fixed normalized timelike vector field $n$. (As explained above, the vector field $n$ 
on $M$ is extended constantly to $M\times U$ and then restricted to $D(W)\times U$.)
Using the vertical Yang-Mills equations $\delta_{A_V}^\ver F^{\ver}(A_V)=0$, 
it is an elementary check that $ \data_\Sigma (\AAA,\PPP)$
satisfies the vertical Yang-Mills constraint of Definition \ref{defi:GData}. Verifying that \eqref{eqn:datamorphismstage}
indeed defines a groupoid morphism and confirming naturality in $U$ are also elementary checks.
Hence, we defined the initial data $\HH$-morphism in \eqref{eqn:datamorphism}.

\subsection{\label{subsec:cauchyexplicit}Cauchy problem}
Using the language of model categories, we can now define the concept of a
well-posed stacky Cauchy problem associated to the Yang-Mills equation.
We will then show that our definition is equivalent to a family of parametrized PDE problems, 
which generalize the ordinary Yang-Mills Cauchy problem, cf.\  \cite{CSYM, CBYM}. 
The purpose of this section is thus to establish a bridge between our model categorical perspective 
on Yang-Mills theory and the language more familiar to PDE theorists. 
\begin{defi}\label{def:stackycauchypbl}
Let $M$ be an object in $\Loc_m$ and $\Sigma\subseteq M$ a spacelike Cauchy surface. 
We say that the {\em stacky Cauchy problem for the Yang-Mills equation is well-posed} if 
the initial data $\HH$-morphism $\data_{\Sigma}^{}: \GSol(M) \to \GData(\Sigma)$, cf.\ \eqref{eqn:datamorphism}, 
is a weak equivalence.
\end{defi}

Our goal is to rephrase this abstract definition in terms of more explicit conditions. 
To begin with, observe that both $\GSol(M)$ and $\GData(\Sigma)$ are stacks, i.e.\ fibrant objects in $\HH$. 
Hence, by Theorem \ref{theo:localH},
the $\HH$-morphism $\data_{\Sigma}^{}: \GSol(M) \to \GData(\Sigma)$ 
is a weak equivalence if and only if for each object $U$ in $\CC$ 
the groupoid morphism $\data_{\Sigma}^{}: \GSol(M)(U) \to \GData(\Sigma)(U)$ 
is a weak equivalence, i.e.\ a fully faithful and essentially surjective functor. 
\sk

Recalling the explicit description in Section \ref{subsec:datamap}, we observe that
 $\data_{\Sigma}^{}: \GSol(M)(U) \to \GData(\Sigma)(U)$ is 
essentially surjective if and only if the following holds true:
For every object $(\AAA^\Sigma,\EEE,\PPP^\Sigma)$ in $\GData(\Sigma)(U)$, 
there exists an object $(\AAA,\PPP)$ in $\GSol(M)(U)$ and a morphism 
$\hhh^\Sigma: \data_{\Sigma}^{}(\AAA,\PPP) \to (\AAA^\Sigma,\EEE,\PPP^\Sigma)$ 
in $\GData(\Sigma)(U)$. Note that $(\AAA,\PPP)$ may be interpreted as a solution for 
the initial datum $(\AAA^\Sigma,\EEE,\PPP^\Sigma)$ up to the isomorphism 
$\hhh^\Sigma: \data_{\Sigma}^{}(\AAA,\PPP) \to (\AAA^\Sigma,\EEE,\PPP^\Sigma)$ of initial data.
In other words, given any smoothly $U$-parametrized initial datum 
$(\AAA^\Sigma,\EEE,\PPP^\Sigma)$ on $\Sigma$, 
essential surjectivity of $\data_{\Sigma}^{}$ is equivalent to the statement
that we can find a smoothly $U$-parametrized Yang-Mills field $(\AAA,\PPP)$ on $M$, 
whose initial datum is isomorphic to $(\AAA^\Sigma,\EEE,\PPP^\Sigma)$
by a smoothly $U$-parametrized gauge transformation 
$\hhh^\Sigma: \data_{\Sigma}^{}(\AAA,\PPP) \to (\AAA^\Sigma,\EEE,\PPP^\Sigma)$
of initial data.
\sk

We now analyze what full faithfulness of $\data_{\Sigma}^{}: \GSol(M)(U) \to \GData(\Sigma)(U)$ implies. 
Explicitly, this property means that, given any two objects 
$(\AAA,\PPP)$ and $(\AAA^\prime,\PPP^\prime)$ in $\GSol(M)(U)$
and any morphism 
$\hhh^\Sigma: \data_{\Sigma}^{}(\AAA,\PPP) \to \data_{\Sigma}^{}(\AAA^\prime,\PPP^\prime)$ 
in $\GData(\Sigma)(U)$, there exists a unique morphism $\hhh: (\AAA,\PPP) \to (\AAA^\prime,\PPP^\prime)$ 
in $\GSol(M)$ such that $\data_\Sigma(\hhh) = \hhh^\Sigma$. 
In other words, this means that every 
smoothly $U$-parametrized gauge transformation  between the initial data
of two smoothly $U$-parametrized Yang-Mills fields 
admits a unique extension from the Cauchy surface $\Sigma$ 
to the whole spacetime $M$. 
\sk

Full faithfulness of $\data_{\Sigma}^{}: \GSol(M)(U) \to \GData(\Sigma)(U)$ is equivalent
to the following more practical condition: Given any object
$(\AAA^\Sigma,\EEE,\PPP^\Sigma)$ in $\GData(\Sigma)(U)$
and any two morphisms $\hhh^\Sigma : \data_{\Sigma}^{}(\AAA,\PPP)\to (\AAA^\Sigma,\EEE,\PPP^\Sigma)$
and $\hhh^{\prime\Sigma} : \data_{\Sigma}^{}(\AAA^\prime,\PPP^\prime)\to (\AAA^\Sigma,\EEE,\PPP^\Sigma)$
in $\GData(\Sigma)(U)$,
there exists  a unique morphism $\hhh: (\AAA,\PPP)\to(\AAA^\prime,\PPP^\prime)$ in $\GSol(M)(U)$
such that $\hhh^{\prime\Sigma} \circ \data_{\Sigma}^{}(\hhh) = \hhh^{\Sigma} $.
In fact, assuming full faithfulness, the desired  morphism
$\hhh: (\AAA,\PPP)\to(\AAA^\prime,\PPP^\prime)$ in $\GSol(M)(U)$
is the one uniquely obtained from the morphism
${\hhh^{\prime\Sigma}}^{-1} \circ \hhh^\Sigma: \data_\Sigma(\AAA,\PPP) 
\to \data_\Sigma(\AAA^\prime,\PPP^\prime)$ in $\GData(\Sigma)(U)$.
Conversely, given any morphism $\hhh^\Sigma: \data_{\Sigma}^{}(\AAA,\PPP) 
\to \data_{\Sigma}^{}(\AAA^\prime,\PPP^\prime)$ in $\GData(\Sigma)(U)$, 
take $\hhh^{\prime\Sigma} : \data_{\Sigma}^{}(\AAA^\prime,\PPP^\prime)\to 
\data_{\Sigma}^{}(\AAA^\prime,\PPP^\prime)$ to be the identity. 
It follows from the condition stated above that 
there exists a unique morphism $\hhh: (\AAA,\PPP) \to (\AAA^\prime,\PPP^\prime)$ 
in $\GSol(M)(U)$ such that $\data_\Sigma(\hhh) = \hhh^\Sigma$. 
\sk

We summarize the results obtained above in the following
\begin{propo}\label{propo:cauchy}
Let $M$ be an object in $\Loc_m$ and $\Sigma\subseteq M$ a spacelike Cauchy surface. 
Then the stacky Cauchy problem for the Yang-Mills equation is well-posed, cf.\ Definition \ref{def:stackycauchypbl}, 
if and only if the following conditions hold true, for all objects $U$ in $\CC$: 
\begin{enumerate}
\item Given any object $(\AAA^\Sigma,\EEE,\PPP^\Sigma)$ in $\GData(\Sigma)(U)$, there exists 
an object $(\AAA,\PPP)$ in $\GSol(M)(U)$ together with a morphism 
$\hhh^\Sigma: \data_{\Sigma}^{}(\AAA,\PPP) \to (\AAA^\Sigma,\EEE,\PPP^\Sigma)$ in $\GData(\Sigma)(U)$.

\item Given any other object $(\AAA^\prime,\PPP^\prime)$ in $\GSol(M)(U)$ 
and morphism $\hhh^{\prime\Sigma}: \data_{\Sigma}^{}(\AAA^\prime,\PPP^\prime) 
\to (\AAA^\Sigma,\EEE,\PPP^\Sigma)$ in $\GData(\Sigma)(U)$, 
there exists a unique morphism $\hhh: (\AAA,\PPP) \to (\AAA^\prime,\PPP^\prime)$ 
in $\GSol(M)(U)$ such that diagram 
\begin{flalign}\label{eqn:cauchydiagram}
\xymatrix{
\data_\Sigma(\AAA,\PPP) \ar[rr]^-{\data_\Sigma(\hhh)} \ar[rd]_-{\hhh^\Sigma~~} && 
\data_\Sigma(\AAA^\prime,\PPP^\prime) \ar[ld]^-{~~\hhh^{\prime\Sigma}} \\
& (\AAA^\Sigma,\EEE,\PPP^\Sigma)
}
\end{flalign}
in $\GData(\Sigma)(U)$ commutes. 
\end{enumerate}
\end{propo}
\begin{rem}
In more explicit words, item 1.\ of this proposition demands 
that there exists for every smoothly $U$-parametrized
initial datum $(\AAA^\Sigma,\EEE,\PPP^\Sigma)$ 
a smoothly $U$-parametrized solution $(\AAA,\PPP)$ of the Yang-Mills equation
whose initial datum is isomorphic to the given one
by a smoothly $U$-parametrized gauge transformation 
$\hhh^\Sigma: \data_{\Sigma}^{}(\AAA,\PPP) \to (\AAA^\Sigma,\EEE,\PPP^\Sigma)$.
Item 2.\ demands that any two such solutions are isomorphic
by a {\em unique} smoothly $U$-parametrized gauge transformation $\hhh: (\AAA,\PPP) \to (\AAA^\prime,\PPP^\prime)$.
This is clearly a stronger condition than existence and uniqueness of
solutions to the Cauchy problem for gauge equivalence classes, where the uniqueness 
aspect in item 2.\ does not play a role.
\end{rem}
\begin{rem}
The fact that the conditions in Proposition \ref{propo:cauchy} have to hold true 
for all objects $U$ in $\CC$ is very important from our stacky perspective.
This is because the groupoids $\GSol(M)(U)$ and $\GData(\Sigma)(U)$
encode the smooth structure of the Yang-Mills stack and initial data stack.
Thus, well-posedness of the stacky Cauchy problem does not only formalize 
the usual notion of well-posedness (i.e.\ existence and uniqueness of solutions for given initial data), 
but it also demands smooth dependence (in the sense of stacks) of solutions on their initial data.
The latter may be interpreted as a smooth analogue of the 
condition of continuous dependence of solutions on their initial data, 
which is more familiar in PDE theory.
\end{rem}
\begin{rem}\label{rem:solvabilityPDE}
As a final remark, we note that the conditions in Proposition \ref{propo:cauchy} 
are known to hold for $U = \bbR^0$ and spacetime dimension $m=2,3,4$ \cite{CSYM, CBYM}. 
However, we are not aware of any results for other objects $U$ in $\CC$,
which leads to the more complicated realm of smoothly $U$-parametrized PDE problems.
Being crucial for a better understanding of the geometry of the stack $\GSol(M)$ of Yang-Mills fields, 
we believe that a detailed study of the explicit conditions of Proposition \ref{propo:cauchy} 
is a very interesting and compelling PDE problem. 
This problem is clearly beyond the scope of the present work.
However, we would like to mention that analogous results
for the  simpler case of smoothly $U$-parametrized normally hyperbolic {\em linear} PDEs
can be established via the theory of symmetric hyperbolic systems \cite{Florian}. 
We expect such techniques to be sufficient for proving that the stacky Cauchy problem
is well-posed for Abelian Yang-Mills theory with structure group $G=U(1)$.
\end{rem}


\section{\label{sec:Lorenz}Yang-Mills stack in Lorenz gauge}
We briefly discuss how gauge fixings may be interpreted
in our framework as weak equivalences of stacks. For simplicity, we 
shall focus on the particular example given by Lorenz gauge fixing, even 
though our ideas apply to other gauge fixings as well.
Recall from Proposition \ref{propo:GSol} the stack $\GSol(M)$ of Yang-Mills fields 
on a globally hyperbolic Lorentzian manifold $M$.
For $U$ an object in $\CC$, the objects of the groupoid
$\GSol(M)(U)$ are smoothly $U$-parametrized gauge fields 
$(\AAA,\PPP) = (\{A_V\},\{g_{VV^\prime}\})$ 
that solve the vertical Yang-Mills equation $\delta_{A_V}^\ver F^\ver(A_V)=0$,
for all causally convex open subsets $V\subseteq M$ diffeomorphic to $\bbR^m$.
We say that $(\AAA,\PPP)$ satisfies the {\em Lorenz gauge condition} if
\begin{flalign}\label{eqn:LorcondA}
\delta^\ver_{A_V} A_V =0~,
\end{flalign}
for all causally convex open subsets $V\subseteq M$ diffeomorphic to $\bbR^m$.
As a consequence of \eqref{eqn:LorcondA} and the conditions $A_{V^\prime}\vert_{(V\cap V^\prime)\times U} 
= A_{V}\vert_{(V\cap V^\prime)\times U}\ra^\ver g_{VV^\prime}$,
it follow that the $g_{VV^\prime}$ have to satisfy the conditions
\begin{flalign}\label{eqn:Lorcondg}
\delta^\ver_{A_{V}\vert_{(V\cap V^\prime)\times U}} \Big(\big(\dd^\ver g_{VV^\prime}\big)~ g^{-1}_{VV^\prime} \Big) =0~,
\end{flalign}
for all causally convex open subsets $V\subseteq M$ and $V^\prime\subseteq M$
diffeomorphic to $\bbR^m$. Notice that these are smoothly $U$-parametrized hyperbolic 
PDEs which are similar to the wave map equation, see e.g.\ \cite{CBsigma}.
Given any morphism $\hhh = \{h_V\} : (\AAA,\PPP)\to (\AAA^\prime,\PPP^\prime)$
in $\GSol(M)(U)$ between two objects that satisfy the Lorenz gauge condition,
it follows from \eqref{eqn:LorcondA} and the conditions $A_{V}^\prime
= A_{V}\ra^\ver h_V$ that the $h_V$ have to satisfy the smoothly $U$-parametrized PDEs
\begin{flalign}\label{eqn:Lorcondh}
\delta^\ver_{A_V} \Big( \big(\dd^\ver h_V\big)~h_V^{-1} \Big) =0~,
\end{flalign}
for all causally convex open subsets $V\subseteq M$ diffeomorphic to $\bbR^m$.
\begin{defi}
Let $M$ be an object in $\Loc_m$. 
The stack $\GSol_{\mathrm{g.f.}}(M)$ 
of {\em Lorenz gauge-fixed Yang-Mills fields} on $M$
is defined as the presheaf of groupoids whose stage at $U$ in $\CC$ is the full sub-groupoid of $\GSol(M)(U)$
specified by the Lorenz gauge condition \eqref{eqn:LorcondA}.
Because Lorenz gauge is natural in $M$, this construction defines
a functor $\GSol_{\mathrm{g.f.}} : \Loc_m^\op \to \HH$.
\end{defi}

By construction, there exists a natural transformation $j : \GSol_{\mathrm{g.f.}}\to \GSol$
of functors $\Loc_m^\op\to \HH$, whose components $j_M : \GSol_{\mathrm{g.f.}}(M) \to \GSol(M)$
are just the sub-presheaf embeddings.
It is an interesting question to ask whether the morphisms
$j_M :  \GSol_{\mathrm{g.f.}}(M) \to \GSol(M)$ are weak equivalences in $\HH$.
This would allow us to describe the Yang-Mills stack $\GSol(M)$ in terms
of the weakly equivalent stack $\GSol_{\mathrm{g.f.}}(M)$, which has the practical 
advantage that all of its gauge fields $(\AAA,\PPP)$ and gauge transformations $\hhh$ 
satisfy a smoothly parametrized hyperbolic PDE. (Recall \eqref{eqn:Lorcondg} and \eqref{eqn:Lorcondh}, 
as well as the fact that the Yang-Mills equation in Lorenz gauge is hyperbolic.)
Proving that $j_M$ is a weak equivalence in $\HH$ requires
assumptions on solvability of smoothly parametrized PDEs similar to those in Proposition \ref{propo:cauchy}.
We summarize the relevant statement in the following
\begin{propo}\label{propo:fixing}
Let $M$ be an object in $\Loc_m$. Then the $\HH$-morphism
$j_M :  \GSol_{\mathrm{g.f.}}(M) \to \GSol(M)$ is a weak equivalence
if and only if the following holds true: For all objects $U$ in $\CC$
and all objects $(\AAA,\PPP) = (\{A_V\},\{g_{VV^\prime}\})$ 
in $\GSol(M)(U)$ there exists a solution $h_V \in C^\infty(V\times U,G)$
of the smoothly $U$-parametrized PDE
\begin{flalign}\label{eqn:gaugefix}
\delta^\ver_{A_V} \Big(h_V^{-1}\,\dd^\ver h_V\Big) = \delta_{A_V}^{\ver} A_V~,
\end{flalign}
for all causally convex open subsets $V\subseteq M$ diffeomorphic to $\bbR^m$.
\end{propo}
\begin{proof}
Using Theorem \ref{theo:localH} and that both  
$\GSol_{\mathrm{g.f.}}(M)$ and $\GSol(M)$ are stacks, 
$j_M$ is a weak equivalence in $\HH$ if and only
if $j_M : \GSol_{\mathrm{g.f.}}(M)(U) \to \GSol(M)(U)$ is a weak equivalence
of groupoids, for all objects $U$ in $\CC$. By construction of $ \GSol_{\mathrm{g.f.}}(M)$, 
all these stages are fully faithful, hence $j_M$ is a weak equivalence if and only if all stages are essentially 
surjective. The latter means that for all objects $U$ in $\CC$ and all objects $(\AAA,\PPP)$ in
$\GSol(M)(U)$, there exists a morphism $\hhh = \{h_V\} : (\widetilde{\AAA},\widetilde{\PPP}) \to (\AAA,\PPP)$
in $\GSol(M)(U)$, such that $(\widetilde{\AAA},\widetilde{\PPP}) $ 
satisfies the Lorenz gauge condition \eqref{eqn:LorcondA}. Applying $\delta_{A_V}^{\ver}$
on the conditions $A_{V} = \widetilde{A}_V\ra^\ver h_V$ that every morphism $\hhh$ has to satisfy, 
one realizes that the Lorenz gauge condition for $\widetilde{A}_V$ coincides with 
the smoothly parametrized PDEs \eqref{eqn:gaugefix}. Hence, essential surjectivity of 
$j_M : \GSol_{\mathrm{g.f.}}(M)(U) \to \GSol(M)(U)$ is equivalent to solvability of the smoothly $U$-parametrized
PDEs \eqref{eqn:gaugefix} for all causally convex open subsets $V\subseteq M$ diffeomorphic to $\bbR^m$.
\end{proof}
\begin{rem}
For the unparametrized case $U=\bbR^0$, solvability of PDEs of the form \eqref{eqn:gaugefix}
has been studied in \cite{CBsigma} for spacetime dimension  $m=2$. 
However, we are not aware of any results for other objects
$U$ in $\CC$. We believe that, together with the smoothly parametrized PDEs explained
in Remark \ref{rem:solvabilityPDE}, these are very interesting problems for PDE theorists
as they are crucial for understanding the geometry of the Yang-Mills stack.
For the case of Abelian Yang-Mills theory with structure group $G=U(1)$,
we expect that similar techniques as in \cite{Florian}, which are based on the theory of 
symmetric hyperbolic systems, can be used to prove that the conditions in Proposition \ref{propo:fixing}
hold true and hence that $j_M$ is a weak equivalence in $\HH$. We hope to come back to this issue in future work.
\end{rem}


\section*{Acknowledgments}
The work of M.B.\ is supported by a Postdoctoral Fellowship 
of the Alexander von Humboldt Foundation (Germany). 
A.S.\ gratefully acknowledges the financial support of 
the Royal Society (UK) through a Royal Society University 
Research Fellowship, a Research Grant and an Enhancement Award. 
U.S.\ was supported by RVO:67985840.


\appendix


\section{\label{app:monoidal}Monoidal model structure on $\HH$}
The goal of this appendix is to show that the local (as well as the global) model structure on $\HH$ is 
compatible with the closed symmetric monoidal structure discussed in Section \ref{subsec:H}.
More precisely, we show that $\HH$ is a {\em symmetric monoidal model category}, 
see e.g.\ \cite[Chapter 4]{Hovey}.
For this purpose we first need the symmetric monoidal model structure on $\Grpd$. 
It is well-known that the closed symmetric monoidal structure and model structure
on $\Grpd$ that we introduced in Section \ref{subsec:Grpd} define a symmetric 
monoidal model category structure on $\Grpd$. As we could not find a proof of this 
statement in the literature, we shall provide it here. 
\begin{propo}\label{propo:Grpdmonoidalmodel}
$\Grpd$ is a symmetric monoidal model category with respect to the closed symmetric monoidal structure 
and the model structure presented in Section \ref{subsec:Grpd}. 
\end{propo}
\begin{proof}
Since all objects in $\Grpd$ are cofibrant, so is the unit object $\{\ast\}$. 
Therefore, to conclude that $\Grpd$ is a symmetric monoidal model category, 
it is sufficient to prove that the monoidal bifunctor $\times: \Grpd \times \Grpd \to \Grpd$ 
is a Quillen bifunctor, cf.\ \cite[Chapter~4.2]{Hovey}. 
\sk

Take two cofibrations $F: G \to H$ and $F^\prime: G^\prime \to H^\prime$ in $\Grpd$ 
and form their pushout product 
\begin{flalign}
F~\Box~F^\prime ~:~ P(F,F^\prime) := (H \times G^\prime) \bigsqcup_{G \times G^\prime} (G \times H^\prime) 
\longrightarrow H \times H^\prime~.
\end{flalign}
We have to show that $F~\Box~F^\prime$ is a cofibration, which is acyclic whenever either $F$ or $F^\prime$ is. 
Recall that cofibrations in $\Grpd$ are functors that are injective on objects. 
Computing the pushout in $\Grpd$, 
one finds that objects of $P(F,F^\prime)$ are equivalence classes of pairs 
of the form $(y , x^\prime) \in H_0 \times G^\prime_0$ 
or of the form $(x,y^\prime) \in G_0 \times H^\prime_0$ under the relation 
$(y,x^\prime) \sim (x,y^\prime) :\iff y = F(x), ~y^\prime =  F^\prime(x^\prime)$. 
(By the subscript $_0$ we denote the set of objects of a groupoid.)
Since $F$ and $F^\prime$ are by hypothesis injective on objects,
it follows by using our equivalence relation that $F~\Box~F^\prime$ is injective on objects too.
Hence, it is a cofibration. 
\sk

The case where one of the cofibrations is acyclic may be simplified by
using that $\Grpd$ is cofibrantly generated (cf.\ \cite{Hollander}) and \cite[Corollary~4.2.5]{Hovey}.
Using also symmetry of the monoidal structure, it is sufficient to show that for the 
generating acyclic cofibration $J : \{\ast\}\to\Delta^1$, given 
by $\ast\mapsto 0$ and $\id_{\ast}\mapsto\id_0$, the pushout product
\begin{flalign}\label{eqn:tmppushprod}
F~\Box~J ~:~ P(F,J) = H\, \bigsqcup_{G} \,(G \times \Delta^1)  \longrightarrow H \times \Delta^1
\end{flalign}
is an acyclic cofibration for any cofibration $F : G\to H$ in $\Grpd$.
The pushout groupoid $P(F,J)$ can be computed explicitly: Its set of objects is
$G_0\sqcup H_0$, i.e.\ an object is either an object $x$ in $G$ or an
object $y$ in $H$. Its morphisms are characterized by
\begin{align}
\nn &\Hom_{P(F,J)}^{}(x,\widetilde{x}) = \Hom_{H}^{}(F(x),F(\widetilde{x}))~,&
&\Hom_{P(F,J)}^{}(x,\widetilde{y}) = \Hom_{H}^{}(F(x),\widetilde{y})~,\\
&\Hom_{P(F,J)}^{}(y,\widetilde{x}) = \Hom_{H}^{}(y,F(\widetilde{x}))~,&
&\Hom_{P(F,J)}^{}(y,\widetilde{y}) = \Hom_{H}^{}(y,\widetilde{y})~.\label{eqn:tmpHomexplicit}
\end{align}
The groupoid morphism $F~\Box~J$ in \eqref{eqn:tmppushprod} is the functor that acts on objects as
\begin{subequations}
\begin{flalign}
\nn F~\Box~J~ : ~ x&\longmapsto (F(x),1)~,\\
  y&\longmapsto (y,0)~,
\end{flalign}
and on morphisms as
\begin{flalign}
\nn F~\Box~J ~:~ \big(h : y\to \widetilde{y}\big)&\longmapsto \big(h\times \id_0 : (y,0)\to (\widetilde{y},0)\big)~,\\
\nn \big(h : x\to \widetilde{y}\big)&\longmapsto \big(h\times ( 1\to 0) : (F(x),1) \to (\widetilde{y},0)\big)~,\\
\nn \big(h : y\to \widetilde{x}\big)&\longmapsto \big(h\times ( 0\to 1) : (y,0) \to (F(\widetilde{x}),1)\big)~,\\
 \big(h : x \to \widetilde{x}\big)&\longmapsto \big(h\times \id_1 : (F(x),1) \to (F(\widetilde{x}),1)\big)~,
\end{flalign}
\end{subequations}
where it is important to recall the definition of morphisms in $P(F,J)$, see \eqref{eqn:tmpHomexplicit}.
It is clear that $F~\Box~J$ is fully faithful and essentially surjective (and of course injective on objects),
hence it is an acyclic cofibration. This completes our proof.
\end{proof}

Proposition \ref{propo:Grpdmonoidalmodel} 
enables us to show that $\HH$ is a symmetric monoidal model category.
\begin{theo}
When equipped with the local or global model structure, $\HH$ is a symmetric monoidal model category 
with respect to the closed symmetric monoidal structure presented in Section \ref{subsec:H}. 
\end{theo}
\begin{proof}
The result follows immediately from Proposition \ref{propo:Grpdmonoidalmodel} 
by using the techniques developed in \cite{Barwick}:
Explicitly, by \cite[Corollary~4.53]{Barwick}, $\HH$ equipped with the global model structure
is a symmetric monoidal model category because our site $\CC$ has all products $U\times U^\prime$.
Using also \cite[Theorem~4.58]{Barwick}, it follow that $\HH$ equipped with the local model structure
is a symmetric monoidal model category.
\end{proof}


\section{\label{app:cofibrep}Cofibrant replacement of manifolds in $\HH$}
Let $M$ be a (finite-dimensional and paracompact) manifold
and $\mathcal{V} = \{V_\alpha\subseteq M\}$ any cover by open subsets.
The goal of this appendix is to prove that the presheaf of \v{C}ech groupoids
associated to $\mathcal V$ is always weakly equivalent to $\underline{M}$ in $\HH$.
Moreover, it provides a cofibrant replacement of $\underline{M}$ when all $V_\alpha$ are diffeomorphic to $\bbR^{\dim(M)}$.
It is important to stress that the latter statement {\em does not} require that the cover is good, in particular 
$V_\alpha\cap V_{\beta}$ may be neither empty nor diffeomorphic to $\bbR^{\dim(M)}$.
We shall always work with the local model structure on $\HH$, see
Theorem \ref{theo:localH}.
\sk

We define the object $\widecheck{\mathcal{V}}$ in $\HH$
by the following functor $\widecheck{\mathcal{V}} : \CC^\op\to \Grpd$:
To any object $U$ in $\CC$, it assigns the groupoid
$\widecheck{\mathcal{V}}(U)$ with objects given by
diagrams
\begin{flalign}
(\alpha,\nu) := \Big(\xymatrix{
M & \ar[l]_-{\rho_\alpha}V_\alpha & \ar[l]_-{\nu}U
}\Big)~,
\end{flalign}
where $\nu: U \to V_\alpha$ is a smooth map and $\rho_\alpha$ is the canonical inclusion $V_\alpha\subseteq M$.
There exists a unique morphism $(\alpha ,\nu) \to (\beta,\nu^\prime)$ in $\widecheck{\mathcal{V}}(U)$ 
if and only if the diagram
\begin{flalign}
\xymatrix@R=0.6pc{
&\ar[ld]_-{\rho_\alpha}~V_\alpha~&\\
M ~&& \ar[lu]_-{\nu}\ar[ld]^-{\nu^\prime} ~U\\
&\ar[lu]^-{\rho_\beta} ~V_{\beta}~&
}
\end{flalign}
commutes. To any $\CC$-morphism $\ell : U\to U^\prime$,
we assign the groupoid morphism 
$\widecheck{\mathcal{V}}(\ell) :\widecheck{\mathcal{V}}(U^\prime) 
\to \widecheck{\mathcal{V}}(U)$ defined by
\begin{flalign}
\widecheck{\mathcal{V}}(\ell) ~:~ (\alpha,\nu)=\Big(\xymatrix{
M & \ar[l]_-{\rho_\alpha}V_\alpha & \ar[l]_-{\nu}U^\prime
}\Big) ~\longmapsto~ \Big(\xymatrix{
M & \ar[l]_-{\rho_\alpha}V_\alpha & \ar[l]_-{\nu\circ\ell}U
}\Big)=(\alpha,\nu\circ\ell)~.
\end{flalign}
The action of $\widecheck{\mathcal{V}}(\ell)$ on morphisms is fixed by their uniqueness.
\sk

There exists a canonical $\HH$-morphism
\begin{flalign}\label{eqn:canqmorph}
q : \widecheck{\mathcal{V}} \longrightarrow \underline{M}~,
\end{flalign}
where $\underline{M}$ is the stack represented by our manifold $M$, cf.\ Example \ref{ex:manifolds}.
Explicitly, recalling that $\underline{M}(U) = C^\infty(U,M)$, for all objects $U$ in $\CC$, 
the stages of $q$ are given by the groupoid morphisms 
\begin{flalign}
q ~:~
(\alpha,\nu)=\Big(\xymatrix{
M & \ar[l]_-{\rho_\alpha}V_\alpha & \ar[l]_-{\nu}U
}\Big) ~\longmapsto~\Big(\xymatrix{
M & \ar[l]_-{\rho_\alpha\circ \nu}U
}\Big) = (\rho_\alpha\circ\nu) ~.
\end{flalign}
Naturality of these stages in $U$ is obvious by definition.
\begin{propo}
Let $M$ be a (finite-dimensional and paracompact) manifold
and $\mathcal{V} = \{V_\alpha\subseteq M\}$ any cover by open subsets.
Then the $\HH$-morphism $q : \widecheck{\mathcal{V}} \to \underline{M}$
is a weak equivalence in the local model structure on $\HH$.
\end{propo}
\begin{proof}
By Theorem \ref{theo:localH}, we have to show that $q$ induces an isomorphism
on the associated sheaves of homotopy groups.
This can be easily confirmed by verifying
the local lifting conditions in \cite[Definition 5.6 and Theorem 5.7]{Hollander}.
Because $q$ is stage-wise fully faithful, it remains to prove the following
property, for all objects $U$ in $\CC$: Given any object $(\rho : U\to M)$ in $\underline{M}(U) = C^\infty(U,M)$,
there exists a good open cover $\{U_i\subseteq U\}$ of $U$ such that all restrictions
$\rho\vert_{U_i} : U_i\to M$ lie in the image of $q$, i.e.\ all 
$\rho\vert_{U_i}$ can be factorized as
\begin{flalign}
\xymatrix{
\ar[dr]_-{\nu}U_i \ar[rr]^-{\rho\vert_{U_i}} && M\\
&V_{\alpha(i)}\ar[ru]_-{\rho_{\alpha(i)}}&
}
\end{flalign}
This property indeed holds true by using any good open refinement $\{U_i\subseteq U\}$ of
the open cover $\{\rho^{-1}(V_\alpha) \subseteq U\}$ of $U$.
\end{proof}

The following proposition shows that  $q : \widecheck{\mathcal{V}} \to \underline{M}$
is a cofibrant replacement when $\mathcal{V}$ consists of open subsets diffeomorphic to $\bbR^{\dim(M)}$.
\begin{propo}\label{propo:cofibcover}
Let $M$ be a (finite-dimensional and paracompact) manifold
and $\mathcal{V} = \{V_\alpha\subseteq M\}$ any cover by open subsets diffeomorphic to $\bbR^{\dim(M)}$.
Then $\widecheck{\mathcal{V}}$ is a cofibrant object
in the local model structure on $\HH$.
\end{propo}
\begin{proof}
We have to prove that for all acyclic fibrations $f : X\to Y$ in $\HH$ and all $\HH$-morphisms
$f^\prime : \widecheck{\mathcal{V}} \to Y$ there exists a lift
\begin{flalign}\label{eqn:tmplift}
\xymatrix{
&&X\ar[d]^-{f}\\
\ar@{-->}[rru]^-{f^{\prime\prime}}\widecheck{\mathcal{V}} \ar[rr]_-{f^\prime}&& Y
}
\end{flalign}
in $\HH$. Because the classes of cofibrations and acyclic fibrations 
coincide in both model structures on $\HH$, the morphism $f : X\to Y$ is a
stage-wise acyclic fibration in $\Grpd$, i.e.\ $f : X(U)\to Y(U)$ is surjective on objects and fully faithful.
\sk

Let us first analyze this lifting problem stage-wise: For any object $U$ in $\CC$,
we can always solve the stage-wise lifting problem
\begin{flalign}\label{eqn:cechlift}
\xymatrix{
&&X(U)\ar[d]^-{f}\\
\ar@{-->}[rru]^-{f^{\prime\prime}}\widecheck{\mathcal{V}}(U) \ar[rr]_-{f^\prime}&& Y(U)
}
\end{flalign}
because every groupoid is cofibrant, see e.g.\ \cite{Hollander}.
Such stage-wise liftings $f^{\prime\prime} : \widecheck{\mathcal{V}}(U)\to X(U)$ 
can be classified as follows:
For every object $(\alpha,\nu)$ in $\widecheck{\mathcal{V}}(U)$, choose any 
object $f^{\prime\prime}(\alpha,\nu)$ in $X(U)$ satisfying 
$f(f^{\prime\prime}(\alpha,\nu)) = f^{\prime}(\alpha,\nu)$. (There always exists such a choice because
$f$ is by hypothesis surjective on objects.) Such choice defines a unique groupoid
morphism $f^{\prime\prime} : \widecheck{\mathcal{V}}(U)\to X(U)$ 
because $f: X(U)\to Y(U)$ is fully faithful and hence \eqref{eqn:cechlift} enforces
a unique choice for the action of $f^{\prime\prime}$ on morphisms. 
\sk

The crucial point is now to prove that the stage-wise liftings 
$f^{\prime\prime} : \widecheck{\mathcal{V}}(U) \to X(U) $ can be chosen
to form a natural transformation, thus providing the stages of a morphism
$f^{\prime\prime} : \widecheck{\mathcal{V}} \to X$ in $\HH$. 
In order to do so, we take advantage of the fact that each $V_\alpha$ is assumed to be an object of $\CC$. 
Given any morphism of the form $\ell : U\to V_\alpha$ in $\CC$ (also regarded as a smooth map of manifolds), 
naturality is expressed as commutativity of the diagram
\begin{flalign}
\xymatrix{
\ar[d]_-{ \widecheck{\mathcal{V}}(\ell)}  \widecheck{\mathcal{V}}(V_\alpha) \ar[rr]^-{f^{\prime\prime}} && X(V_\alpha)\ar[d]^-{X(\ell)}\\
 \widecheck{\mathcal{V}}(U) \ar[rr]_-{f^{\prime\prime}} && X(U)
}
\end{flalign}
in $\Grpd$. Taking the object $(\alpha,\id_{V_\alpha})$ in $ \widecheck{\mathcal{V}}(V_\alpha)$,
this commutative diagram implies that 
\begin{flalign}\label{eqn:tmpcomdiag}
f^{\prime\prime}(\alpha,\ell) = 
X(\ell)\big(f^{\prime\prime}(\alpha,\id_{V_\alpha})\big)~.
\end{flalign}
Hence, the stages $f^{\prime\prime} : \widecheck{\mathcal{V}}(U) \to X(U)$ of a natural lift 
are uniquely determined by their actions on the objects $(\alpha,\id_{V_\alpha})$, for all $\alpha$.
Because there are no further coherence conditions for the restriction of $f^{\prime\prime}$ to objects of the form
$(\alpha,\id_{V_\alpha})$, it follows from our discussion at the end of the previous paragraph of this proof
that there exists a lift for our original problem \eqref{eqn:tmplift}, 
which is determined by \eqref{eqn:tmpcomdiag} and by the choice of a preimage 
along $f: X(V_\alpha) \to Y(V_\alpha)$ of the object $f^\prime(\alpha,\id_{V_\alpha})$,
for each $\alpha$.
As a consequence, $\widecheck{\mathcal{V}}$ is a cofibrant object in $\HH$.
\end{proof}

Because the results of this appendix hold true for every cover $\mathcal{V} = \{V_\alpha\subseteq M\}$
of $M$ by open subsets diffeomorphic to $\bbR^{\dim(M)}$ (without the requirement of a good open cover),
we can make a particular choice which eventually leads to a functorial cofibrant replacement.
For any manifold $M$, we take the open cover $\mathcal{V}_M = \{V\subseteq M\}$ given by {\em all} 
open subsets $V\subseteq M$ diffeomorphic to $\bbR^{\dim(M)}$. 
We denote the corresponding presheaf of \v{C}ech groupoids by
\begin{flalign}
\widecheck{M} := \widecheck{\mathcal{V}}_M
\end{flalign}
and notice that this defines a functor
\begin{flalign}
\widecheck{(-)} : \Man_{\hookrightarrow} \longrightarrow \HH
\end{flalign}
from the category $\Man_{\hookrightarrow}$ of (finite-dimensional and paracompact)
manifolds with morphisms given by open embeddings $f : M\to M^\prime$.
Explicitly, given any open embedding $f : M\to M^\prime$
the stage at $U$ in $\CC$ of the $\HH$-morphism $\widecheck{f} : \widecheck{M}\to \widecheck{M^\prime}$
is defined by 
\begin{flalign}
 \widecheck{f}~ :~ 
(V,\nu)=\Big(\xymatrix{
M & \ar[l]_-{\rho_V}V & \ar[l]_-{\nu}U
}\Big) ~\longmapsto  ~
\Big(\xymatrix{
M^\prime & \ar[l]_-{\rho_{f(V)}}f(V) & \ar[l]_-{f\vert_{V} \circ\nu}U
}\Big)=(f(V),f\vert_{V} \circ\nu)~.
\end{flalign}
Notice that the image $f(V)$ is an open subset of $M^\prime$ diffeomorphic to $\bbR^{\dim(M^\prime)}$ 
by hypothesis on the morphisms in $\Man_{\hookrightarrow}$.  (Notice that $\dim(M) = \dim(M^\prime)$.)
By $f\vert_{V} : V\to f(V)$ we denote the $\CC$-morphisms
associated to $f : M\to M^\prime$ by restricting the domain to $V$ and the codomain to the image $f(V)$.
For the particular cover $\mathcal{V}_M$, we denote the canonical $\HH$-morphism \eqref{eqn:canqmorph} by
\begin{flalign}
q_M : \widecheck{M} \longrightarrow \underline{M}~.
\end{flalign}
It is easy to see that $q_M$ are the components of a natural transformation
$q : \widecheck{(-)} \to \underline{(-)}$
of functors from $\Man_{\hookrightarrow}$ to $\HH$. Summing up, we obtained
\begin{cor}\label{cor:functorialcofibrepMan}
The functor $\widecheck{(-)} : \Man_{\hookrightarrow} \to \HH$
together with the natural transformation $q : \widecheck{(-)} \to \underline{(-)}$
define a functorial cofibrant replacement of manifolds (with open embeddings) in $\HH$.
\end{cor}

\section{\label{app:truncation}Fibrant replacement in the $(-1)$-truncation of $\HH/K$}
We describe an explicit, albeit rather involved, construction of
fibrant replacements in the $(-1)$-truncation of the canonical model structure on the over-category $\HH/K$,
for $K$ an object in $\HH$. These are needed to concretify our mapping stacks in Section \ref{subsec:concrete}.
For the general theory of $n$-truncations in simplicial model categories, 
see \cite[Section 5, Application II]{Barwick} and \cite[Section 7]{Rezk}. 
See also \cite{Lurie} for $n$-truncations in the language of $\infty$-categories.
For an explicit construction of $n$-truncations of simplicial presheaves 
we also refer to \cite[Section 3.7]{ToenVezzosi}.
In this appendix we shall freely use the concept of (left) Bousfield localization of a 
simplicial model category $\mathsf{C}$ with respect to a set of morphisms $S$, 
see \cite{Hirschhorn} or \cite[Section 5]{Dugger} for details.
We shall nevertheless briefly review the relevant terminology, such as $S$-local objects
and $S$-local equivalences. (The set of morphisms \eqref{eqn:S-1maps} at which we shall localize is denoted by $S_{-1}$,
hence we shall speak of $S_{-1}$-local objects and $S_{-1}$- local equivalences in what follows.)
\sk

Let us consider $\HH = \PSh(\CC,\Grpd)$ equipped with the local model structure of Theorem \ref{theo:localH}.
Let $K$ be any object in $\HH$ and form the over-category 
$\HH/K$. Recall that an object in $\HH/K$ is a morphism $f_X : X\to K$ in $\HH$
and a morphism from $f_X : X\to K$ to $f_Y : Y\to K$ in $\HH/K$ is a commutative diagram
\begin{flalign}\label{eqn:fXmapfY}
\xymatrix{
\ar[dr]_-{f_X}X\ar[rr]^-{f} && Y\ar[dl]^-{f_Y}\\
&K&
}
\end{flalign}
in $\HH$. There exists a canonical model structure on $\HH/K$:
A morphism \eqref{eqn:fXmapfY} in $\HH/K$  is a fibration, cofibration or weak equivalence
if and only if $f: X\to Y$ is a fibration, cofibration or weak equivalence in $\HH$.
In particular, the fibrant objects in $\HH/K$ are the fibrations $f_X : X\to K$ with target $K$ in $\HH$.
\sk

We localize this model structure on $\HH/K$ with respect to the set of $\HH/K$-morphisms
\begin{flalign}\label{eqn:S-1maps}
S_{-1} := 
\left\{\text{\parbox{6cm}{$\xymatrix{
\ar[rd]_-{f_{\{0,1\} \times \underline{U}}} \{0,1\} \times \underline{U} \ar[rr]^-{\iota \times \id_{\underline{U}}} && \Delta^1 \times \underline{U}\ar[dl]^-{f_{\Delta^1 \times \underline{U}}}\\
&K&
}$}}
: ~~\text{for all } f_{\Delta^1 \times \underline{U}} : \Delta^1 \times \underline{U} \to K~~
\right\}~,
\end{flalign}
where we note that the morphism $f_{\{0,1\} \times \underline{U}}$ is fixed by commutativity of the diagram.
By $\{0,1\}$ we denote the groupoid with two objects, $0$ and $1$ together with 
their identity morphisms, and $\iota : \{0,1\} \to \Delta^1$ is the obvious groupoid morphism.
We call the $S_{-1}$-localized model structure on $\HH/K$
the {\em $(-1)$-truncation of $\HH/K$}. For $K=\{\ast\}$, 
our definition is a special instance of the general construction of
$n$-truncations for simplicial presheaves given in \cite[Corollary 3.7.4]{ToenVezzosi}. 
(Choose $n=-1$ and truncate also to groupoid-valued presheaves.)
For completeness, we observe that by slightly adapting \cite[Section 3.7]{ToenVezzosi}
we can define sets $S_n$, for all $n\geq -1$, such that the localization of
$\HH/K$ with respect to $S_n$ describes the $n$-truncation of $\HH/K$.
In the following we shall only focus on the case $n=-1$ that we need
in the main part of this paper.
\sk

The fibrant objects in the $(-1)$-truncation of $\HH/K$
coincide with the so-called {\em $S_{-1}$-local objects} in $\HH/K$ 
(cf.\ \cite[Definition 3.1.4 and Proposition 3.4.1]{Hirschhorn}), 
which are defined using the mapping groupoids of $\HH/K$:
For objects $f_X$ and $f_Y$ of $\HH/K$, the mapping groupoid
$\Grpd_{\HH/K}(f_X , f_Y)$
has as objects all commutative diagrams \eqref{eqn:fXmapfY}
in $\HH$ and as morphisms all commutative diagrams
\begin{flalign}
\xymatrix{
\ar[rd]_-{f_X \circ \pr_X} X\times \Delta^1 \ar[rr]^-u && Y\ar[ld]^-{f_Y}\\
&K&
}
\end{flalign}
in $\HH$, where $\pr_X : X\times \Delta^1\to X$ is the projection $\HH$-morphism.
\begin{defi}\label{defi:S-1local}
An {\em $S_{-1}$-local object} in $\HH/K$ is 
a fibration $f_X : X\to K$ in $\HH$ (i.e.\ a fibrant object in the un-truncated model structure
on $\HH/K$) such that the induced morphism
\begin{flalign}\label{eqn:S-1local}
(\iota \times \id_{\underline{U}})^\ast : \Grpd_{\HH/K}(f_{\Delta^1\times\underline{U}} , f_X) \longrightarrow  
\Grpd_{\HH/K}(f_{\{0,1\}\times\underline{U}} , f_X) 
\end{flalign}
between mapping groupoids (see above) is a weak equivalence in $\Grpd$ 
for all $\HH/K$-morphisms in $S_{-1}$, see \eqref{eqn:S-1maps}, 
namely for all $\HH$-morphisms $f_{\Delta^1 \times \underline{U}} : \Delta^1 \times \underline{U} \to K$
and all objects $U$ in $\CC$.
\end{defi}
\begin{lem}\label{lem:S-1localobjects}
A fibration $f_X : X\to K$ in $\HH$ is an $S_{-1}$-local object in $\HH/K$ if and only if
$f_X : X(U) \to K(U)$ is fully faithful, for all objects $U$ in $\CC$.
\end{lem}
\begin{proof}
Using the Yoneda lemma, 
the objects of the source groupoid in \eqref{eqn:S-1local} can be described by
commutative diagrams
\begin{flalign}
\xymatrix{
\Delta^1\ar[dr]_-{F_{\Delta^1}} \ar[rr]^-{F} && X(U)\ar[dl]^-{f_X}\\
&K(U)&
}
\end{flalign}
in $\Grpd$, where $F_{\Delta^1}$ is uniquely fixed by $f_{\Delta^1\times\underline{U}}$.
Equivalently, the objects are all $X(U)$-morphisms
\begin{flalign}
 \xymatrix{
 x_0 \ar[rr]^-{g} && x_1
 }
\end{flalign}
such that $f_X(g) = F_{\Delta^1}(0\to 1) : F_{\Delta^1}(0)\longrightarrow F_{\Delta^1}(1)$.
(In particular, $f_X(x_0) = F_{\Delta^1}(0)$ and $f_X(x_1) = F_{\Delta^1}(1)$.)
The  morphisms of this groupoid are  all commutative $X(U)$-diagrams
\begin{flalign}
 \xymatrix{
\ar[d]_-{h_0} x_0 \ar[rr]^-{g} && x_1\ar[d]^-{h_1}\\
x_0^\prime \ar[rr]_-{g^\prime}&& x_1^\prime
 }
\end{flalign}
such that $f_X(h_0) = \id_{F_{\Delta^1}(0)}$ and $f_X(h_1) = \id_{F_{\Delta^1}(1)}$.
By a similar argument, we find that the objects of the target groupoid in \eqref{eqn:S-1local}
are all  pairs of $X(U)$-objects (no arrow between them!)
\begin{flalign}
 \xymatrix{
 x_0  && x_1
 }
\end{flalign}
such that  $f_X(x_0) = F_{\Delta^1}(0)$ and $f_X(x_1) = F_{\Delta^1}(1)$,
and the morphisms are all $X(U)$-diagrams
\begin{flalign}
 \xymatrix{
\ar[d]_-{h_0} x_0 && x_1\ar[d]^-{h_1}\\
x_0^\prime && x_1^\prime
 }
\end{flalign}
such that $f_X(h_0) = \id_{F_{\Delta^1}(0)}$ and $f_X(h_1) = \id_{F_{\Delta^1}(1)}$.
The groupoid morphism in \eqref{eqn:S-1local} is 
\begin{flalign}
\nn (\iota \times \id_{\underline{U}})^\ast ~~:~~ \xymatrix{
 x_0 \ar[rr]^-{g} && x_1
 } &~~\longmapsto ~~\xymatrix{
 x_0  && x_1
 }\\
 \xymatrix{
\ar[d]_-{h_0} x_0 \ar[rr]^-{g} && x_1\ar[d]^-{h_1}\\
x_0^\prime \ar[rr]_-{g^\prime}&& x_1^\prime
 } &~~\longmapsto ~~ 
 \xymatrix{
\ar[d]_-{h_0} x_0  && x_1\ar[d]^-{h_1}\\
x_0^\prime && x_1^\prime
 }
\end{flalign}
This is a weak equivalence for all $\Grpd$-morphisms  $F_{\Delta^1} : \Delta^1\to K(U)$
if and only if $f_X : X(U) \to K(U)$ is fully faithful. The proof then follows by applying this result to all
objects $U$ in $\CC$.
\end{proof}

In the following we will use frequently an explicit characterization of the fibrations
$f: X\to Y$ in the local model structure on $\HH$ that was obtained in
\cite[Proposition 4.2]{Hollander3}.
\begin{propo}[Local fibrations in $\HH$ \cite{Hollander3}]\label{propo:localfibration}
A morphism $f : X\to Y$ in $\HH$ is a fibration in the local model structure
if and only if
\begin{enumerate}
\item $f : X(U) \to Y(U)$ is a fibration in $\Grpd$, for all objects $U$ in $\CC$; and
\item for all good open covers $\{U_i\subseteq U\}$, the commutative diagram
\begin{flalign}
\xymatrix{
\ar[d]_-{f} X(U)\ar[rr] && \holim_{\Grpd}^{} X(U_\bullet)\ar[d]^-{f}\\
Y(U)\ar[rr] && \holim_{\Grpd}^{} Y(U_\bullet)
}
\end{flalign}
is a homotopy pullback diagram in $\Grpd$.
\end{enumerate}
\end{propo}
\begin{rem}
Item 2.\ is equivalent to the condition that the canonical morphism
\begin{flalign}
X(U)\longrightarrow Y(U) \times^h_{ \holim_{\Grpd}^{} Y(U_\bullet)} \holim_{\Grpd}^{} X(U_\bullet)
\end{flalign}
to the homotopy fiber product is a weak equivalence in $\Grpd$, for all good open covers $\{U_i\subseteq U\}$.
The homotopy fiber product in $\Grpd$ is analogous to the one in $\HH$, see Proposition \ref{propo:homfibprod}.
\end{rem}

We prove a technical lemma that provides us with a useful characterization of the
$S_{-1}$-local objects up to (un-truncated) weak equivalences in $\HH/K$.
Given any object $f_X : X\to K$ in $\HH/K$, we denote by $\Imm(f_X)$ the object in $\HH$ given by
the full image sub-presheaf of $K$. Explicitly, for any object $U$ in $\CC$, the groupoid
$\Imm(f_X)(U)$ has as objects all objects $f_X(x)$ in $K(U)$, for all objects
$x$ in $X(U)$, and as morphisms all morphisms $k : f_X(x) \to f_X(x^\prime)$
in $K(U)$. There is a canonical commutative diagram
\begin{flalign}\label{eqn:S-1localimage}
\xymatrix{
\ar[dr]_-{f_X}X \ar[rr]^-{\widehat{f}_X} && \Imm(f_X)\ar[dl]^-{f_{\Imm(f_X)}}\\
&K&
}
\end{flalign}
in $\HH$, i.e.\ a canonical $\HH/K$-morphism $\widehat{f}_X$
from  $f_X : X\to K$ to  $f_{\Imm(f_X)} : \Imm(f_X) \to K$.
\begin{lem}\label{lem:S-1classification}
Let $f_{X} : X\to K$ be any $S_{-1}$-local object in $\HH/K$. Then
$f_{\Imm(f_X)} : \Imm(f_X) \to K$ is an $S_{-1}$-local object in $\HH/K$
and the canonical $\HH/K$-morphism $\widehat{f}_X$ in \eqref{eqn:S-1localimage}
is a (un-truncated) weak equivalence.
\end{lem}
\begin{proof}
By Lemma \ref{lem:S-1localobjects}, $f_X : X\to K$ is stage-wise fully faithful.
By construction of $\Imm(f_X)$, we then find that $\widehat{f}_X : X\to \Imm(f_X)$ is
stage-wise fully faithful and also  stage-wise surjective on objects, 
hence a weak equivalence in $\HH/K$. 
\sk

It remains to show that $f_{\Imm(f_X)} : \Imm(f_X) \to K$ is an $S_{-1}$-local object. For this 
we make use of Lemma \ref{lem:S-1localobjects} and Proposition \ref{propo:localfibration}.
It is immediately clear by construction that $f_{\Imm(f_X)} : \Imm(f_X) \to K$ is stage-wise fully faithful and a stage-wise
fibration. (To prove the latter statement, use that $f_X : X\to K$ is by hypothesis a stage-wise fibration.)
Thus, it remains to prove that
\begin{flalign}\label{eqn:canmaptmp}
\Imm(f_X)(U) \longrightarrow K(U) \times^h_{ \holim_{\Grpd}^{} K(U_\bullet)} \holim_{\Grpd}^{} \Imm(f_X)(U_\bullet)
\end{flalign}
is a weak equivalence in $\Grpd$, for all good open covers $\{U_i\subseteq U\}$.
The target groupoid can be computed explicitly by using Propositions \ref{propo:descentcat}
and \ref{propo:homfibprod}. One finds that its objects are tuples
\begin{flalign}
\big( y , \big(\{ f_X(x_i)\}, \{g_{ij}\}\big), \{k_i\}  \big)~,
\end{flalign}
where $y$ is an object in $K(U)$, $x_i$ are objects in $X(U_i)$,
$g_{ij} :  f_X(x_i)\vert_{U_{ij}}^{}\to f_X(x_j)\vert_{U_{ij}}^{}$ are morphisms
in $K(U_{ij})$, and $k_i : y\vert_{U_i}^{} \to f_X(x_i)$ are morphisms in $K(U_i)$.
This data has to satisfy 
$g_{ii} =\id_{f_{X}(x_i)}$, for all $i$, 
$g_{jk}\vert_{U_{ijk}}^{} \circ g_{ij}\vert_{U_{ijk}}^{} = g_{ik}\vert_{U_{ijk}}^{}$, for all $i,j,k$,
and $g_{ij}\circ k_{i}\vert_{U_{ij}}^{} = k_{j}\vert_{U_{ij}}^{}$, for all $i,j$.
The morphisms of the target groupoid are tuples
\begin{flalign}
\big(h , \{h_i\}\big) : \big( y , \big(\{ f_X(x_i)\}, \{g_{ij}\}\big), \{k_i\}  \big) \longrightarrow 
\big( y^\prime , \big(\{ f_X(x^\prime_i)\}, \{g^\prime_{ij}\}\big), \{k^\prime_i\}  \big)~,
\end{flalign}
where $h : y\to y^\prime$ is a morphism in $K(U)$ and
$h_i : f_{X}(x_i) \to f_X(x_i^\prime)$ are morphisms in $K(U_i)$,
satisfying  $g^\prime_{ij}\circ h_i\vert_{U_{ij}}^{} = h_{j}\vert_{U_{ij}}^{} \circ g_{ij}$,
for all $i,j$, and $k_i^\prime \circ h\vert_{U_i}^{} = h_i\circ k_i$, for all $i$.
The canonical morphism \eqref{eqn:canmaptmp} is explicitly given by
\begin{flalign}
\nn f_X(x) &~ \longmapsto~ \big( f_X(x), \big(\{f_X(x\vert_{U_i}^{})\}, \{\id\}\big),\{\id\}  \big)~,\\
\big(h : f_X(x) \to f_X(x^\prime)\big)& ~\longmapsto ~\big(h, \{h\vert_{U_i}^{}\}\big)~,
\end{flalign}
from which one immediately observes that it is fully faithful. Using that $f_X$ is $S_{-1}$-local, we can
show that the canonical morphism \eqref{eqn:canmaptmp} is also essentially surjective and thus a weak equivalence:
Given any object $( y , (\{ f_X(x_i)\}, \{g_{ij}\}), \{k_i\}  )$ 
of the target groupoid in \eqref{eqn:canmaptmp}, we obtain from 
the property that $f_X$ is stage-wise fully faithful
an object $(y, (\{x_i\} , \{\widehat{g}_{ij}\}), \{k_i\} )$ of the homotopy fiber product
\begin{flalign}\label{eqn:tmpfiberprod}
K(U) \times^h_{ \holim_{\Grpd}^{} K(U_\bullet)} \holim_{\Grpd}^{} X(U_\bullet)~
\end{flalign}
defined by means of the $H$-morphism $f_{X} : X\to K$.
Explicitly, the $X(U_{ij})$-morphisms $\widehat{g}_{ij} : x_i\vert_{U_{ij}}^{} \to x_{j}\vert_{U_{ij}}^{}$
are uniquely determined by full faithfulness and $f_X(\widehat{g}_{ij}) = g_{ij}$.
Because $f_X$ is a (local) fibration, there exists a morphism $(h , \{\widehat{h}_i\})$ in \eqref{eqn:tmpfiberprod}
from $(f_X(x), (\{x\vert_{U_i}^{}\},\{\id\}), \{\id\})$,
where $x$ is an object in $X(U)$, to $(y, (\{x_i\} , \{\widehat{g}_{ij}\}), \{k_i\} )$.
Explicitly, $h : f_X(x) \to y$ is a morphism in $K(U)$ and $\widehat{h}_{i} : x\vert_{U_i}^{} \to x_i$ are
morphisms in $X(U_i)$, such that $\widehat{g}_{ij}\circ \widehat{h}_i\vert_{U_{ij}}^{} = \widehat{h}_{j}\vert_{U_{ij}}^{}$,
for all $i,j$, and $k_i\circ h\vert_{U_i}^{} = f_X(\widehat{h}_i) $, for all $i$.
The associated morphism $(h,\{f_X(\widehat{h}_i)\})$ in the target groupoid of \eqref{eqn:canmaptmp} 
from $(f_X(x), (\{ f_X(x\vert_{U_i}^{})\},\{\id\}) ,\{\id\})$ to $( y , (\{ f_X(x_i)\}, \{g_{ij}\}), \{k_i\}  )$ 
proves that the canonical 
morphism in \eqref{eqn:canmaptmp} is essentially surjective. This completes our proof.
\end{proof}

The weak equivalences in the $(-1)$-truncation of $\HH/K$ are by construction
the so-called {\em $S_{-1}$-local equivalences} in $\HH/K$.
We say that a morphism \eqref{eqn:fXmapfY} in $\HH/K$ is an $S_{-1}$-local equivalence if
\begin{flalign}
Q(f)^\ast : \Grpd_{\HH/K}(f_{Q(Y)}, f_Z) \longrightarrow \Grpd_{\HH/K}(f_{Q(X)}, f_Z)
\end{flalign}
is a weak equivalence in $\Grpd$, for all $S_{-1}$-local objects $f_Z : Z\to K$ in $\HH/K$.
Here $Q$ is a cofibrant replacement functor in $\HH$
and $Q(f)$ is the corresponding $\HH/K$-morphism
\begin{flalign}
\xymatrix{
\ar@/_2pc/[ddrr]_-{f_{Q(X)}} \ar[dr]^-{q_X}Q(X) \ar[rrrr]^-{Q(f)} &&&& Q(Y)\ar[dl]_-{q_Y} \ar@/^2pc/[ddll]^-{f_{Q(Y)}} \\
&\ar[dr]^-{f_X} X\ar[rr]^-{f} &&Y\ar[dl]_-{f_Y}&\\
&&K&&
}
\end{flalign}
\begin{rem}\label{rem:simpleS-1eqv}
By Lemma \ref{lem:S-1classification}, $S_{-1}$-local equivalences may be equivalently characterized
by the condition that the groupoid morphism
\begin{flalign}
Q(f)^\ast : \Grpd_{\HH/K}(f_{Q(Y)}, f_{\Imm(f_Z)}) \longrightarrow \Grpd_{\HH/K}(f_{Q(X)}, f_{\Imm(f_Z)})
\end{flalign}
is a weak equivalence, for all $S_{-1}$-local objects $f_Z : Z\to K$ in $\HH/K$.
This condition is easier to analyze because $f_{\Imm(f_Z)} : \Imm(f_Z) \to K$
is not only an $S_{-1}$-local object, but also a stage-wise injection on objects.
\end{rem}
\begin{lem}\label{lem:emptyorpoint}
Let $f_X : X\to K$ be any object in $\HH/K$ and $f_Z : Z\to K$ any $S_{-1}$-local object in $\HH/K$.
Then the mapping groupoid $\Grpd_{\HH/K}(f_X,f_{\Imm(f_Z)})$ is either $\emptyset$ or (isomorphic to) $\{\ast\}$.
\end{lem}
\begin{proof}
We have to show that whenever $\Grpd_{\HH/K}(f_X,f_{\Imm(f_Z)})$ is not empty it is (isomorphic to) $\{\ast\}$.
Consider two objects in $\Grpd_{\HH/K}(f_X,f_{\Imm(f_Z)})$, i.e.\
\begin{flalign}
\xymatrix{
\ar[dr]_-{f_X }X\ar[rr]^-{f} && \Imm(f_Z)\ar[dl]^-{f_{\Imm(f_Z)}}&&\ar[dr]_-{f_X }X\ar[rr]^-{\widetilde{f}} && \Imm(f_Z)\ar[dl]^-{f_{\Imm(f_Z)}}\\
&K&&&&K&
}
\end{flalign}
and let $U$ be any object in $\CC$. Commutativity of the diagrams implies that
$f_{\Imm(f_Z)} f(x)  = f_X(x) = f_{\Imm(f_Z)} \widetilde{f}(x)$, hence
by stage-wise injectivity on objects of $f_{\Imm(f_Z)}$ we find that $f(x) = \widetilde{f}(x)$, for all
objects $x$ in $X(U)$. Given any morphism $g: x\to x^\prime$ in $X(U)$,
this result implies that the two morphisms
$f(g) : f(x) \to f(x^\prime)$ and $\widetilde{f}(g) : \widetilde{f}(x) \to \widetilde{f}(x^\prime)$
in $\Imm(f_Z)(U)$ have the same source and target. Commutativity of the diagrams
implies that $f_{\Imm(f_Z)} f(g) = f_X(g) = f_{\Imm(f_Z)} \widetilde{f}(g)$, hence we find by 
stage-wise full faithfulness of $f_{\Imm(f_Z)}$ that $f(g)=\widetilde{f}(g)$. As a consequence, 
$f = \widetilde{f}$ which implies that the groupoid $\Grpd_{\HH/K}(f_X,f_{\Imm(f_Z)})$ cannot have more than one object.
\sk

It remains to prove that the object $f$ 
in $\Grpd_{\HH/K}(f_X,f_{\Imm(f_Z)})$  cannot have non-trivial automorphisms.
Consider any morphism 
\begin{flalign}
\xymatrix{
\ar[rd]_-{f_X \circ \pr_X} X\times \Delta^1 \ar[rr]^-u && \Imm(f_Z) \ar[ld]^-{f_{\Imm(f_Z)}}\\
&K&
}
\end{flalign}
in $\Grpd_{\HH/K}(f_X,f_{\Imm(f_Z)})$, which is necessarily an automorphism of $f$ due to the result above. 
Let $U$ be any object in $\CC$. This automorphism of $f$ is uniquely
determined by the $\Imm(f_Z)(U)$-morphisms 
$\eta_x := u(\id_x \times (0\to1)) : f(x)=u(x,0) \to u(x,1)=f(x)$, for all objects $x$ in $X(U)$.
Commutativity of the diagram implies that
$f_{\Imm(f_Z)}(\eta_x) = \id_{f_X(x)}$, for all objects $x$ in $X(U)$.
As a consequence of stage-wise full faithfulness of $f_{\Imm(f_Z)}$, 
this implies $\eta_x=\id_{f(x)}$, for all objects $x$ in $X(U)$. 
Hence, the only automorphism of $f$ is the identity.
\end{proof}

We can now give a sufficient condition for a morphism
in $\HH/K$ to be an $S_{-1}$-local equivalence.
Recall from Definition \ref{def:sheavesofhomgroups} the notion 
of sheaves of homotopy groups associated to objects in $\HH$.
\begin{propo}\label{propo:S-1equivalences}
A morphism \eqref{eqn:fXmapfY}
in $\HH/K$ is an $S_{-1}$-local equivalence in $\HH/K$
if  $f : X\to Y$ induces an epimorphism of sheaves
of $0$-th homotopy groups.
\end{propo}
\begin{proof}
Because $f:X\to Y$ induces by hypothesis 
an epimorphism of sheaves of $0$-th homotopy groups,
so does its cofibrant replacement $Q(f)$.
Let us denote the corresponding {\em pre}sheaf morphism by 
$\pi_0(Q(f)) : \pi_0(Q(X))\to \pi_0(Q(Y))$ and recall that
its sheafification is a sheaf epimorphism if and only if the following
condition holds true, see e.g.\ \cite[Section III.7]{MacLaneMoerdijk}:
For each object $U$ in $\CC$ and each element $[y]\in \pi_0(Q(Y))(U)$,
there exists a good open cover $\{U_i\subseteq U\}$ and a tuple
$\{[x_i]\in \pi_0(Q(X))(U_i)\}$ such that
$[y]\vert_{U_i}^{} = \pi_0(Q(f))[x_i]$, for all $i$. Equivalently, 
for each object $U$ in $\CC$ and each 
object $y$ in $Q(Y)(U)$, there exists a good open cover
$\{U_i\subseteq U\}$, objects $x_i$ in $Q(X)(U_i)$
and morphisms $h_i : y\vert_{U_i}^{} \to Q(f)(x_i)$
in $Q(Y)(U_i)$.
\sk

Using Remark \ref{rem:simpleS-1eqv} and Lemma \ref{lem:emptyorpoint},
our claim would follow by proving that
\begin{flalign}
Q(f)^\ast : \Grpd_{\HH/K}(f_{Q(Y)}, f_{\Imm(f_Z)}) \longrightarrow \Grpd_{\HH/K}(f_{Q(X)}, f_{\Imm(f_Z)})
\end{flalign}
is never a groupoid morphism from $\emptyset$ to $\{\ast\}$, for all $S_{-1}$-local objects $f_Z: Z\to K$. 
Let us therefore assume that the target groupoid is $\{\ast\}$. We have to prove that there exists a 
dashed arrow
\begin{flalign}\label{eqn:tmpdiagramS-1eqv}
\xymatrix{
\ar@/^2pc/[rrrr]^-{\widetilde{f}}\ar[rrd]_-{f_{Q(X)}} Q(X) \ar[rr]^-{Q(f)} && \ar[d]_-{f_{Q(Y)}} Q(Y) \ar@{-->}[rr]^-{\widehat{f}} && \Imm(f_Z)\ar[dll]^-{~~f_{\Imm(f_Z)}}\\
&&K&&
}
\end{flalign}
completing the commutative diagram. 
It is sufficient to define the
dashed arrow $\widehat{f}$ on objects 
as its action on morphisms is then 
fixed uniquely by the commutative diagram
and stage-wise full faithfulness of $f_{\Imm(f_Z)}$.
Moreover, if the dashed arrow $\widehat{f}$ exists
it is unique because of Lemma \ref{lem:emptyorpoint}.
Let $U$ be any object in $\CC$ and $y$ any object in $Q(Y)(U)$.
Because $Q(f)$ induces an epimorphism on sheaves of $0$-th homotopy groups, 
by our discussion above there exists a good open cover $\{U_i\subseteq U\}$,
a family of objects $x_i$ in $Q(X)(U_i)$ and a family of morphism
$h_i : y\vert_{U_i}^{} \to Q(f)(x_i)$ in $Q(Y)(U_i)$.
We use these data and our diagram \eqref{eqn:tmpdiagramS-1eqv} to define an object 
\begin{flalign}\label{eqn:tmphofibobj}
\big(f_{Q(Y)}(y), \big(\{  \widetilde{f}(x_i)\} , \{ g_{ij}\}\big),\{k_i \}\big)
\end{flalign}
of the homotopy fiber product
\begin{flalign}\label{eqn:homfibprodtmp}
K(U)\times^h_{\holim_{\Grpd}^{} K(U_\bullet)} \holim_{\Grpd}^{} \Imm(f_Z)(U_\bullet)~.
\end{flalign}
Explicitly, the $\Imm(f_Z)(U_{ij})$-morphisms 
$g_{ij} : \widetilde{f}(x_i)\vert_{U_{ij}}^{} \to \widetilde{f}(x_j)\vert_{U_{ij}}^{}$ 
are defined by using stage-wise full faithfulness of $f_{\Imm(f_Z)}$  and the commutative diagrams
\begin{flalign}
\xymatrix{
\ar@{=}[d] f_{\Imm(f_Z)}\widetilde{f}(x_i)\vert_{U_{ij}}^{} \ar@{-->}[rrrr]^-{f_{\Imm(f_Z)}(g_{ij})}&&&&f_{\Imm(f_Z)}\widetilde{f}(x_j)\vert_{U_{ij}}^{} \ar@{=}[d] \\
f_{Q(Y)} Q(f)(x_i)\vert_{U_{ij}}^{} &&
\ar[ll]^-{f_{Q(Y)}(h_i)\vert_{U_{ij}}^{}} f_{Q(Y)} (y)\vert_{U_{ij}}^{} \ar[rr]_-{f_{Q(Y)}(h_j)\vert_{U_{ij}}^{}}
&& f_{Q(Y)} Q(f)(x_j)\vert_{U_{ij}}^{}
}
\end{flalign}
The $K(U_i)$-morphisms $k_i : f_{Q(Y)}(y\vert_{U_i}^{}) \to f_{\Imm(f_Z)}\widetilde{f}(x_i)$
are defined by the commutative diagrams
\begin{flalign}
\xymatrix{
\ar[drr]_-{f_{Q(Y)}(h_i)~~~~~} f_{Q(Y)}(y\vert_{U_i}^{}) \ar@{-->}[rr]^-{k_i} && f_{\Imm(f_Z)}\widetilde{f}(x_i)\ar@{=}[d]\\
&&f_{Q(Y)}Q(f)(x_i)
}
\end{flalign}
Because $f_{\Imm(f_Z)}: \Imm(f_Z) \to K$ is  a (local) fibration in $\HH$,
there exists a morphism in the groupoid \eqref{eqn:homfibprodtmp} from an object
$(f_{\Imm(f_Z)}(z), (\{z\vert_{U_{i}}^{}\},\{\id\}),\{\id\})$, where $z$ is an object in $\Imm(f_Z)(U)$, 
to \eqref{eqn:tmphofibobj}. Using that $f_{\Imm(f_Z)}$
is a stage-wise fibration, the object $z$ in $\Imm(f_Z)(U)$ may be chosen such
that $f_{\Imm(f_Z)}(z) = f_{Q(Y)}(y)$. Because $f_{\Imm(f_Z)}$ is stage-wise injective on objects, such $z$ 
is unique and we may define the dashed arrow by setting $\widehat{f}(y) = z$. Naturality of our construction
of $\widehat{f}$ follows immediately from uniqueness, which completes our proof.
\end{proof}

We developed sufficient technology to obtain a functorial fibrant replacement in the $(-1)$-truncation
of $\HH/K$. Let $f_X : X\to K$ be any object in $\HH/K$.
We define a new object in $\HH/K$, which we call the {\em $1$-image of $f_X$}
and denote as $f_{\Imm_1(f_X)} : \Imm_1(f_X)\to K$, by the following
construction: For an object $U$ in $\CC$, the groupoid
$\Imm_1(f_X)(U)$ has as objects all objects $z$ in $K(U)$
for which there exist a good open cover $\{U_i\subseteq U\}$,
objects $x_{i}$ in $X(U_i)$ and $K(U_i)$-morphisms
$h_{i} : z\vert_{U_i}^{} \to f_X(x_i)$. The morphisms
between two objects $z$ and $z^\prime$ in 
$\Imm_1(f_X)(U)$ are all $K(U)$-morphisms $k : z\to z^\prime$.
For a morphism $\ell : U\to U^\prime$ in $\CC$,
the groupoid morphism $\Imm_1(f_X)(\ell) : \Imm_1(f_X)(U^\prime)\to \Imm_1(f_X)(U)$
is the one induced by $K(\ell) : K(U^\prime)\to K(U)$. (To show that $K(\ell)(z)$ is an 
object in $\Imm_1(f_X)(U)$ for any object $z$ in $\Imm_1(f_X)(U^\prime)$ 
use refinements of open covers to good open covers.)
The $\HH$-morphism $f_{\Imm_1(f_X)} : \Imm_1(f_X)\to K$ is then given by
stage-wise full subcategory embedding. There is a canonical commutative diagram
\begin{flalign}\label{eqn:H/Kfibrep}
\xymatrix{
\ar[dr]_-{f_X}X \ar[rr]^-{\widetilde{f}_X} && \Imm_1(f_X)\ar[dl]^-{~~f_{\Imm_1(f_X)}}\\
&K&
}
\end{flalign}
in $\HH$, i.e.\ a canonical $\HH/K$-morphism $\widetilde{f}_X$
from  $f_X : X\to K$ to  $f_{\Imm_1(f_X)} : \Imm_1(f_X) \to K$.
\begin{propo}\label{propo:fibrepH/K}
Let $f_X : X\to K$ be any object in $\HH/K$. Then \eqref{eqn:H/Kfibrep}
is a fibrant replacement in the $(-1)$-truncation of $\HH/K$.
More explicitly, this means that $f_{\Imm_1(f_X)} : \Imm_1(f_X) \to K$ is an $S_{-1}$-local object
in $\HH/K$ and $\widetilde{f}_X$ is an $S_{-1}$-local equivalence in $\HH/K$.
\end{propo}
\begin{proof}
By construction of $\Imm_1(f_X)$, it is clear that $\widetilde{f}_X$ induces an
epimorphism on sheaves of $0$-th homotopy groups, hence $\widetilde{f}_X$ 
is an $S_{-1}$-local equivalence because of Proposition \ref{propo:S-1equivalences}.
We now prove that $f_{\Imm_1(f_X)} : \Imm_1(f_X) \to K$ is $S_{-1}$-local, 
cf.\ Lemma \ref{lem:S-1localobjects}. By construction, it is clear that
$f_{\Imm_1(f_X)}$ is stage-wise fully faithful and a stage-wise fibration. It remains to verify
item 2.\ of Proposition \ref{propo:localfibration}, i.e.\ that the canonical morphism
\begin{flalign}\label{tmp:Imm1canmap}
\Imm_1(f_X)(U) \longrightarrow K(U)\times^h_{\holim_{\Grpd}^{} K(U_\bullet)} \holim_{\Grpd}^{}\Imm_1(f_X)(U_\bullet)
\end{flalign}
is a weak equivalence in $\Grpd$, for all good open covers $\{U_i\subseteq U\}$.
Similarly to the proof of Lemma \ref{lem:S-1classification}, objects in the target groupoid
are tuples $(z,(\{z_i\},\{g_{ij}\}),\{k_i\})$, where $z$ is an object in $K(U)$,
$z_i$ are objects in $\Imm_1(f_X)(U_i)$, $g_{ij} : z_i\vert_{U_{ij}}^{}\to z_i\vert_{U_{ij}}^{}$ 
are  $K(U_{ij})$-morphisms and $k_i : z\vert_{U_i}^{}\to z_i$ are $K(U_i)$-morphisms.
This data has to satisfy $g_{ii} =\id$, for all $i$, $g_{jk}\vert_{U_{ijk}}^{}\circ g_{ij}\vert_{U_{ijk}}^{} =
g_{ik}\vert_{U_{ijk}}^{}$, for all $i,j,k$, and $g_{ij}\circ k_i\vert_{U_{ij}}^{}=k_j\vert_{U_{ij}}^{}$, for all $i,j$.
The morphisms are tuples 
$(h,\{h_i\}) : (z,(\{z_i\},\{g_{ij}\}),\{k_i\}) \to (z^\prime,(\{z^\prime_i\},\{g^\prime_{ij}\}),\{k^\prime_i\})$,
where $h:z\to z^\prime$ is a $K(U)$-morphism and $h_i : z_i\to z_i^\prime$ are $K(U_i)$-morphisms,
satisfying $g^\prime_{ij}\circ h_i\vert_{U_{ij}}^{} = h_{j}\vert_{U_{ij}}^{}\circ g_{ij}$, for all $i,j$,
and $k_i^\prime\circ h\vert_{U_i}^{} = h_i\circ k_i$, for all $i$.
We prove that the canonical morphism \eqref{tmp:Imm1canmap} is essentially surjective.
For each object $(z,(\{z_i\},\{g_{ij}\}),\{k_i\})$ of the target groupoid, there exists
a morphism $(\id_z,\{k_i\}): (z,(\{z\vert_{U_i}^{}\},\{\id\}),\{\id\}) \to (z,(\{z_i\},\{g_{ij}\}),\{k_i\})$.
Notice that $z$ is an object in $\Imm_1(f_X)(U)$: There are $K(U_i)$-morphisms 
$k_i : z\vert_{U_i}^{}\to z_i$ to objects $z_i$ in $\Imm_1(f_X)(U_i)$, for all $i$. By definition, 
there exists a good open cover $\{U_{i,\alpha_i}\subseteq U_i\}$,
objects $x_{i,\alpha_i}$ in $X(U_{i,\alpha_i})$ and $K(U_{i,\alpha_i})$-morphisms 
$h_{i,\alpha_i} : z_i\vert_{U_{i,\alpha_i}}^{} \to f_X(x_{i,\alpha_i})$, for all $i$.
Taking any good open cover of $U$ refining the (not necessarily good) 
open cover $\{U_{i,\alpha_i}\subseteq U\}$, one shows that $z$ is indeed 
an object in $\Imm_1(f_X)(U)$ and hence that \eqref{tmp:Imm1canmap} 
is essentially surjective. Full faithfulness of the canonical morphism 
\eqref{tmp:Imm1canmap} is easy to confirm, which completes our proof.
\end{proof}

The $1$-image  $f_{\Imm_1(f_X)} : \Imm_1(f_X)\to K$ 
of an object $f_X : X\to K$ in $\HH/K$, which according to Proposition
\ref{propo:fibrepH/K} is a fibrant replacement in the $(-1)$-truncation of $\HH/K$,
is not very convenient to work with practically. 
We conclude this appendix by showing that
the canonical morphism
\begin{flalign}\label{eqn:ImmtoImm1}
\xymatrix{
\ar[dr]_-{f_{\Imm(f_X)}} \Imm(f_X) \ar[rr] && \Imm_1(f_X)\ar[dl]^-{f_{\Imm_1(f_X)}}\\
&K&
}
\end{flalign}
from the full image of $f_X : X \to K$ to its $1$-image 
is a weak equivalence in the un-truncated model structure on $\HH/K$.
This allows us to use the simpler concept of 
full image instead of the $1$-image in the main body of this paper.
\begin{propo}\label{propo:easyfibrepH/K}
Let $f_X : X\to K$ be any object in $\HH/K$. Then the canonical $\HH/K$-morphism 
\eqref{eqn:ImmtoImm1} is a weak equivalence in the un-truncated model structure on $\HH/K$.
\end{propo}
\begin{proof}
One has to show that the corresponding 
$\HH$-morphism  $\Imm(f_X)  \to \Imm_1(f_X)$ is a weak equivalence in the local model structure
on $\HH$, i.e.\ that it induces an isomorphism on sheaves of homotopy groups (cf.\ Theorem \ref{theo:localH}).
This can be easily confirmed by verifying
the local lifting conditions in \cite[Definition 5.6 and Theorem 5.7]{Hollander},
making use of the facts that  $\Imm(f_X)  \to \Imm_1(f_X)$ is stage-wise fully faithful
and that it induces an epimorphism on sheaves of $0$-th homotopy groups.
\end{proof}

Let $g: K \to \widetilde{K}$ be a weak equivalence in $\HH$. 
We establish a relation induced by $g$ between our fibrant replacements 
in the $(-1)$-truncations of  $\HH/K$ and $\HH/\widetilde{K}$.
\begin{propo}\label{propo:basechange}
Let $g: K \to \widetilde{K}$ be a weak equivalence in $\HH$ and $f_X:X \to K$ an object in $\HH/K$.
Then there exist canonical weak equivalences $\Imm_1(f_X) \to \Imm_1(g \circ f_X)$ 
and $\Imm(f_X) \to \Imm(g \circ f_X)$ in $\HH$. 
\end{propo}
\begin{proof}
Notice that there exists a unique dashed $\HH$-morphism such that the diagram
\begin{flalign}\label{tmp:KtotildeKdiagram}
\xymatrix@C=4pc{
\ar[d]_-{\id_X}X\ar[r]^-{\widetilde{f_X}} & \ar@{-->}[d]\Imm_1(f_X) \ar[r]^-{f_{\Imm_1(f_X)}}& K\ar[d]^-{g}\\
X\ar[r]_-{\widetilde{g\circ f_X}} & \Imm_1(g \circ f_X) \ar[r]_-{f_{\Imm_1(g \circ f_X)}} & \widetilde{K}
}
\end{flalign}
in $\HH$ commutes. Recall that $\widetilde{f_X}$ and $\widetilde{g \circ f_X}$ induce 
epimorphisms on sheaves of $0$-th homotopy groups. Moreover, $f_{\Imm_1(f_X)}$ 
and $f_{\Imm_1(g \circ f_X)}$ are stage-wise fully faithful, hence
they induce monomorphisms on sheaves of $0$-th homotopy groups and isomorphisms
on all sheaves of $1$-st homotopy groups. Taking sheaves of homotopy groups 
in \eqref{tmp:KtotildeKdiagram}, a diagram chasing argument shows that
the dashed morphism induces isomorphisms on sheaves of homotopy groups.
Hence, the canonical morphism  $\Imm_1(f_X) \to \Imm_1(g \circ f_X)$  
is a weak equivalence in $\HH$, cf.\ Theorem \ref{theo:localH}.
\sk

Replacing the $1$-image $\Imm_1$ in \eqref{tmp:KtotildeKdiagram}
by the full image $\Imm$, one can also show that there 
exists a unique $\HH$-morphism $\Imm(f_X) \to \Imm(g \circ f_X)$ such that the 
corresponding diagram commutes. Moreover, it is compatible with the
canonical morphism  $\Imm_1(f_X) \to \Imm_1(g \circ f_X)$, i.e.\ the diagram
\begin{flalign}
\xymatrix{
\Imm(f_X) \ar[r] \ar[d] & \Imm_1(f_X) \ar[d] \\
\Imm(g \circ f_X) \ar[r] & \Imm_1(g \circ f_X) 
}
\end{flalign}
in $\HH$ commutes. Because $\Imm_1(f_X) \to \Imm_1(g \circ f_X)$ and
the horizontal arrows are weak equivalences (cf.\ Proposition \ref{propo:easyfibrepH/K}),
it follows by the 2-out-of-3 property of weak equivalences that
$\Imm(f_X) \to \Imm(g \circ f_X)$ is a weak equivalence too. 
\end{proof}


\section{\label{app:concretification}Concretification}
For the sake of completeness, we compare 
our concretification prescription of Section \ref{subsec:concrete} 
with the original construction proposed by \cite{Fiorenza, Schreiber}. 
In particular, we highlight why the latter fails to produce the desired result,
i.e.\ the stack describing smoothly parametrized families of principal $G$-bundles
with connections, and how the former fixes this aspect. Since this issue arises already 
for manifolds $M$ diffeomorphic to $\bbR^n$, for simplicity here we restrict to this case. 
As any such manifold $M$ may be regarded as an object in $\CC$,
the corresponding object $\underline{M}$ in $\HH$ is representable
and thus already cofibrant. As a consequence, the mapping stack ${\BGcon}^{\widecheck{M}}$
is weakly equivalent to the internal hom-object ${\BGcon}^{\underline{M}}$ in $\HH$.
\sk

According to \cite{Fiorenza, Schreiber}, the concretification of ${\BGcon}^{\underline{M}}$ 
is obtained by a two-step construction: First, one forms the fibrant replacements 
of both $\zeta: {\BG}^{\underline{M}} \to \sharp ({\BG}^{\underline{M}})$ 
and $\zeta_\mathrm{con}: {\BGcon}^{\underline{M}} \to \sharp ({\BGcon}^{\underline{M}})$ 
in the respective $(-1)$-truncations of $\HH/\sharp ({\BG}^{\underline{M}})$ 
and $\HH/\sharp ({\BGcon}^{\underline{M}})$, cf.\ Appendix \ref{app:truncation}.
(In the language of \cite{Fiorenza, Schreiber} this is called ``$1$-image factorization".) 
Using the weakly equivalent model provided by Proposition \ref{propo:easyfibrepH/K}, 
we obtain the factorizations 
\begin{flalign}
\xymatrix@R1pc{
{\BG}^{\underline{M}} \ar[r]^-{\zeta_1} ~&~ \sharp_1 \big({\BG}^{\underline{M}}\big) := 
\Imm(\zeta) \ar[r] ~&~ \sharp \big({\BG}^{\underline{M}}\big)~,	\\
{\BGcon}^{\underline{M}} \ar[r] ~&~ \sharp_1 \big({\BGcon}^{\underline{M}}\big) := 
\Imm(\zeta_\mathrm{con}) \ar[r] ~&~ \sharp \big({\BGcon}^{\underline{M}}\big)~,
}
\end{flalign} 
where we introduce the notation $\sharp_1$ for ease of comparison with \cite{Fiorenza, Schreiber}. 
Explicitly, at each stage $U$ in $\CC$, the groupoid 
$\sharp_1 ({\BG}^{\underline{M}})(U) = \sharp ({\BG}^{\underline{M}})(U)$ has only one object $\ast$ 
and morphisms given by families $\{g_p \in C^\infty(M,G)\big\}$, where $p \in U$ runs over all points of $U$. 
The groupoid $\sharp_1 ({\BG}^{\underline{M}})(U)$ has as objects all smoothly 
$U$-parametrized gauge fields $A \in \Omega^{1,0}(M \times U,\g)$, 
where $\Omega^{1,0}$ denotes the vertical $1$-forms on $M \times U \to U$, 
and as morphisms from $A$ to $A^\prime$ all families $\{g_p \in C^\infty(M,G)\}$, 
where $p\in U$ runs over all points of $U$, 
such that $ (\id_M \times p)^\ast A^\prime = ((\id_M \times p)^\ast A )\ra g_p$,
for all points $p:\bbR^0\to U$.
Notice that $\mathrm{forget}^{\underline{M}}: {\BGcon}^{\underline{M}} \to {\BG}^{\underline{M}}$ 
induces an $\HH$-morphism 
\begin{subequations}\label{eqn:sharp1forget}
\begin{flalign}
\sharp_1 \big({\mathrm{forget}^{\underline{M}}}\big) \, :\, 
\sharp_1 \big({\BGcon}^{\underline{M}}\big) ~\longrightarrow~ \sharp_1 \big({\BG}^{\underline{M}}\big)~.
\end{flalign}
Explicitly, at stage $U$ in $\CC$, this is the groupoid morphism
\begin{flalign}\nn
\sharp_1 \big({\mathrm{forget}^{\underline{M}}}\big)~ :~ 
A & \longmapsto \ast~, \\ 
\{g_p\} & \longmapsto \{g_p\}~.
\end{flalign}
\end{subequations}
Let us stress that this is \textit{not} a fibration in $\HH$. Loosely speaking, this happens because 
acting with a non-smoothly parametrized gauge transformation $\{g_p\}$ 
on a smoothly parametrized gauge field in general results in a non-smoothly parametrized gauge field. 
\sk

Notice that $\sharp_1 ({\BGcon}^{\underline{M}})$ has the desired objects, 
i.e.\ smoothly parametrized families of gauge fields,
however the morphisms are general families of gauge transformations $\{g_p\}$
between smoothly parametrized gauge fields.
The construction of \cite{Fiorenza, Schreiber} claims to fix this issue 
in a second step, that consists of taking the homotopy fiber product 
in $\HH$ of the diagram 
\begin{flalign}\label{tmp:hofibprodapp}
\xymatrix{
\sharp_1 \big({\BGcon}_{~}^{\underline{M}}\big) \ar[rr]^-{\sharp_1 (\mathrm{forget}^{\underline{M}})} 
~&&~ \sharp_1 \big({\BG}_{~}^{\underline{M}}\big)~ &&~\ar[ll]_-{\zeta_1} {\BG}_{~}^{\underline{M}}
}~,
\end{flalign}
which we shall  denote  by $\widetilde{\GCon}(M)$.
Note that the ordinary fiber product would produce the desired result, 
i.e.\ smoothly parametrized gauge fields and gauge transformations.
However, none of the $\HH$-morphisms in this pullback diagram is a fibration,
which prevents us from computing the homotopy fiber product as the ordinary one 
(as opposed to the situation in Section \ref{subsec:concrete}). 
We may still compute $\widetilde{\GCon}(M)$ by using Proposition \ref{propo:homfibprod}. 
Explicitly, at stage $U$ in $\CC$, 
the objects of the groupoid $\widetilde{\GCon}(M)(U)$ are tuples 
\begin{flalign}
\Big( A \in \Omega^{1,0}(M \times U,\g), ~ \big\{h_p \in C^\infty(M,G)\big\} \Big)~,
\end{flalign}
where $p\in U$ runs over all points of $U$, and
the morphisms from $(A, \{h_p\})$ to $(A^\prime, \{h^\prime_p\})$ are tuples 
\begin{flalign}
\Big( \big\{g_p \in C^\infty(M,G)\big\}, ~ \widetilde{g} \in C^\infty(M \times U,G)  \Big)~,
\end{flalign}
where $p\in U$ runs over all points of $U$, satisfying
\begin{subequations}
\begin{flalign}
 (\id_M \times p)^\ast A^\prime ~&=~ ((\id_M \times p)^\ast A)  \ra g_p ~, \\
 g_p  ~ h^\prime_p ~ &= ~ h_p ~ (\id_M \times p)^\ast\, \widetilde{g} ~,
\end{flalign}
\end{subequations}
for all points $p: \bbR^0 \to U$. 
Note that, loosely speaking, $\widetilde{\GCon}(M)(U)$ contains ``too many objects''. 
More precisely, it is {\em not} true that for each object $(A, \{h_p\})$ there exists 
a morphism $(\{g_p\}, \widetilde{g})$ to an object of the form 
$(A^\prime, \{e\})$, where $e$ is the constant map to the identity of $G$. 
In fact, existence of such morphism would imply the identities
\begin{flalign}
g_p = h_p ~ (\id_M \times p)^\ast \, \widetilde{g}~,\qquad
(\id_M \times p)^\ast A^\prime = ((\id_M \times p)^\ast A) \ra g_p ~,
\end{flalign}
for all points  $p: \bbR^0 \to U$, which in general cannot be
satisfied because $\{h_p\}$ is a generic family, 
while both $A^\prime$ and $\widetilde{g}$ must be smoothly parametrized.
\sk

Our proposal is to fix this issue as follows: Consider the canonical $\HH$-morphism 
${\BGcon}_{~}^{\underline{M}} \to \widetilde{\GCon}(M)$ to the homotopy fiber product of \eqref{tmp:hofibprodapp}. 
Explicitly, it projects $1$-forms 
on $M \times U$ onto their vertical parts on $M \times U \to U$ (and attaches the family $\{e\}$).
Computing the fibrant replacement of ${\BGcon}_{~}^{\underline{M}} \to \widetilde{\GCon}(M)$ 
in the $(-1)$-truncation of $\HH/\widetilde{\GCon}(M)$ 
(cf.\ Proposition  \ref{propo:easyfibrepH/K} for a convenient weakly equivalent model) 
then yields a factorization
\begin{flalign}
\xymatrix{
{\BGcon}_{~}^{\underline{M}} \ar[r] ~&~ \GCon(M) \ar[r] ~&~ \widetilde{\GCon}(M)
}~,
\end{flalign}
where $\GCon(M)$  correctly describes
smoothly parametrized gauge fields and gauge transformations. 
Explicitly, at stage $U$ in $\CC$, the objects of $\GCon(M)(U)$ are
all $A \in \Omega^{1,0}(M \times U,\g)$ 
and the morphisms from $A$ to $A^\prime$ are specified by $g \in C^\infty(M \times U,G)$,
such that $A^\prime =  A \ra^{\ver} g$.
Here $A \ra^{\ver} g:= g^{-1}\, A\, g + g^{-1}\, \dd^{\ver} g$ 
is the vertical action of gauge transformations that is
defined by the vertical de Rham differential $\dd^{\ver}$ on $M \times U \to U$. 

\begin{rem}
In the main body of this paper, starting from Section \ref{subsec:concrete}, 
we use a simplified, but equivalent, version of our concretification prescription proposed above 
(cf.\ Definition \ref{defi:diffconcret}). 
This is based on the observation that one may skip $\sharp_1 ({\BGcon}^{\underline{M}})$
and the homotopy fiber product of \eqref{tmp:hofibprodapp}. 
(The only reason why we introduced
$\sharp_1 ({\BGcon}^{\underline{M}})$ here is
to compare with \cite{Fiorenza, Schreiber}.) Our simplified 
construction goes as follows: Instead of \eqref{tmp:hofibprodapp}, 
consider the pullback diagram 
\begin{flalign}
\xymatrix{
\sharp \big({\BGcon}_{~}^{\underline{M}}\big) \ar[rr]^-{\sharp (\mathrm{forget}^{\underline{M}})} 
~&&~ \sharp \big({\BG}_{~}^{\underline{M}}\big)~ &&~\ar[ll]_-{\zeta} {\BG}_{~}^{\underline{M}}
}
\end{flalign}
in $\HH$, whose homotopy fiber product is weakly equivalent in $\HH$ to the ordinary one (because
the right-pointing morphism is a fibration). Denoting the 
fiber product by $P$, the fibrant replacement of the canonical 
$\HH$-morphism ${\BGcon}_{~}^{\underline{M}} \to P$ 
in the $(-1)$-truncation of $\HH/P$ 
(use again Proposition  \ref{propo:easyfibrepH/K} for a convenient weakly equivalent model)
yields a factorization
\begin{flalign}
\xymatrix{
{\BGcon}_{~}^{\underline{M}} \ar[r] ~&~ \GCon(M) \ar[r] ~&~ P
}~,
\end{flalign}
where $\GCon(M)$ is the concretified mapping stack that was found above.
\end{rem}


\end{document}